\documentclass[a4paper,10pt,numbers=noenddot,twoside=on,headinclude=true,footinclude=false,headsepline=true]{scrartcl}

\usepackage[utf8]{inputenc}
\usepackage[T1]{fontenc}
\usepackage{textcomp}
\usepackage[english]{babel}
\usepackage{color}
\usepackage{amsmath}
\usepackage{amsopn}
\usepackage{amsbsy}
\usepackage[shortlabels]{enumitem}
\usepackage[plainpages=false,pdfpagelabels,breaklinks=true,pdfborderstyle={/S/U/W 1.2}]{hyperref}
\usepackage{graphics}
\usepackage{units}
\usepackage[cal=boondoxo,bb=boondox]{mathalfa}
\usepackage{dsfont}
\usepackage{csquotes}
\usepackage{import, xifthen, pdfpages, transparent}
\usepackage{comment}

\numberwithin{equation}{section} 
\numberwithin{table}{section} 
\numberwithin{figure}{section} 

\setlist[enumerate,1]{leftmargin=2.00em,itemsep=1pt,listparindent=15pt} 
\setlist[enumerate,2]{leftmargin=2.00em,itemsep=1pt,listparindent=15pt} 

\usepackage[amsmath,thmmarks,thref,hyperref]{ntheorem}

\theoremstyle{plain}
\newtheorem{theorem}{Theorem}[section]
\newtheorem{definition}[theorem]{Definition}
\newtheorem{lemma}[theorem]{Lemma}
\newtheorem{corollary}[theorem]{Corollary}

\newtheorem{proposition}[theorem]{Proposition}
\newtheorem{assumption}[theorem]{Assumption}

\theorembodyfont{\upshape}
\newtheorem{remark}[theorem]{Remark}

\theoremstyle{nonumberplain}

\theoremsymbol{\ensuremath{_\Box}}
\newtheorem{proof}{Proof}
\usepackage[scaled]{beramono}
\usepackage{helvet}
\usepackage[charter]{mathdesign} 
\SetMathAlphabet{\mathcal}{normal}{OMS}{cmsy}{m}{n} 
\SetMathAlphabet{\mathcal}{bold}{OMS}{cmsy}{m}{n} 
\providecommand{\ie}{i.~e.~}
\providecommand{\eg}{e.~g.~}
\providecommand{\cf}{cf.~}

\providecommand{\R}{\mathbb{R}}

\providecommand{\C}{\mathbb{C}}

\renewcommand{\C}{\mathbb{C}}

\providecommand{\N}{\mathbb{N}}
\providecommand{\Z}{\mathbb{Z}}
\providecommand{\ii}{\mathrm{i}}
\providecommand{\e}{\mathrm{e}}

\providecommand{\Hil}{\mathcal{H}}

\providecommand{\eps}{\varepsilon}
\providecommand{\Cont}{\mathcal{C}}

\providecommand{\dd}{\mathrm{d}}

\providecommand{\order}{\mathcal{O}}
\providecommand{\Fourier}{\mathcal{F}}

\providecommand{\abs}[1]{\left \lvert #1 \right \rvert}
\providecommand{\sabs}[1]{\lvert #1 \vert}
\providecommand{\babs}[1]{\bigl \lvert #1 \bigr \rvert}
\providecommand{\Babs}[1]{\Bigl \lvert #1 \Bigr \rvert}
\providecommand{\norm}[1]{\left \lVert #1 \right \rVert}
\providecommand{\snorm}[1]{\lVert #1 \rVert}
\providecommand{\bnorm}[1]{\bigl \lVert #1 \bigr \rVert}
\providecommand{\Bnorm}[1]{\Bigl \lVert #1 \Bigr \rVert}
\providecommand{\scpro}[2]{\left \langle #1 , #2 \right \rangle}
\providecommand{\sscpro}[2]{\langle #1 , #2 \rangle}
\providecommand{\bscpro}[2]{\bigl \langle #1 , #2 \bigr \rangle}

\providecommand{\sexpval}[1]{\langle #1 \rangle}
\providecommand{\bexpval}[1]{\bigl \langle #1 \bigr \rangle}

\providecommand{\weyl}{\star}
\providecommand{\semisuper}{\bullet}
\providecommand{\super}{\sharp}
\providecommand{\kerprod}{\diamond}
\providecommand{\op}{\mathrm{op}}
\providecommand{\Op}{\mathrm{Op}}
\providecommand{\Int}{\mathrm{Int}}
\providecommand{\smoothpoly}{C_{\textup{pol}}^\infty}

\providecommand{\Schwartz}{\mathcal{S}}
\providecommand{\bounded}{\mathcal{B}}

\providecommand{\triflux}{\gamma^B}

\providecommand{\cocyp}{\omega^B} 
\providecommand{\trifluxp}{\gamma^B}
\providecommand{\recfluxp}{\Omega^B}
\providecommand{\wigner}{\mathfrak{w}}
\providecommand{\Wigner}{\mathcal{W}}

\providecommand{\ssMoyalSpace}{\mathfrak{m}^B(\Pspace)}
\providecommand{\MoyalAlgebra}{\mathcal{M}}
\providecommand{\sMoyalAlg}{\MoyalAlgebra^B(\Pspace)}

\providecommand{\Alg}{\mathcal{A}}
\providecommand{\vNAlg}{\mathfrak{A}}
\providecommand{\pspace}{\Xi}
\providecommand{\Pspace}{\Xi^2}

\providecommand{\Ad}{\mathrm{Ad}}

\providecommand{\Xbf}{\mathbf{X}}
\providecommand{\Ybf}{\mathbf{Y}}
\providecommand{\Zbf}{\mathbf{Z}}

\providecommand{\ssMoyalSpace}{\mathfrak{m}^B(\Pspace)}

\title{A Calculus for Magnetic \\ Pseudodifferential Super Operators}
\author{Gihyun Lee${}^1$ \& Max Lein${}^2$}

\begin{document}

\maketitle
\vspace{-9mm}
\begin{center}
	${}^1$ Max-Planck-Institut f\"ur Mathematik \linebreak
	Vivatsgasse 7, 53111 Bonn, Germany \linebreak
	{\footnotesize \href{mailto:gihyun@mpim-bonn.mpg.de}{\texttt{gihyun@mpim-bonn.mpg.de}}}
	\medskip
	\\
	${}^2$ Advanced Institute of Materials Research,  
	Tohoku University \linebreak
	2-1-1 Katahira, Aoba-ku, 
	Sendai, 980-8577, 
	Japan \linebreak
	{\footnotesize \href{mailto:max.lein@tohoku.ac.jp}{\texttt{max.lein@tohoku.ac.jp}}}
\end{center}
\begin{abstract}
	This work develops a magnetic pseudodifferential calculus for \emph{super} operators $\Op^A(F)$; these map operators onto operators (as opposed to $L^p$ functions onto $L^q$ functions). Here, $F$ could be a tempered distribution or a Hörmander symbol. An important example is Liouville super operators $\hat{L} = - \ii \, \bigl [ \op^A(h) , \, \cdot \, \bigr ]$ defined in terms of a magnetic pseudodifferential operator $\op^A(h)$. Our work combines ideas from magnetic Weyl calculus developed in  \cite{Mantoiu_Purice:magnetic_Weyl_calculus:2004,Iftimie_Mantoiu_Purice:magnetic_psido:2006,Lein:progress_magWQ:2010} and the pseudodifferential calculus on the non-commutative torus from \cite{Ha_Lee_Ponge:pseudodifferential_theory_on_noncommutative_torus_1:2018,Ha_Lee_Ponge:pseudodifferential_theory_on_noncommutative_torus_2:2018}. Thus, our calculus is inherently gauge-covariant, which means all essential properties of $\Op^A(F)$ are determined by properties of the magnetic field $B = \dd A$ rather than the vector potential $A$. 
	
	There are conceptual differences to ordinary pseudodifferential theory. For example, in addition to an analog of the (magnetic) Weyl product that emulates the composition of two magnetic pseudodifferential super operators on the level of functions, the so-called semi-super product describes the action of a pseudodifferential \emph{super} operator on a pseudodifferential operator. 
\end{abstract}

\newpage
\tableofcontents

\section{Introduction} 
\label{intro}
Pseudodifferential theory and microlocal analysis have come a long way from their infancy in the 1960s and 1970s \cite{Hoermander:Fourier_integral_operators_1:1971,Hoermander:Fourier_integral_operators_2:1973,Hoermander:Weyl_calculus:1979,Grigis_Sjoestrand:microlocal_analysis:1994,Martinez:intro_microlocal_analysis:2002}, and has since become a robust tool in analysis. Its origins can be traced back to the seminal works of Weyl \cite{Weyl:qm_gruppentheorie:1927}, Wigner \cite{Wigner:Wigner_transform:1932} and Moyal \cite{Moyal:Weyl_calculus:1949}.

\subsection{The basic ideas of pseudodifferential calculi} 
\label{intro:pseudodifferential_calculus}
In its most basic form, pseudodifferential theory assigns operators 
\begin{align*}
	\op(f) : \mathcal{D} \subseteq L^p(\R^d) \longrightarrow L^q(\R^d)
	, 
	&&
	1 \leq p , q \leq \infty 
	, 
\end{align*}
between $L^p$ and Sobolev spaces to suitable functions $f : \R^{2d} \longrightarrow \C$ via the quantization map $\op$. Importantly, we can infer properties of the \emph{operator} $\op(f)$ from properties of the \emph{function} $f$. 

What is more, the map $\op$, also referred to as \emph{Weyl quantization} in this context, gives rise to a \emph{calculus:} we can pull back the operator adjoint $\op(f)^* = \op(\bar{f})$, which translates to complex conjugation, and the operator product of two pseudodifferential operators 
\begin{align*}
	\op(f) \, \op(g) =: \op(f \weyl g) 
\end{align*}
to the level of functions. This foreshadows that many classes of functions — or, equivalently, of pseudodifferential operators — form $\ast$-algebras. More advanced results like Beals- and Bony-type commutator criteria \cite{Beals:characterization_psido:1977,Bony:characterization_psido:1996,Iftimie_Mantoiu_Purice:commutator_criteria:2008,Cornean_Helffer_Purice:simple_proof_Beals_criterion_magnetic_PsiDOs:2018} stipulate conditions when an operator $\hat{f} = \op(f)$ is a pseudodifferential operator and characterize the function $f$. Among other things, they show that resolvents of certain pseudodifferential operators are pseudodifferential operators; this gives rises to a functional calculus via the Helffer-Sjöstrand formula \cite{Helffer_Sjoestrand:mag_Schroedinger_equation:1989}. 

In applications, such a functional calculus is tremendously useful since it allows us to easily construct pseudodifferential operators \emph{with known properties} from other pseudodifferential operators. One relevant case is density operators 
\begin{align}
	\hat{\rho} = Z_T^{-1} \, f_T \bigl ( \op(h) \bigr ) = \op(\rho_T)
	\label{intro:eqn:thermal_state_Fermi_Dirac}
\end{align}
that model states of condensed matter systems at finite temperature $T > 0$. These are defined via functional calculus for the selfadjoint pseudodifferential operator $\op(h) = \op(h)^*$ with respect to the Fermi-Dirac distribution 
\begin{align*}
	f_T(E) &= \frac{1}{1 + \e^{\frac{1}{k_{\mathrm{B}} T}(E - \mu)}} 
	. 
\end{align*}
Here, the partition sum $Z_T > 0$ is a normalization constant, $k_{\mathrm{B}}$ is the Boltzmann constant and the parameter $\mu \in \R$ is the chemical potential or Fermi energy. 

Continuous refinements over the years have led to pseudodifferential calculi that are adapted to specific situations, including spin systems \cite{Stratonovich:distributions_rep_space:1957,Gat_Lein_Teufel:semiclassical_dynamics_particle_spin:2013}. Given the importance magnetic fields have played in many recent works, we will base this work on a variant adapted to systems with magnetic fields called \emph{magnetic Weyl calculus} \cite{Mueller:product_rule_gauge_invariant_Weyl_symbols:1999,Mantoiu_Purice:magnetic_Weyl_calculus:2004,Iftimie_Mantoiu_Purice:magnetic_psido:2006,Iftimie_Mantoiu_Purice:commutator_criteria:2008,Iftimie_Purice:magnetic_Fourier_integral_operators:2011,Lein:two_parameter_asymptotics:2008,Lein:progress_magWQ:2010}; we will outline the definitions with more references to the pertinent literature in Section~\ref{magnetic_Weyl_calculus}. 

\subsection{Non-commutative geometry and its relation to pseudodifferential theory} 
\label{intro:non_commutative_geometry}
Another separate, but equally influential development was Alain Connes's non-com\-mutative geometry \cite{Connes:noncommutative_geometry:1994}. It has proven invaluable in many areas of mathematical physics, including condensed matter physics. This approach exploits that if the relevant operators are elements of or affiliated to a $C^*$- or von Neumann algebra $\vNAlg$ \cite{Dixmier:C_star_algebras:1977,Dixmier:von_Neumann_algebras:1981}. Properties of, say, the hamiltonian $\hat{h}$ are encoded in the algebraic structure of $\vNAlg$, and we can exploit that then many operations preserve this algebraic structure. For example, if $\hat{h}$ is affiliated to $\vNAlg$, then the finite-temperature state $\hat{\rho}_T = Z_T^{-1} \, f_T(\hat{h}) \in \vNAlg$ lies inside the algebra. 

Among other things, they have been used to great success to understand topological phenomena from physics such as the Quantum Hall Effect \cite{Bellissard:quantum_Hall_effect_noncommutative:1984,Bellissard:K_theory_Cast_algebras_solid_state_physics:1986,Bellissard:Cstar_algebras_solid_state_physics:1988,Bellissard_van_Elst_Schulz_Baldes:noncommutative_geometry_quantum_hall_effect:1994,Bellissard_Schulz_Baldes:kinetic_theory_quantum_transport_aperiodic_media:1998,Prodan:non_commutative_geometry_topological_insulators:2014}. Very naturally this approach allows one to include effects of disorder, which is one reason why it has proven very popular. 

As the name suggests one of the driving forces behind non-commutative geometry was to extend concepts from analysis and differential geometry like derivatives to operator algebras. So it comes as no surprise that the Venn diagram of pseudodifferential theory and non-commutative geometry is not empty. Indeed, Baaj made the connection between pseudodifferential operators and crossed product $C^*$-algebras explicit \cite{Connes:Cstar_algebras_and_differential_geometry:1980,Baaj:twisted_X_products_1:1988,Baaj:twisted_X_products_2:1988}; this was later extended to \emph{magnetic} pseudodifferential operators that are affiliated to \emph{twisted} crossed product $C^*$-algebras \cite{Mantoiu_Purice_Richard:twisted_X_products:2004}. That connection has been used to great effect to \eg derive propagation estimates \cite{Mantoiu_Purice_Richard:Cstar_algebraic_framework:2007}, prove continuity of spectra \cite{Mantoiu_Purice:continuity_spectra:2009}, to characterize the essential spectrum of magnetic pseudodifferential operators \cite{Lein_Mantoiu_Richard:anisotropic_mag_pseudo:2009} and to study semiclassical limits \cite{Belmonte_Lein_Mantoiu:mag_twisted_actions:2010}. These works show how ideas and methods from very different fields of mathematics can be combined to one's advantage. 

\subsection{Super operators and their applications to physics} 
\label{intro:super_operators}
One way to describe a quantum mechanical system is to propose that the \emph{density operator} $\hat{\rho}(t)$ which describes the state of the system evolves according to the Liouville equation 
\begin{align*}
	\frac{\dd}{\dd t} \hat{\rho}(t) &= \hat{L} \bigl ( \hat{\rho}(t) \bigr )
	, 
	&&
	\hat{\rho}(t_0) \in \mathfrak{L}^1(\vNAlg)
	. 
\end{align*}
The \emph{super} operator $\hat{L}$ is often a Liouville \emph{super} operator 
\begin{align}
	\hat{L}(\hat{\rho}) = - \ii \, \bigl [ \hat{h} \, , \, \hat{\rho} \bigr ] 
	\label{intro:eqn:Liouville_super_operator}
\end{align}
defined with respect to some selfadjoint operator $\hat{h} = \hat{h}^*$ affiliated to the algebra $\vNAlg$. The qualifier \emph{super} operator emphasizes that $\hat{L}$ maps operators onto operators rather than $L^p$ functions onto $L^q$ functions, and we will use this delineation throughout this manuscript to help the reader follow along. Of course, the commutator on the right of \eqref{intro:eqn:Liouville_super_operator} needs to be carefully defined in case $\hat{h}$ is an \emph{unbounded} operator \cite[Chapter~3.3]{DeNittis_Lein:linear_response_theory:2017}. More generally, $\hat{L}$ is a Lindblad-type super operator \cite{Lindblad:generator_semigroups:1976} that contains additional terms to account for dissipation and thermalization mechanisms; a discussion of the physical processes these extra terms model can be found in  \cite[Section~IV.A]{Bellissard_van_Elst_Schulz_Baldes:noncommutative_geometry_quantum_hall_effect:1994} and \cite[Section~2.2]{Bellissard_Schulz_Baldes:kinetic_theory_quantum_transport_aperiodic_media:1998}. 

Density operators by definition satisfy $\hat{\rho}(t) = \hat{\rho}(t)^* \geq 0$ and have trace $1 = \mathcal{T} \bigl ( \hat{\rho}(t) \bigr )$. Thus, the Liouville equation is naturally studied on the space of trace class operators. When developing the general theory though, it is necessary to consider $\mathfrak{L}^p(\vNAlg)$ for $p \neq 1$ as well. Importantly, $\mathcal{T}$ need not be the canonical trace on a Hilbert space and could instead be \eg the trace-per-unit-volume \cite{Lenz:random_operators_crossed_products:1999}. Consequently, it does not suffice to consider the regular $1$-Schatten class and need to work with non-commutative space $\mathfrak{L}^1(\vNAlg)$; for more exposition and precise definitions we refer to \cite[Chapter~3]{DeNittis_Lein:linear_response_theory:2017}. The algebraic point of view has proven very successful to understand systems from statistical mechanics in the thermodynamic limit, where the number of constituents tends to $\infty$ \cite{Bratteli_Robinson:operator_algebras_1:2002,Bratteli_Robinson:operator_algebras_2:2003}. 

\subsection{Motivation for developing a pseudodifferential theory for super operators} 
\label{intro:our_motivation}
Ultimately, we want to work at the intersection of the Venn diagram where $\hat{h} = \op(h)$ that defines the Liouville \emph{super} operator~\eqref{intro:eqn:Liouville_super_operator} is a pseudodifferential operator \emph{and} affiliated to some algebra $\vNAlg$. That suggests to ask: is the Liouville super operator in that case then a pseudodifferential \emph{super} operator? The answer is yes, namely that $\hat{L}$ is defined from the function 
\begin{align*}
	L(X_L,X_R) = - \ii \, \bigl ( h(X_L) - h(X_R) \bigr )
	, 
	&&
	X_L , X_R \in \R^{2d} 
	. 
\end{align*}
The purpose of this paper is to make that statement precise by developing a calculus for magnetic pseudodifferential super operators; for an overview of the calculus and precise equations we refer to Section~\ref{formal_super_calculus}. 

Before we embark on this endeavor, though, let us ask what such a pseudodifferential super calculus might buy us. Recently, one of the authors of this work developed a framework for linear response theory on the level of non-commutative $\mathfrak{L}^p$-spaces \cite{DeNittis_Lein:linear_response_theory:2017}, which generalizes results by Bouclet, Germinet, Klein and Schenker \cite{Bouclet_Germinet_Klein_Schenker:linear_response_theory_magnetic_Schroedinger_operators_disorder:2005} as well as Dombrowski and Germinet \cite{Dombrowski_Germinet:linear_response_theory:2008}. All of the advantages of an algebraic approach are on full display: the theory is able to make linear response theory rigorous for a wide array of systems, be it on the continuum or discrete, and can naturally include effects of disorder. However, the price to pay are a set of rather technical assumptions, Hypotheses~5–6 from \cite[Chapter~2]{DeNittis_Lein:linear_response_theory:2017} in particular, which seem difficult to verify for concrete models. Roughly speaking, these conditions are necessary to ensure that products and commutators are well-defined, and that all expressions have the right regularity properties. 

Should all of the operators involved be pseudodifferential, we then expect that products and commutators yield well-defined pseudodifferential operators. For example, this should be the case when the Liouville super operator~\eqref{intro:eqn:Liouville_super_operator} is defined from $\hat{h} = \op(h)$ and $\hat{\rho}_0 = \op(\rho_T)$ the thermal state~\eqref{intro:eqn:thermal_state_Fermi_Dirac}. A well-developed pseudodifferential calculus would allow us to infer properties of products and commutators from properties of the input. It therefore stands to reason that this would simplify Hypotheses~5–6 from \cite[Chapter~2]{DeNittis_Lein:linear_response_theory:2017} and dispense with many of the technical complications in the proofs. 

\subsection{The state-of-the-art} 
\label{intro:literature}
Fortunately, the other author together with Ha and Ponge has developed such a calculus for the special case of the non-commutative torus \cite{Ha_Lee_Ponge:pseudodifferential_theory_on_noncommutative_torus_1:2018,Ha_Lee_Ponge:pseudodifferential_theory_on_noncommutative_torus_2:2018}. Because the non-commutative torus can be regarded as (a representation of) the twisted crossed product $\C \rtimes^B \Z^d$, the present work aims to generalize their construction to twisted crossed products of the form $\Alg \rtimes^{\omega} \R^d$. The important difference between the two is that the former applies to a special case of discrete systems whereas we wish to cover (super) operators on the continuum. 

Apart from \cite{Ha_Lee_Ponge:pseudodifferential_theory_on_noncommutative_torus_1:2018,Ha_Lee_Ponge:pseudodifferential_theory_on_noncommutative_torus_2:2018}, we are only aware of one \emph{non-rigorous} work on this subject by Tarasov \cite{Tarasov:quantum_mechanics_non-hamiltonian_dissipative_systems:2008}. He has only defined the super Weyl quantization for non-magnetic systems (\cf \cite[Chapter~12.8, equation~(15)]{Tarasov:quantum_mechanics_non-hamiltonian_dissipative_systems:2008}), but has \emph{not} developed a full calculus with products and dequantization. Moreover, even though his work gives definitions and stipulates theorems, which makes it easy to read for mathematicians, most proofs are missing. So to the best of our knowledge, our results are original. 

\subsection{Future developments} 
\label{sub:future_developments}
In spirit, this paper should be thought of as being the equivalent of \cite{Mantoiu_Purice:magnetic_Weyl_calculus:2004} and parts of \cite{Iftimie_Mantoiu_Purice:magnetic_psido:2006}: we merely set up the calculus in its most bare-bones form and cover the composition properties of Hörmander symbols. 

Quite naturally, a topic that was originally meant to be part of this work was the asymptotic expansion of the semi-super and the super Weyl products as well as a discussion of the semiclassical limit akin to \cite{Lein:two_parameter_asymptotics:2008} and \cite[Chapter~3]{Lein:progress_magWQ:2010}. We had to cut this part for length and instead plan to publish it as a separate work. Of course, this should not merely recast the semiclassical limit from \cite[Theorem~3.6.1]{Lein:progress_magWQ:2010} in a different language, but extend it to Lindblad-type operators. That should explain how Lindblad-type dynamics are well-approximated by a classical dissipative dynamical system. 

Boundedness results are another avenue. In standard pseudodifferential theory analytic and algebraic results are well-separated from one another. Interestingly, though, for magnetic pseudodifferential \emph{super} operators the two are strongly linked and it is strictly necessary to develop a better understanding of the operator algebraic aspects \emph{first} before proving analytic results. Say we would like to establish the boundedness of the super operator 
\begin{align*}
	\Op^A(F) : \mathfrak{L}^p(\vNAlg) \longrightarrow \mathfrak{L}^q(\vNAlg)
	, 
	&&
	1 \leq p , q \leq \infty 
	, 
\end{align*}
once we place suitable conditions on the function $F$. However, this requires a better understanding of the non-commutative spaces $\mathfrak{L}^p(\vNAlg)$. At first glance, it might be tempting to simply pick $\vNAlg = \mathcal{B} \bigl ( L^2(\R^d) \bigr )$ and choose the canonical trace $\mathrm{Tr}_{L^2(\R^d)}$; then the non-commutative $L^p$ spaces are just the usual $p$-Schatten class ideals. However, in \eg periodic and random systems one usually wants to define “integration” with respect to the trace-per-unit-volume, \ie we would like to choose a different “measure” with respect to which we define “integration”. Moreover, the trace needs to be compatible with the algebra, so $\mathcal{B} \bigl ( L^2(\R^d) \bigr )$ might be too big for our purposes to be useful. In the context of pseudodifferential theory the relevant $C^*$- and von Neumann algebras are twisted crossed product algebras (\cf \cite{Mantoiu_Purice_Richard:twisted_X_products:2004}) and their bicommutants after representation through super Weyl quantization $\Op^A$. 

That would set the stage for more advanced analytic results such as commutator criteria akin to those by Beals and Bony \cite{Beals:characterization_psido:1977,Bony:characterization_psido:1996,Cornean_Helffer_Purice:simple_proof_Beals_criterion_magnetic_PsiDOs:2018} as well as a functional calculus similar to that set up in \cite[Section~6.2]{Iftimie_Mantoiu_Purice:commutator_criteria:2008}. 

\subsection{Outline} 
\label{intro:outline}
Section~\ref{magnetic_Weyl_calculus} summarizes magnetic Weyl quantization. While introducing basic notions and notation in a pedagogical fashion, we will also emphasize the structure of a pseudodifferential calculus and how to develop it systematically and rigorously. Then Section~\ref{formal_super_calculus} outlines our magnetic pseudodifferential super calculus. We give the main notions, explicit formulas and indicate how to formulate them rigorously. Our goal was to give readers not yet familiar with pseudodifferential theory an easy-to-follow introduction. And it might suffice to explain our results to experts in pseudodifferential theory, since the proofs rely only on standard techniques. Indeed, Sections~\ref{rigorous_definition_supercalculus_on_S} and \ref{super_calculus_extension_by_duality} set up the calculus for Schwartz functions and extend it by duality. This is the basis for establishing a calculus for Hörmander symbol classes in Section~\ref{symbol_super_calculus} with oscillatory integral techniques. Lastly, we have relegated some technical results to an appendix. 

\subsection*{Acknowledgements} 
\label{sub:acknowledgements}
G.~L.~acknowledges the support of BK21 PLUS SNU, Mathematical Sciences Division (South Korea) and Max Planck Institute for Mathematics in Bonn (Germany). The initial part of the research of this manuscript was carried out during the first named author's visit to Advanced Institute for Materials Research of Tohoku University (Japan), and G.L. would like to thank them for their hospitality. 

M.~L.~thanks JSPS for their support of this project through a Wakate~B (No.~16K17761) and a Kiban~C grant (No.~20K03761). 

The authors would like to thank Elmar Schrohe for making us aware of prior use of what we called Hörmander super symbol classes in the literature. 
\section{Magnetic Weyl calculus} 
\label{magnetic_Weyl_calculus}
The origins of modern pseudodifferential theory go back to the works of Weyl \cite{Weyl:qm_gruppentheorie:1927}, Wigner \cite{Wigner:Wigner_transform:1932} and Moyal \cite{Moyal:Weyl_calculus:1949}. Their goal was to make mathematical sense of Dirac's quantization procedure \cite{Dirac:foundations_qm_en:1947} and solve the operator ordering problem: in its simplest form, the question is what operator we should associate to classical observables (that is, functions on phase space) such as $f(x,\xi) = x \cdot \xi$? The question is non-trivial to solve even formally, because position and momentum operators no longer commute, 
\begin{align*}
	\ii \, [ Q_j , Q_k ] = 0 
	, 
	&&
	\ii \, [ P_j , P_k ] = 0 
	, 
	&&
	\ii \, [ P_j , Q_k ] = \delta_{jk}
	. 
\end{align*}
So formally at least, $Q \cdot P$, $P \cdot Q$ and $\tfrac{1}{2} (Q \cdot P + P \cdot Q)$ are all possible solutions. 

Mathematically, we can rephrase the operator ordering problem as looking for a functional calculus to a family of non-commuting operators. One solution to this is to map suitable functions $f : \R^{2d} \longrightarrow \C$ to operators 
\begin{align}
	f(Q,P) \equiv \op(f)
	= \frac{1}{(2\pi)^d} \int_{\R^{2d}} \dd x \, \dd \xi \, (\Fourier_{\sigma} f)(x,\xi) \, \e^{+ \ii (x \cdot P - \xi \cdot Q)}
	\label{magnetic_Weyl_calculus:eqn:definition_non_magnetic_op}
\end{align}
by formally applying the \emph{symplectic Fourier transform} 
\begin{align}
	(\Fourier_{\sigma} f)(x,\xi) := \frac{1}{(2\pi)^d} \int_{\R^{2d}} \dd x' \, \dd \xi' \, \e^{+ \ii (\xi \cdot x' - x \cdot \xi')} \, f(x',\xi') 
	, 
	\label{magnetic_Weyl_calculus:eqn:symplectic_Fourier_transform}
\end{align}
twice, thereby “replacing” the position variable $x$ with the position operator $Q$ and $\xi$ with the momentum operator $P$. The reason why we prefer to use the symplectic Fourier transform rather than the ordinary one is that $\Fourier_{\sigma} = \Fourier_{\sigma}^{-1}$ is its own inverse. This simplifies many formulæ since we do not need to keep track of signs in the exponent of the phases as much. The index $\sigma$ stems from the abbreviation we will use for the phase factor, namely 
\begin{align*}
	\sigma \big( (x,\xi) , (x',\xi') \big) := \xi \cdot x' - x \cdot \xi' 
	. 
\end{align*}
Of course, \emph{a priori} equation~\eqref{magnetic_Weyl_calculus:eqn:definition_non_magnetic_op} is just a formal expression that needs to be made mathematical sense of. But this has become standard fare. Hörmander's early works on that topic are still worth reading \cite{Hoermander:Fourier_integral_operators_1:1971,Hoermander:Weyl_calculus:1979}, although by now several text books have been dedicated to this topic \cite{Folland:harmonic_analysis_hase_space:1989,Robert:tour_semiclassique:1987,Taylor:PsiDO:1981,Kumanogo:pseudodiff:1981}; an entry-level introduction that is more approachable for students is given in the lecture notes \cite{Lein:quantization_semiclassics:2010}. 

In a nutshell, the standard approach consists of three steps: (1)~we first define it for Schwartz functions; (2)~we extend the expression by duality to tempered distributions with nice composition properties; and (3)~we show that certain classes of functions, most notably Hörmander symbols, are among those with nice composition properties. Since we will implement this procedure for the super calculus, we content ourselves pointing the interested reader to the excellent exposition in \cite[Sections~II]{Mantoiu_Purice:magnetic_Weyl_calculus:2004}; more precisely, Sections~\ref{rigorous_definition_supercalculus_on_S}, \ref{super_calculus_extension_by_duality} and \ref{symbol_super_calculus} correspond to each of these steps. 

Adapted pseudodifferential calculi to certain scenarios such as spin systems \cite{Stratonovich:distributions_rep_space:1957,Gat_Lein_Teufel:resonance_phenomena_wavepacket_qubit:2013} have also been investigated. Of particular relevance for this work is \emph{magnetic Weyl calculus} for magnetic systems, first proposed non-rigorously by Müller \cite{Mueller:product_rule_gauge_invariant_Weyl_symbols:1999}. His ideas were made rigorous independently by the works of M\u{a}ntoiu, Purice and co-workers \cite{Mantoiu_Purice:magnetic_Weyl_calculus:2004,Iftimie_Mantoiu_Purice:magnetic_psido:2006,Iftimie_Mantoiu_Purice:commutator_criteria:2008}, including several contributions by one of the authors \cite{Lein:two_parameter_asymptotics:2008,Lein_Mantoiu_Richard:anisotropic_mag_pseudo:2009,Lein:progress_magWQ:2010,DeNittis_Lein:Bloch_electron:2009,Fuerst_Lein:scaling_limits_Dirac:2008}. 

The purpose of this section is two-fold: on the one hand, we wish to introduce the basic tenets of magnetic pseudodifferential theory to fix basic notions and notation. On the other hand, we aim to emphasize the structural similarities between usual magnetic Weyl calculus and the magnetic super Weyl calculus we will develop in the coming sections.

\subsection{The magnetic Weyl system} 
\label{magnetic_Weyl_calculus:magnetic_weyl_system}
Suppose we want to study a charged quantum particle subjected to the magnetic field $B$. To represent this system on a Hilbert space, we need to pick a vector potential $A$ so that $\dd A = B$; our assumptions on $B$ guarantee that the components of $A$ are $\Cont^{\infty}_{\mathrm{pol}}$.\footnote{We will tacitly identify $1$-forms $A = \sum_{j = 1}^d A_j \, \dd x_j$ with vector fields and $2$-forms $B = \sum_{j , k = 1}^d B_{jk} \, \dd x_j \wedge \dd x_k$ with their antisymmetric matrix-valued functions $(B_{jk})_{1 \leq j , k \leq d}$ comprised of the coefficients. } 

For the purpose of this work, we shall always assume that the magnetic field and vector potential satisfy one of the following two assumptions:
\begin{assumption}[Polynomially bounded magnetic field]\label{intro:assumption:polynomially_bounded_magnetic_field}
	We say a magnetic field $B = \dd A$ and any associated vector potential $A$ satisfy the \emph{polynomially bounded magnetic field assumption} if and only if the components of $B$ are of class $\Cont^{\infty}_{\mathrm{u,pol}}(\R^d)$, \ie smooth, uniformly polynomially bounded functions. 
	
	Furthermore, we assume that we choose an associated vector potential whose components are of class $\Cont^{\infty}_{\mathrm{pol}}(\R^d)$, \ie smooth and polynomially bounded and satisfy $\dd A = B$. 
\end{assumption}
Note that it is always possible to choose a smooth, polynomially bounded vector potential when $B$ is of class $\Cont^{\infty}_{\mathrm{pol}}$. 

However, when we want to work with Hörmander symbols, we have to work with a more restricted class of electromagnetic fields, namely 
\begin{assumption}[Bounded magnetic field]\label{intro:assumption:bounded_magnetic_field}
	We say a magnetic field $B = \dd A$ and any associated vector potential $A$ satisfy the \emph{bounded magnetic field assumption} if and only if the components of $B$ are of class $\Cont^{\infty}_{\mathrm{b}}(\R^d)$, \ie bounded with bounded derivatives to any order. 
	
	Furthermore, we assume that we choose an associated vector potential whose components are of class $\Cont^{\infty}_{\mathrm{pol}}(\R^d)$, \ie smooth and polynomially bounded and satisfy $\dd A = B$. 
\end{assumption}
We will further introduce two small parameters as in \cite{Lein:two_parameter_asymptotics:2008}: the coupling of the charge to the magnetic field is quantified by the dimensionless parameter $\lambda$. Furthermore, we include a semiclassical parameter $\eps$ in our considerations. While for much of the general theory small parameters are irrelevant, they are crucial for asymptotic expansions that have proven useful in making semiclassical limits (\cf \cite[Théorème~IV.19]{Robert:tour_semiclassique:1987} or \cite[Theorem~3.6.1]{Lein:progress_magWQ:2010}) and perturbation expansions \cite{PST:sapt:2002,PST:effective_dynamics_Bloch:2003,PST:Born-Oppenheimer:2007,DeNittis_Lein:Bloch_electron:2009,Gat_Lein_Teufel:semiclassical_dynamics_particle_spin:2013,Fuerst_Lein:scaling_limits_Dirac:2008} mathematically rigorous. Should these parameters not be needed, then we can still set their values to $1$; this choice simplifies some of our estimates and paves the way for future works. 

The selfadjoint position operators $Q_j = Q_j^*$ and momentum operators $P^A_j = {P^A_j}^*$ for magnetic Weyl calculus are characterized by the commutation relations 
\begin{align}
	\ii \, [ Q_j , Q_k ] = 0 
	, 
	&&
	\ii \, [ P^A_j , P^A_k ] = \eps \, \lambda \, B_{jk}(Q) 
	, 
	&&
	\ii \, [ P^A_j , Q_k ] = \eps \, \delta_{jk}
	. 
	\label{magnetic_Weyl_calculus:eqn:commutation_relations}
\end{align}
The difference to non-magnetic systems is that momenta along different directions no longer commute. Of course, $Q_j$ and $P^A_k$ are unbounded operators, so these commutators are initially just formal expressions. One way to make them rigorous is to initially define them on a joint core such as $\Cont^{\infty}_{\mathrm{c}}(\R^d)$ or $\Schwartz(\R^d)$. But the avenue we will opt for here is to formulate them in terms of the composition properties of the associated unitary evolution groups. 

We emphasize that the commutation relations~\eqref{magnetic_Weyl_calculus:eqn:commutation_relations} are invariant under changes of representation, \ie they capture algebraic relations amongst position and momentum operators. Importantly, the commutator of two magnetic momenta only depends on the magnetic \emph{field} $B$ rather than the vector potential $A$. 

In Schrödinger representation on the Hilbert space $L^2(\R^d)$ we define the position operators $Q = \eps \hat{x}$ for Schwartz functions $\psi \in \Schwartz(\R^d)$ as 
\begin{align}
	(Q_j \psi)(x) := \eps \, x_j \, \psi(x) 
	, 
	&&
	j = 1 , \ldots , d
	. 
	\label{magnetic_Weyl_calculus:eqn:position_operator_adiabatic_scaling}
\end{align}
To define the kinetic momentum operator on the dense domain $\Schwartz(\R^d) \subset L^2(\R^d)$, 
\begin{align}
	(P^A_j \psi)(x) := - \ii \partial_{x_j} \psi(x) - \lambda A_j(x) \, \psi(x) 
	, 
	&&
	j = 1 , \ldots , d 
	, 
	\label{magnetic_Weyl_calculus:eqn:momentum_operator_adiabatic_scaling}
\end{align}
we necessarily have to pick a magnetic vector \emph{potential} $A$ for the magnetic field $B = \dd A$. Keeping track of the choice of vector potential is of conceptual and practical importance. Therefore, we will make it explicit in our choice of Hilbert space 
\begin{align*}
	\Hil^A = L^2(\R^d)
	, 
\end{align*}
even thought the space itself evidently is independent of our choice of vector potential. 

For example, when we change gauge, \ie we pick some other vector potential $A' = A + \eps \dd \chi$ that differs from $A$ by the gradient of some real-valued function $\chi \in \Cont^{\infty}_{\mathrm{pol}}(\R^d,\R)$, the operator 
\begin{align*}
	\e^{+ \ii \lambda \chi(Q)} : \Hil^A \longrightarrow \Hil^{A'}
\end{align*}
defines a unitary between these two Hilbert spaces. While it maps the position operators 
\begin{align*}
	Q_j' := \e^{+ \ii \lambda \chi(Q)} \, Q_j \, \e^{- \ii \lambda \chi(Q)} 
	= Q_j 
\end{align*}
onto themselves, kinetic momenta change covariantly, 
\begin{align*}
	P^{A'}_j &= \e^{+ \ii \lambda \chi(Q)} \, P^A_j \, \e^{- \ii \lambda \chi(Q)} 
	. 
\end{align*}
This notation will also make even more sense later on when we introduce magnetic Weyl quantization as the representation map for certain algebras of functions or distributions onto operator algebras on the Hilbert spaces $\Hil^A$ (\cf Section~\ref{magnetic_Weyl_calculus:algebraic_point_of_view} and references therein). 

Similarly, we can view $Q$ and $P^A$ in other representations by \eg adjoining with the 
Fourier transform (to change into momentum representation). Or we could switch from microscopic to macroscopic units for measuring lengths (\cf \cite[Section~2.2]{DeNittis_Lein:Bloch_electron:2009}) so that $Q$ and $P^A$ are mapped to 
\begin{align*}
	q_j &= \hat{x}_j
	, 
	\\
	p_j^A &= - \ii \eps \partial_{x_j} - \lambda A_j(\hat{x})
	. 
\end{align*}
This scaling is perhaps more familiar to many, but since it is unitarily equivalent to \eqref{magnetic_Weyl_calculus:eqn:position_operator_adiabatic_scaling} and \eqref{magnetic_Weyl_calculus:eqn:momentum_operator_adiabatic_scaling}, it is ultimately a mere matter of preference. Indeed, many aspects of the calculus like the formula~\eqref{magnetic_Weyl_calculus:eqn:magnetic_Weyl_product_formula} for the magnetic Weyl product $\weyl^B$ are independent of the choice of representation. 

The best way to encode the commutation relations~\eqref{magnetic_Weyl_calculus:eqn:commutation_relations} rigorously is to look at the evolution groups. To do that effectively, we define the magnetic Weyl system 
\begin{align}
	w^A(x,\xi) \equiv w^A(X) := \e^{- \ii \sigma(X,(Q,P^A))}
	= \e^{+ \ii (x \cdot P^A - \xi \cdot Q)}
	, 
	&&
	X = (x,\xi) \in T^* \R^d := \pspace
	, 
	\label{magnetic_Weyl_calculus:eqn:definition_magnetic_Weyl_system}
\end{align}
which differs from the product of the evolution groups for position and momentum operators by a phase. 
\begin{lemma}[Fundamental properties of the magnetic Weyl system]\label{magnetic_Weyl_calculus:lem:properties_Weyl_system}
	Suppose the magnetic field $B = \dd A$ and the associated vector potential $A$ are polynomially bounded in the sense of Assumption~\ref{intro:assumption:polynomially_bounded_magnetic_field}. Then the following holds true: 
	\begin{enumerate}[(1)]
		\item The action of $w^A(Y)$ on $\varphi \in \Schwartz(\R^d) \subset L^2(\R^d)$ is given by 
		\begin{align*}
			\bigl ( w^A(Y) \varphi \bigr )(x) = \e^{- \ii \eps (x + \frac{y}{2}) \cdot \eta} \, \e^{- \ii \frac{\lambda}{\eps} \Gamma^A([\eps x,\eps x+\eps y])} \, \varphi(x+y)
		\end{align*}
		where $\Gamma^A(x,y) := \int_{[x,y]} A$ is the magnetic circulation of the one-form $A$ along the line segment $[x,y]$ connecting $x$ and $y$. 
		\item The magnetic Weyl system is gauge-covariant, \ie if $A' = A + \eps \dd \chi$ is an equivalent vector potential for $\dd A' = B = \dd A$ and $\chi \in \Cont^{\infty}_{\mathrm{pol}}(\R^d,\R)$, then the two magnetic Weyl systems are unitarily equivalent, 
		\begin{align}
			w^{A'}(X) &= w^{A + \eps \dd \chi}(X) 
			= \e^{+ \ii \lambda \chi(Q)} \, w^A(X) \, \e^{- \ii \lambda \chi(Q)} 
			&&
			\forall X \in \pspace
			. 
		\end{align}
		\item For all points in phase space $X , Y \in \pspace$ we have 
	    \begin{align}
	        w^A(X) \, w^A(Y) &= \e^{+ \ii \frac{\eps}{2} \sigma(X,Y)} \, \e^{- \ii \frac{\lambda}{\eps} \Gamma^B(Q,Q + \eps x,Q + \eps x + \eps y)} \, w^A(X+Y) 
			\notag \\
	        &= \e^{+ \ii \frac{\eps}{2}\sigma(X,Y)} \, \cocyp(Q;x,y) \, w^A(X+Y) 
			,
			\label{magnetic_Weyl_calculus:eqn:composition_magnetic_Weyl_system}
	    \end{align}
		where the magnetic flux 
		\begin{align*}
			\Gamma^B(q,x,y) := \int_{\sexpval{q,x,y}} B 
		\end{align*}
		through the triangle with corners $x$, $y$ and $z$ enters into the definition of 
		\begin{align}
			\cocyp(q;x,y) := \e^{- \ii \frac{\lambda}{\eps} \Gamma^B(q , q + \eps x,q + \eps x + \eps y)}
			. 
			\label{magnetic_Weyl_calculus:eqn:definition_scaled_magnetic_phase_factor}
		\end{align}
	\end{enumerate}
\end{lemma}
Proofs for these specific facts can be found in \cite[Sections~3.1–3.2]{Lein:progress_magWQ:2010}. 
\begin{remark}[Notation]
	As is usual in pseudodifferential theory, capital letters $X = (x,\xi)$, $Y = (y,\eta)$ and $Z = (z,\zeta)$ denote points in phase space $\pspace = T^* \R^d \cong \R^{2d}$, \ie the cotangent bundle $T^* \R^d$ equipped with the symplectic form 
	\begin{align*}
		\varpi^B = \sum_{j = 1}^d \dd x_j \wedge \dd \xi_j + \frac{\lambda}{2} \sum_{j , k = 1}^d B_{jk} \, \dd x_j \wedge \dd x_k 
		. 
	\end{align*}
	However, the expert reader will have noticed that we have deviated from standard notation. Traditionally, the magnetic Weyl system is denoted with $W^A(X)$ and magnetic Weyl quantization with $\Op^A$, \ie they are denoted with capital letters. However, later on we will introduce the analogs on the level of \emph{super} operators. Since they build upon the objects introduced in this section, we will denote operators acting on Hilbert spaces with small letters, \eg $w^A(X)$ and $\op^A(f)$; the analogous expressions for super operators are capitalized, \eg $W^A(\Xbf)$ denotes the magnetic \emph{super} Weyl system~\eqref{formal_super_calculus:eqn:super_Weyl_system} and $\Op^A(F)$ the magnetic \emph{super} Weyl quantization~\eqref{formal_super_calculus:eqn:Op_A_formal_definition}. 
\end{remark}
%

\subsection{Magnetic Weyl quantization and the Wigner transform} 
\label{magnetic_Weyl_calculus:quantization_Wigner_transform}
Magnetic Weyl quantization is defined analogously to \eqref{magnetic_Weyl_calculus:eqn:definition_non_magnetic_op}, we just need to replace the magnetic Weyl system to obtain 
\begin{align}
	f(Q,P^A) \equiv \op^A(f)
	= \frac{1}{(2\pi)^d} \int_{\pspace} \dd X \, (\Fourier_{\sigma} f)(X) \, w^A(X) 
	. 
	\label{magnetic_Weyl_calculus:eqn:definition_magnetic_Weyl_quantization}
\end{align}
Of course, this formula initially makes sense for Schwartz functions $f \in \Schwartz(\pspace)$, but can be extended to more general functions and tempered distributions (see \eg \cite[Section~IV]{Mantoiu_Purice:magnetic_Weyl_calculus:2004} for details). 

Magnetic Weyl quantization intertwines the operator adjoint with complex conjugation, that is, 
\begin{align*}
	\op^A(f)^* = \op^A(\bar{f}) 
	. 
\end{align*}
Importantly, magnetic Weyl quantization directly inherits gauge-covariance from the Weyl system (Lemma~\ref{magnetic_Weyl_calculus:lem:properties_Weyl_system}~(1)), 
\begin{align}
	\op^{A + \eps \dd \chi}(f) &= \e^{+ \ii \lambda \chi(Q)} \, \op^A(f) \, \e^{- \ii \lambda \chi(Q)} 
	. 
	\label{magnetic_Weyl_calculus:eqn:gauge_covariance_op_A}
\end{align}
In fact, this is \emph{the} distinguishing feature between magnetic Weyl quantization and non-magnetic Weyl quantization after minimal substitution. More precisely, the non-magnetic Weyl quantization of  $f_A(x,\xi) := f \bigl ( x , \xi - \lambda A(x) \bigr )$ fails to satisfy the gauge-covariance condition, 
\begin{align*}
	\op_{A + \eps \dd \chi}(f) := \op(f_{A + \eps \dd \chi}) &\neq \e^{+ \ii \lambda \chi(Q)} \, \op_A(f) \, \e^{- \ii \lambda \chi(Q)} 
	. 
\end{align*}
The only exception are polynomials up to quadratic order in $\xi$ and functions of $x$ only. 

This is more than just a cosmetic problem. From a mathematical perspective, this places additional, unnecessary restrictions on the magnetic fields: say, we want to consider the $\Psi$DO associated to a Hörmander symbol, \ie a function that belongs to the following class: 
\begin{definition}[Hörmander symbols $S^m_{\rho,\delta}(\pspace)$]\label{magnetic_Weyl_calculus:defn:Hoermander_classes}
	The class of Hörmander symbols of order $m \in \R$ and type $(\rho,\delta)$, $0 \leq \delta \leq \rho \leq 1$, form the Fréchet space 
	\begin{align*}
		S^m_{\rho,\delta}(\pspace) := \Bigl \{ f \in \Cont^{\infty}(\pspace) \; \; \big \vert \; \; \forall a , \alpha \in \N_0^d : \; \snorm{f}_{m,a \alpha} < \infty \Bigr \} 
	\end{align*}
	where for $a , \alpha \in \N_0^d$ the seminorms are defined as 
	\begin{align*}
		\snorm{f}_{m , a \alpha} := \sup_{(x,\xi) \in \pspace} \Bigl ( \sexpval{\xi}^{-m - \sabs{a} \delta + \sabs{\alpha} \rho} \; \babs{\partial_x^a \partial_{\xi}^{\alpha} f(x,\xi)} \Bigr ) 
	\end{align*}
	Here, $\sexpval{\xi} := \sqrt{1 + \sabs{\xi}^2}$ is the Japanese bracket. 
\end{definition}
\begin{remark}[Equivalent family of seminorms]
	Very often we shall use a different family of seminorms that generates the same Fréchet topology in the end: for $m \in \R$ and $N \in \N_0$ we set 
	\begin{align}
		p^m_N(f) := \max_{\sabs{a} + \sabs{\alpha} \leq N} \snorm{f}_{m , a \alpha} 
		. 
		\label{magnetic_Weyl_calculus:eqn:max_seminorm_Hoermander_symbols}
	\end{align}
\end{remark}
For Hörmander symbols, we can expect to have a good theory for magnetic fields that admit a vector potential $A$ whose components are $\Cont^{\infty}_{\mathrm{b}}$; this ensures that $f_A$ belongs to the same Hörmander class as $f$. With a little more finagling, the assumption can be relaxed to the case where all derivatives of $A$ need to be bounded, \ie $A$ may grow at most linearly. 

But there is a conceptual cost, too: properties of systems with magnetic fields only depend on the magnetic field rather than the choice of vector potential. Gauge-covariance ensures that \eg spectrum and the spectral decomposition of $\op^A(f)$ does not depend on the choice of gauge and only on $B$; that is false for $\op_A(f)$. 

\subsection{The magnetic Wigner transform} 
\label{magnetic_Weyl_calculus:wigner_transform}
The inverse of Weyl quantization is constructed from the magnetic Wigner transform 
\begin{align}
	(\wigner^A K)(x,\xi) :=  \frac{1}{(2\pi)^{d/2}} \int_{\R^d} \dd y \, \e^{- \ii y \cdot \xi} \, 
	\e^{- \ii \frac{\lambda}{\eps} \Gamma^A([x - \frac{\eps}{2} y , x + \frac{\eps}{2} y])} \, 
	K \bigl ( \tfrac{x}{\eps} + \tfrac{y}{2} , \tfrac{x}{\eps} - \tfrac{y}{2} \bigr )
	\label{magnetic_Weyl_calculus:eqn:magnetic_Wigner_transform}
\end{align}
that maps some suitable function or distribution on $\R^d \times \R^d$ onto a function or distribution on phase space $\pspace$. 

More specifically, if $K_T$ is the distributional operator kernel of some operator $T \in \mathcal{B}(\Hil^A)$, then the inverse of $\op^A$ is given by 
\begin{align}
	{\op^A}^{-1}(T) := \wigner^A K_T 
	. \label{magnetic_Weyl_calculus:eqn:inverse_of_magnetic_Weyl_quantization}
\end{align}
As before, we refer the interested reader to the aforementioned literature if they would like to know how to make sense of this formal expression~\eqref{magnetic_Weyl_calculus:eqn:magnetic_Wigner_transform} for more general classes of functions and distributions. 

\subsection{The magnetic Weyl product} 
\label{magnetic_Weyl_calculus:product}
The third ingredient of a pseudodifferential calculus is a product $\weyl^B$ implicitly defined through 
\begin{align}
	\op^A \bigl ( f \weyl^B g \bigr ) := \op^A(f) \, \op^A(g)
\end{align}
that mimics the operator product on the level of functions or distributions on phase space $\pspace$. For certain classes of functions such as $\mathcal{S}(\pspace)$ and $S^m_{\rho,\delta}(\pspace)$, we can express the magnetic Weyl product as the (oscillatory) integral 
\begin{align}
	\bigl ( f \weyl^B g \bigr )(X) &= \frac{1}{(2\pi)^{2d}} \int_{\pspace} \dd Y \int_{\pspace} \dd Z \, \e^{+ \ii \sigma(X,Y+Z)} \, \e^{+ \ii\frac{\eps}{2} \sigma(Y,Z)} \, \e^{- \ii \lambda \trifluxp(x,y,z)} (\Fourier_{\sigma} f)(Y) \, (\Fourier_{\sigma} g)(Z)  
	\label{magnetic_Weyl_calculus:eqn:magnetic_Weyl_product_formula}
\end{align}
where the magnetic phase factor is the scaled magnetic flux 
\begin{align}
	\trifluxp(x,y,z) := \tfrac{1}{\eps} \, \Gamma^B \bigl ( x - \tfrac{\eps}{2}(y+z) \, , x + \tfrac{\eps}{2}(y-z) \, , x + \tfrac{\eps}{2}(y+z) \bigr ) 
	. 
	\label{magnetic_Weyl_calculus:eqn:scaled_magnetic_flux}
\end{align}
Note that since the area of the triangle the magnetic field passes through is $\order(\eps^2)$, this scaled magnetic flux is actually $\triflux = \order(\eps)$ small. 

The magnetic Weyl product inherits the non-commutativity of the operator product, and gauge-covariance~\eqref{magnetic_Weyl_calculus:eqn:gauge_covariance_op_A} of $\op^A$ implies that $\weyl^B$ depends on the magnetic \emph{field} $B$, not on the magnetic \emph{vector potential} $A$. 

One of the standard results in the literature is that the magnetic Weyl product $\weyl^B$ maps two Hörmander symbols onto a Hörmander symbol, \ie the bilinear map 
\begin{align*}
	\weyl^B : S^{m_1}_{\rho,\delta}(\pspace) \times S^{m_2}_{\rho,\delta}(\pspace) \longrightarrow S^{m_1 + m_2}_{\rho,\delta}(\pspace) 
\end{align*}
is continuous with respect to the relevant Fréchet topologies (\cf \cite[Theorem~2.6]{Iftimie_Mantoiu_Purice:magnetic_psido:2006}). 

What makes (magnetic) Weyl calculus such a nice tool in rigorous perturbation theory is that it allows one to systematically expand the operator product in a small parameter, \eg we are able to expand 
\begin{align*}
	f \weyl^B g &= \sum_{n = 0}^{\infty} \eps^n \, \bigl ( f \weyl^B g \bigr )_{(n)} + \order(\eps^{\infty})
	\\
	&= f \, g - \eps \, \tfrac{\ii}{2} \{ f , g \}_B + \order(\eps^2)
\end{align*}
asymptotically in the semiclassical parameter $\eps$, where the first subleading correction is given in terms of the magnetic Poisson bracket 
\begin{align*}
	 \{ f , g \}_B := \nabla_{\xi} f \cdot \nabla_x g - \nabla_x f \cdot \nabla_{\xi} g - \lambda \, \sum_{j , k = 1}^d B_{jk} \, \partial_{\xi_j} f \, \partial_{\xi_k} g 
	. 
\end{align*}
Making this asymptotic expansion rigorous for Hörmander-class symbols (\cf \cite{Lein:two_parameter_asymptotics:2008} and \cite[Chapter~3]{Lein:progress_magWQ:2010}) allows one to prove an Egorov-type theorem, \cite[Theorem~3.6.1]{Lein:progress_magWQ:2010}, which is one path to connecting the full quantum dynamics to classical Hamiltonian equations of motion. 

\subsection{Algebraic point of view of pseudodifferential theory} 
\label{magnetic_Weyl_calculus:algebraic_point_of_view}
Magnetic pseudodifferential theory is intimately connected to the theory of twisted crossed product $C^*$-algebras \cite{Mantoiu_Purice_Richard:twisted_X_products:2004,Lein_Mantoiu_Richard:anisotropic_mag_pseudo:2009,Belmonte_Lein_Mantoiu:mag_twisted_actions:2010}. They appear naturally and allow for \eg to study essential spectra of magnetic pseudodifferential operators (\cf \cite[Theorem~4.2]{Lein_Mantoiu_Richard:anisotropic_mag_pseudo:2009}). Likewise, the theory of magnetic pseudodifferential \emph{super} operators has an algebraic point of view, something we intend to develop in a follow-up paper. Even though we are postponing an in-depth discussion to another time, we have decided to include this section to point out a few facts that will help the reader better understand our choice of notation, our thought processes and where the journey will be going. 

The twisted crossed product algebras are all $C^*$-subalgebras of 
\begin{align*}
	\mathfrak{C}^B := {\op^A}^{-1} \Bigl ( \mathcal{B} \bigl ( L^2(\R^d) \bigr ) \Bigr ) 
	. 
\end{align*}
Endowed with the norm $\norm{f}_B := \bnorm{\op^A(f)}_{\mathcal{B}(L^2(\R^d))}$, the involution $f^{\weyl^B} := {\op^A}^{-1} \bigl ( \op^A(f)^* \bigr )$ and the product 
\begin{align*}
	f \weyl^B g := {\op^A}^{-1} \bigl ( \op^A(f) \, \op^A(g) \bigr ) 
\end{align*}
the vector space $\mathfrak{C}^B$ in fact inherits the $C^*$-algebraic structure of its parent $\mathcal{B} \bigl ( L^2(\R^d) \bigr )$. 

The Schwartz Kernel Theorem tells us that $\mathfrak{C}^B$ is composed of tempered distributions, although of course, many of its elements can be regarded as functions in the ordinary sense. For, say, Schwartz functions the product $\weyl^B$ is computed through equation~\eqref{magnetic_Weyl_calculus:eqn:magnetic_Weyl_product_formula} and the involution $f^{\weyl^B} = \bar{f}$ is just pointwise complex conjugation. 

As the notation suggests, $\mathfrak{C}^B$ depends only on the magnetic field $B$, courtesy again of gauge-covariance~\eqref{magnetic_Weyl_calculus:eqn:gauge_covariance_op_A}. Indeed, this can be seen from the explicit expressions of the product $\weyl^B$, involution ${}^{\weyl^B}$ and norm $\norm{\, \cdot \,}_B$. 

That is why $\mathfrak{C}^B$ and suitable subalgebras $\mathfrak{A}^B \subseteq \mathfrak{C}^B$ are often considered more fundamental. In the language of algebras, choosing a vector potential $A$ is tantamount to choosing a \emph{representation} 
\begin{align*}
	\op^A : \mathfrak{A}^B \subseteq \mathfrak{C}^B \longrightarrow \mathcal{B} \bigl ( L^2(\R^d) \bigr ) = \mathcal{B}(\Hil^A) 
	. 
\end{align*}
Representations with respect to equivalent gauges are \emph{unitarily} equivalent. This is why we prefer to keep track of the choice of vector potential $A$ in our notation for the Hilbert space $\Hil^A = L^2(\R^d)$; the superscript $A$ is meant to remind the reader that we are using $\op^A$ to promote functions to operators rather than $\op^{A + \eps \dd \chi}$. By construction, $\op^A$ is faithful and hence, norm-preserving. Therefore, $\norm{\, \cdot \,}_B$ coincides with the universal $C^*$-norm, a fact that is quite important when connecting these more intrinsic, $B$-dependent algebras to pseudodifferential theory. 

Once we select a trace with suitable properties, we can introduce a notion of integration and develop a theory of non-commutative $L^p$ spaces (\cf \cite[Chapter~3.2]{DeNittis_Lein:linear_response_theory:2017} and references therein). Rather than equivalence classes of functions these non-commutative $L^p$ spaces consist of suitable operators, some of which can be regarded as magnetic pseudodifferential operators. Our ultimate goal is to study how magnetic pseudodifferential \emph{super} operators act on these non-commutative $L^p$ spaces. Quite naturally, this necessitates notions from analysis such as operator domains and Sobolev spaces to be translated to the realm of operator algebras. 
\section{Formal derivations of magnetic pseudodifferential super operator calculus} 
\label{formal_super_calculus}
One of our motivations to introduce and study a magnetic pseudodifferential calculus is to define \emph{super} operators. A relevant example from quantum mechanics is the Liouville operator 
\begin{align*}
	\hat{L}_{\hat{h}^A} = - \ii \, \bigl [ \hat{h}^A \, , \, \cdot \, \bigr ]
	= \Op^A(L_h)
\end{align*}
associated to the pseudodifferential operator $\hat{h}^A = \op^A(h)$, which is the magnetic quantization of a function $L_h$. It turns out the solution is 
\begin{align*}
	L_h(X_L,X_R) = - \ii \, \bigl ( h(X_L) - h(X_R) \bigr ) 
	, 
\end{align*}
where $X_L = (x_L,\xi_L) \in \pspace$ and $X_R = (x_R,\xi_R) \in \pspace$ are left and right phase space variables; the reason why we call them left and right variables is that the contribution due to $- \ii \, h(X_L)$ acts from the left and that from $+ \ii \, h(X_R)$ acts from the right. 

Exploiting the pseudodifferential nature of the super operator $\hat{L}_{\hat{h}^A}$ would be tremendously beneficial for the same reasons that ordinary pseudodifferential is: it would allow us to deduce many of the fundamental properties of $\hat{L}_{\hat{h}^A}$ from properties of the functions $L_h$ and $h$ that define it. Moreover, it might allow us to side step many of the technical problems one encounters when working with super operators. One is the issue of measurability that complicates the precise definition of products and commutators such as the one that enters the Liouville super operator $\hat{L}_{\hat{h}^A}$ (\cf \cite[Chapter~3.3]{DeNittis_Lein:linear_response_theory:2017}). That could lead to simplified proofs in applications such as linear response theory \cite{DeNittis_Lein:linear_response_theory:2017} and the development of new rigorous perturbation schemes on the level of super operators in the spirit of \eg \cite{PST:sapt:2002}. 

To motivate the main equations and fix some notation, let us dispense with mathematical rigor for the moment. Suppose we would like to understand super operators of the form 
\begin{align*}
	\hat{F}^A \, \hat{g}^A := \op^A(f_L) \, \hat{g}^A \, \op^A(f_R)
	. 
\end{align*}
Here, the operator in the middle $\hat{g}^A \in \mathcal{B} \bigl ( L^2(\R^d) \bigr ) , \mathfrak{L}^p \bigl ( \mathcal{B} \bigl ( L^2(\R^d) \bigr ) \bigr )$ is a bounded operator or belongs to the $p$-Schatten class with respect to the trace;\footnote{Our notation borrows from the theory of non-commutative $L^p$ spaces rather than more traditional functional analysis where $\mathcal{T}^p(\Hil) = \mathfrak{L}^p \bigl ( \mathcal{B}(\Hil) \bigr )$ means we use the canonical trace on the Hilbert space $\Hil$. Since this is not necessarily the case for us, we shall use slightly different notation. } the superscript $A$ serves as a reminder to the reader that we must specify a vector potential to fix a representation on a Hilbert space. The operators on the left and on the right are ordinary magnetic pseudodifferential operators; the indices $L$ and $R$ will indicate whether the associated $\Psi$DO acts on $\hat{g}^A$ from the left or the right. 

Plugging in the definition~\eqref{magnetic_Weyl_calculus:eqn:definition_magnetic_Weyl_quantization} for magnetic pseudodifferential operators yields 
\begin{align}
	\hat{F}^A \, \hat{g}^A = \frac{1}{(2\pi)^{2d}} \int_{\pspace} \dd X_L \int_{\pspace} \dd X_R \, (\Fourier_{\sigma} f_L)(X_L) \, (\Fourier_{\sigma} f_R)(X_R) \, w^A(X_L) \, \hat{g}^A \, w^A(X_R) 
	. 
	\label{formal_super_calculus:eqn:product_super_operator_Op_A_formal_definition}
\end{align}
This suggests to introduce the magnetic Weyl super quantization 
\begin{align}
	\Op^A(F) \, \hat{g}^A := \frac{1}{(2\pi)^{2d}} \int_{\Pspace} \dd \Xbf \, (\Fourier_{\Sigma} F)(\Xbf) \, W^A(\Xbf) \, \hat{g}^A 
	. 
	\label{formal_super_calculus:eqn:Op_A_formal_definition}
\end{align}
Here we have collected left and right phase space coordinates into 
\begin{align*}
	\Xbf = (X_L,X_R) \in \pspace \times \pspace 
	=: \Pspace 
	. 
\end{align*}
Analogously to equation~\eqref{magnetic_Weyl_calculus:eqn:symplectic_Fourier_transform} we have introduced a symplectic Fourier transform $\Fourier_{\Sigma}$ on the doubled phase space $\Pspace$, 
\begin{align*}
	(\Fourier_{\Sigma} F)(\Xbf) := \frac{1}{(2\pi)^{2d}} \int_{\Pspace} \dd \Xbf' \, \e^{+ \ii \Sigma(\Xbf,\Xbf')} \, F(\Xbf')
	, 
\end{align*}
that is defined in terms of the symplectic form 
\begin{align*}
	\Sigma(\Xbf,\Ybf) := \sigma(X_L,Y_L) + \sigma(X_R,Y_R) 
	. 
\end{align*}
The last piece of the puzzle is the magnetic super Weyl system 
\begin{align}
	W^A(\Xbf) \, \hat{g}^A := w^A(X_L) \, \hat{g}^A \, w^A(X_R) 
	. 
	\label{formal_super_calculus:eqn:super_Weyl_system}
\end{align}
Importantly, the magnetic super Weyl quantization of $F(X_L,X_R) = f_L(X_L) \, f_R(X_R)$ via equation~\eqref{formal_super_calculus:eqn:Op_A_formal_definition} reduces to the operator from equation~\eqref{formal_super_calculus:eqn:product_super_operator_Op_A_formal_definition}. Since we will use product operators~\eqref{formal_super_calculus:eqn:product_super_operator_Op_A_formal_definition} every now and then, we abbreviate their symbols with $F = f_L \otimes f_R$ and write $\hat{F}^A =  \hat{f}_L^A \otimes \hat{f}_R^A$ with abuse of notation. 

Given that $\Op^A(F)$ is an operator acting on other operators, we may ask what happens if $\hat{g}^A = \op^A(g)$ is itself a magnetic pseudodifferential operator. For product operators~\eqref{formal_super_calculus:eqn:product_super_operator_Op_A_formal_definition} we can immediately rewrite the resulting operator 
\begin{align*}
	\Op^A \bigl ( f_L \otimes f_R \bigr ) \, \op^A(g) &= \op^A(f_L) \, \op^A(g) \, \op^A(f_R) 
	= \op^A \bigl ( f_L \weyl^B g \weyl^B f_R \bigr ) 
\end{align*}
as a double magnetic Weyl product. At this point, we could write out the double magnetic Weyl product explicitly and verify once again that we can replace $(\Fourier_{\sigma} f_L)(X_L) \, (\Fourier_{\sigma} f_R)(X_R)$ with $(\Fourier_{\Sigma} F)(\Xbf)$ for functions $F$ that are not necessarily products. But we will skip ahead to the result and propose the formula 
\begin{align}
	\bigl ( F \semisuper^B g \bigr )(X) = \frac{1}{(2\pi)^{3d}} \int_{\Pspace} \dd \Ybf \int_{\pspace} \dd Z \, &\e^{+ \ii \sigma(X , Y_L + Y_R + Z)} \, \e^{+ \ii \frac{\eps}{2} \sigma(Y_L + Z , Y_R + Z)} \, \e^{- \ii \lambda \recfluxp(x,y_L,y_R,z)} \, \cdot 
	\notag \\
	&\cdot \, 
	(\Fourier_{\Sigma} F)(\Ybf) \, (\Fourier_{\sigma} g)(Z) 
	\label{formal_super_calculus:eqn:semi_super_product_explicit_Fourier}
\end{align}
for the \emph{semi-super product} $\semisuper^B$ that is defined by 
\begin{align}
	 \op^A \bigl ( F \semisuper^B g \bigr ) :& \negmedspace = \Op^A(F) \, \op^A(g) 
	 \label{formal_super_calculus:eqn:definition_semisuper_product}
\end{align}
and involves the magnetic phase factor 
\begin{align}
	\recfluxp(x,y_L,y_R,z) := \trifluxp(x,y_L,z) + \trifluxp(x,y_L+z,y_R)
	. 
	\label{formal_super_calculus:eqn:definition_magnetic_phase_factor_super_product}
\end{align}
When we write out the two summands as magnetic flux integrals, we see that $\recfluxp$ is the magnetic flux through a quadrangle that is made up of two conjoined triangles; we have illustrated this in Figure~\ref{magnetic_super_PsiDOs:figure:flux_quadrangle}. 
\begin{figure}[t]
	\begin{center}
		\begin{minipage}{0.6\linewidth}
			\centering{\def\svgwidth{\columnwidth}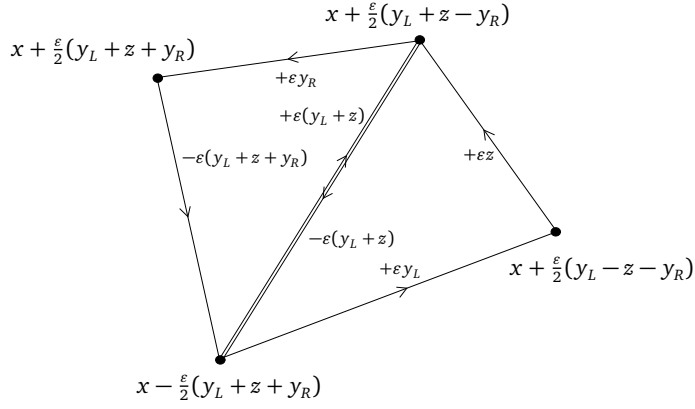}
		\end{minipage}
	\end{center}
	\label{magnetic_super_PsiDOs:figure:flux_quadrangle}
	\caption{This is the magnetic flux quadrangle that appears in the semi-super product $\semisuper^B$.}
\end{figure}

This product does \emph{not} have an analog in regular Weyl calculus, though. Since one of its arguments is a symbol on $\Pspace$ and the other a symbol on $\pspace$, we will speak of the semi-super product. 

The analog of $\weyl^B$ on the level of super operators is played by the magnetic super Weyl product that is defined as the function or tempered distribution $F \super^B G$ that satisfies 
\begin{align*}
	\Op^A \bigl ( F \super^B G \bigr ) := \Op^A(F) \, \Op^A(G) 
	. 
\end{align*}
The magnetic super Weyl product of two product super operators~\eqref{formal_super_calculus:eqn:product_super_operator_Op_A_formal_definition} reduces to 
\begin{align*}
	\Op^A \bigl ( f_L \otimes f_R \bigr ) \, \Op^A(g_L \otimes g_R) \, \hat{h}^A &= \Op^A \bigl ( f_L \otimes f_R \bigr ) \, \op^A \bigl ( g_L \weyl^B h \weyl^B g_R \bigr )
	\\
	&= \op^A \Bigl ( \bigl ( f_L \weyl^B g_L \bigr ) \weyl^B h \weyl^B \bigl ( g_R \weyl^B f_R \bigr ) \Bigr )
	. 
\end{align*}
Note that the order of multiplication of the symbols to the right of $g$ is opposite of those to the left. The brackets are just for emphasis and entirely unnecessary as the magnetic Weyl product $\weyl^B$ is associative, 
\begin{align*}
	\bigl ( f \weyl^B g \bigr ) \weyl^B h &= f \weyl^B \bigl ( g \weyl^B h \bigr )
	= f \weyl^B g \weyl^B h
	. 
\end{align*}
When $F$ and $G$ are not necessarily product symbols, the formula for the magnetic super Weyl product cannot be disentangled into a left and a right component and instead takes the form 
\begin{align}
	F \super^B G(\Xbf) = \frac{1}{(2\pi)^{4d}} \int_{\Pspace} \dd \Ybf \int_{\Pspace} \dd \Zbf \; &\e^{+ \ii \Sigma(\Xbf,\Ybf+\Zbf)} \, \e^{+ \ii \frac{\eps}{2} \Sigma(r(\Ybf),\Zbf)} \, 
	\cdot \notag \\ 
	&
	\e^{- \ii \lambda \trifluxp(x_L,y_L,z_L)} \, \e^{- \ii \lambda \trifluxp(x_R,z_R,y_R)} \, (\Fourier_{\Sigma} F)(\Ybf) \, (\Fourier_{\Sigma} G)(\Zbf) 
	\label{formal_super_calculus:eqn:super_Weyl_product_formula}
	. 
\end{align}
As with the usual magnetic Weyl product $\weyl^B$, also the formulas for the other two products $\semisuper^B$ and $\super^B$ can be recast in different, but ultimately equivalent ways. This is achieved by writing out the Fourier transforms and integrating out some of the variables. 
\medskip

\noindent
The last piece of the puzzle is a dequantization procedure via a super Wigner transform. Also here we can infer the formula for the general case, 
\begin{align}
	\Wigner^B \, K(\Xbf) &= \frac{\eps^{3d}}{2^{4d}\pi^d} \int_{\Pspace} \dd \Ybf \, \e^{+ \ii \xi_L \cdot y_L} \, \e^{+ \ii \xi_R \cdot y_R} \, \e^{+ \frac{\ii}{2} \left ( (-x_L + x_R) \cdot \eta_R - \frac{\eps}{2} (y_L + y_R) \cdot \eta_L \right )} \, 
	\cdot \notag \\
	&\qquad \qquad \qquad \quad
	\e^{+ \ii \lambda \recfluxp \left ( \frac{1}{2}(x_L+x_R) - \frac{\eps}{4} (y_L - y_R) , y_L , y_R , \frac{1}{\eps}(x_R-x_L) - \frac{1}{2} (y_L + y_R) \right )} \, 
	\cdot \notag \\
	&\qquad \qquad \qquad \quad 
	K \bigl ( \tfrac{1}{2} (x_L + x_R) - \tfrac{\eps}{4} (y_L - y_R) , \tfrac{\eps}{4} (\eta_L - \eta_R) , 
	\bigr . \notag \\
	&\qquad \qquad \qquad \qquad \; \bigl . 
	\tfrac{1}{2}(x_L + x_R) + \tfrac{\eps}{4} (y_L - y_R) , \tfrac{\eps}{4} (\eta_L + \eta_R) \bigr )
	, 
\end{align}
from the formula for product symbols since then clearly the $\Op^A$ dequantization is the product of the $\op^A$ dequantizations, 
\begin{align*}
	{\Op^A}^{-1} \bigl ( \hat{f}_L^A \otimes \hat{f}_R^A \bigr ) = {\op^A}^{-1}(\hat{f}_L^A) \otimes {\op^A}^{-1}(\hat{f}_R^A) 
	. 
\end{align*}
While these considerations so far have been purely formal, readers with a background in pseudodifferential theory can probably already see how to make these formal manipulations rigorous. And the people who lack this background should not worry, we try to be pedagogical in our exposition. 

Just like with $\op^A$ the first step is to establish the magnetic super Weyl calculus for Schwartz class functions. In a second step, we extend by duality. And then in a third we use oscillatory integral techniques to prove that \eg $\super^B$ maps two Hörmander class symbols onto a Hörmander class symbol in a continuous fashion. 
\section{Rigorous definition of super calculus on $\Schwartz$} 
\label{rigorous_definition_supercalculus_on_S}
Our arguments in the previous section suggest that we might exploit the tensor product structure. More concretely, it is tempting to first prove all important facts for product super operators by referencing well-established results from ordinary magnetic Weyl calculus. 

Then we bootstrap these arguments to magnetic Weyl super calculus by approximating general functions 
\begin{align*}
	F \approx \sum_{k = 1}^n f_{L,k} \otimes f_{R,k} 
\end{align*}
on $\Pspace$ from the relevant function spaces by products. 

However, this presumes that finite linear combinations lie dense with respect to the standard Fréchet topologies. And unfortunately, in the cases that matter, this is not true. To be more specific, in general there are several ways to construct a tensor product of two topological vector spaces with Fréchet topology \cite[Chapter~43]{Treves:topological_vector_spaces:1967}. Only when at least one of the the locally convex topological vector spaces $\mathcal{X}_1$ and $\mathcal{X}_2$ is nuclear is there only one way to complete the algebraic tensor product $\mathcal{X}_1 \otimes \mathcal{X}_2$ (\cf Definition~50.1 and Theorem~50.1 in \cite{Treves:topological_vector_spaces:1967}); otherwise there are at least two tensor products in the sense that we can take completions of $\mathcal{X}_1 \otimes \mathcal{X}_2$ with respect to two different topologies. 

Schwartz spaces and their duals are nuclear, so we can apply this strategy to some of our arguments in the beginning and borrow some facts from \cite{Mantoiu_Purice:magnetic_Weyl_calculus:2004,Lein:progress_magWQ:2010}. 

Unfortunately, infinite-dimensional Banach spaces \cite[Chapter~50, Corollary~2]{Treves:topological_vector_spaces:1967} and Hörmander classes are \emph{not} nuclear (see \cite{Witt:weak_topology_symbol_spaces:1997}, specifically Theorem~3.2, Remark~3.3 and Proposition~4.4). Therefore the topology of the tensor product “$S^{m_L}_{\rho,\delta} \otimes S^{m_R}_{\rho,\delta}$” differs from the Fréchet topology given in Definition~\ref{symbol_super_calculus:defn:Hoermander_super_symbols}. That unfortunately means we need to establish the existence of certain oscillatory integrals to make sense of \eg $F \super^B G$ by hand; since the arguments are completely standard and not very enlightening, we have opted to include them as Appendix~\ref{appendix:oscillatory_integrals}. The techniques are completely standard, but the computations are somewhat lengthy and involved because we are dealing with twice as many variables and the formulas contain magnetic phase factors. 
\medskip

\noindent
That is why we first develop the magnetic Weyl super calculus for Schwartz functions. Compared to ordinary magnetic Weyl calculus from Section~\ref{magnetic_Weyl_calculus} the only piece of the puzzle that does not have an exact analog is the magnetic semi-super Weyl product. 

Before we proceed, from hereon out let us make the following
\begin{assumption}[Small parameters]
	We assume that the small parameters $\eps$ and $\lambda$ lie in $(0,1]$. 
\end{assumption}
The reason for this \emph{non-essential} assumption is to simplify some of the estimates. That is because when $\lambda \leq 1$ we can simply replace factors like $\lambda^n$ by $1$. In principle, we could replace it with $\eps , \lambda \in (0,R]$ for some $R > 0$.

\subsection{Definition of non-commutative $L^p$ spaces} 
\label{rigorous_definition_supercalculus_on_S:non_commutative_Lp_spaces}
Standard (magnetic) pseudodifferential operators typically act on $L^p(\R^d)$ or spaces derived from it such as (magnetic) Sobolev spaces. The definition of these spaces is part and parcel of every course on functional analysis, the
\begin{align}
	L^p(\R^d) := \Bigl \{ \mbox{$f : \R^d \longrightarrow \C$ measurable} \; \; \big \vert \; \; \int_{\R^d} \dd x \, \sabs{f(x)}^p < \infty \Bigr \} / \sim 
	\label{rigorous_definition_supercalculus_on_S:eqn:definition_ordinary_Lp_spaces}
\end{align}
consist of equivalence classes of measurable functions that agree on a set of full measure and whose $p$th power is absolutely integrable. Most of the time we consider $\op^A(f)$ as an operator acting on the Hilbert space $L^2(\R^d)$. 

The situation with super operators is trickier for several reasons. While there is only one choice of translation-invariant measure (up to a multiplicative constant) on $\R^d$ for constructing $L^p(\R^d)$, several notions of traces can be defined on suitable subsets of $\mathcal{B} \bigl ( L^2(\R^d) \bigr )$. An obvious one is the standard trace 
\begin{align*}
	\mathrm{Tr}_{L^2(\R^d)}(A) := \sum_{n = 1}^{\infty} \scpro{\varphi_n}{A \varphi_n}
\end{align*}
that is initially defined with respect to some orthonormal basis $\{ \varphi_n \}_{n \in \N}$. Another choice in the context of periodic operators is the trace-per-unit-volume 
\begin{align*}
	\mathcal{T}_{\mathrm{puv}}(A) := \lim_{n \to \infty} \mathrm{Vol}(\Lambda_n)^{-1} \, \mathrm{Tr}_{L^2(\R^d)}(1_{\Lambda_n}(\hat{x}) \, A \, 1_{\Lambda_n}(\hat{x}) \bigr ) 
	, 
\end{align*}
where $\Lambda_n \nearrow \R^d$ is any Følner exhausting sequence that eventually covers the entire $\R^d$ and $1_{\Lambda_n}(\hat{x})$ localizes the operator $A$ to the region $\Lambda_n$ \cite{Lenz:random_operators_crossed_products:1999}. 

The definition of non-commutative $L^p$ spaces mimics \eqref{rigorous_definition_supercalculus_on_S:eqn:definition_ordinary_Lp_spaces}, namely 
\begin{align*}
	\mathfrak{L}^p(\Alg) := \Bigl \{ \mbox{$A$ measurable} \; \; \big \vert \; \; \mathcal{T} \bigl ( \sabs{A}^p \bigr ) < \infty \Bigr \} 
	, 
\end{align*}
where $\Alg \subseteq \mathcal{B}(\Hil)$ is a von Neumann algebra and $\mathcal{T}$ a faithful, normal and semifinite (f.n.s.) trace (\cf \cite[Chapter~VII, Definition~1.1]{Takesaki:operator_algebras_2:2003}); we refer to \cite[Section~3.2]{DeNittis_Lein:linear_response_theory:2017} for a more thorough overview of the construction and references therein. Importantly, the notion of measurability is defined with respect to the trace $\mathcal{T}$ (\cf \cite[Definition~3.2.2]{DeNittis_Lein:linear_response_theory:2017}). 
When $\Alg = \mathcal{B} \bigl ( L^2(\R^d) \bigr )$ and $\mathcal{T} = \mathrm{Tr}_{L^2(\R^d)}$, this is nothing but the usual $p$-Schatten class $\mathfrak{L}^p \bigl ( \mathcal{B} \bigl ( L^2(\R^d) \bigr ) \bigr )$. 
But other cases are frequently of interest as well. When we choose the algebra of periodic operators $\Alg = \Alg_{\mathrm{per}} \subset \mathcal{B} \bigl ( L^2(\R^d) \bigr )$, then the natural trace is the trace-per-unit-volume $\mathcal{T} = \mathcal{T}_{\mathrm{puv}}$. And in case $\Alg = L^{\infty}(\R^d)$ is the commutative von Neumann algebra, we eventually recover the usual $L^p$ spaces (\cf \cite[Example~3.2.6~(4)]{DeNittis_Lein:linear_response_theory:2017}). 

So far the whole theory looks very similar to those of standard (commutative) $L^p$ spaces. But there is a twist, one that is a major motivator for developing the theory of pseudodifferential super operators: not all operators we encounter are measurable. An example is the humble Laplacian $-\Delta$ or periodic Schrödinger operators $H = -\Delta + V_{\mathrm{per}}$, which are \emph{not} measurable with respect to the trace-per-unit-volume (\cf \cite[Remark~3.2.7]{DeNittis_Lein:linear_response_theory:2017}). Therefore, the Liouville operator associated to periodic Schrödinger operators can be formally thought of as the product of a measurable and a \emph{non}-measurable operator. The fact that this occurs with just the Laplacian shows that this is not a complication we can deal with by assuming the hamiltonian is measurable. 

Since the purpose of this paper is to introduce the calculus in its simplest form, we will stick to $p$-Schatten classes, \ie we will pick the usual trace $\mathrm{Tr}_{L^2(\R^d)}$ where the define non-commutative $L^p$ spaces $\mathfrak{L}^p \bigl ( L^2(\R^d) \bigr )$. Nevertheless, after suitable extensions and restrictions, we still expect that all formulas hold verbatim when we choose another f.n.s.\ trace and work with suitable subalgebras. 

\subsection{Construction of the quantization} 
\label{rigorous_definition_supercalculus_on_S:construction_quantization}
The first step is to construct an analog of magnetic Weyl quantization $\op^A$. We shall start with a Weyl system that implements the group action of translations in position and momentum. Afterwards, we define $\Op^A$ and discuss the adjoint operation.

\subsubsection{The magnetic super Weyl system and gauge covariance} 
\label{rigorous_definition_supercalculus_on_S:construction_quantization:super_weyl_system}
As before, the starting point of the calculus is the 
\begin{definition}[Magnetic super Weyl system]\label{magnetic_super_PsiDOs:defn:super_Weyl_system}
	For $\Xbf \in \Pspace$ we define the operator 
	\begin{align*}
		W^A(\Xbf) \, \hat{g}^A \equiv W^A(X_L,X_R) \, \hat{g}^A
		:= w^A(X_L) \, \hat{g}^A \, w^A(X_R) 
		, 
		&&
		\hat{g}^A \in \mathcal{B} \bigl ( L^2(\R^d) \bigr ) 
	\end{align*}
	in terms of the magnetic Weyl systems~\eqref{magnetic_Weyl_calculus:eqn:definition_magnetic_Weyl_system}. 
\end{definition}
On the level of super operators gauge covariance works a little differently, because this also changes representation of the underlying Hilbert space. Suppose we are given a unitary 
\begin{align*}
	U : L^2(\R^d) \longrightarrow \Hil
\end{align*}
and we consider the position and kinetic momentum operators 
\begin{subequations}\label{rigorous_definition_supercalculus_on_S:eqn:unitarily_equivalent_building_block_operators}
	\begin{align}
		Q_U :& \negmedspace = U \, Q \, U^{-1} 
		,
		\\
		P^A_U :& \negmedspace = U \, P^A \, U^{-1} 
		,
	\end{align}
\end{subequations}
in the new representation. Note that a change of gauge is just such a change of representation where 
\begin{align*}
	U &= \e^{+ \ii \lambda \chi(Q)}
\end{align*}
holds for some real-valued function $\chi : \R^d \longrightarrow \R$. Then the unitary commutes with $Q$ and modifies the vector potential to $A' = A + \dd \chi$ for kinetic momentum, 
\begin{align*}
	Q_U &= Q 
	, 
	\\
	P_U^A &= P^{A + \dd \chi}
	. 
\end{align*}
But there are other interesting unitaries to choose from, \eg the Fourier transform or the Bloch-Floquet-Zak transform (\cf \cite[Section~2.2]{DeNittis_Lein:Bloch_electron:2009}) that is used to analyze (perturbed) periodic operators. Importantly, the unitarity ensures that there is no need to touch our estimates below for \eg the magnetic super Weyl product. 

Unlike ordinary Weyl calculus, we also need to covariantly change the operator the Weyl system acts on, \ie we apply the adjoint map with respect to $U$, 
\begin{align*}
	\Ad_U : \mathcal{B} \bigl ( L^2(\R^d) \bigr ) \ni \hat{g}^A \mapsto \hat{g}^A_U := U \, \hat{g}^A \, U^{-1} \in \mathcal{B}(\Hil) 
	. 
\end{align*}
Thanks to the unitarity of $U$ the two operators have identical operator norm, 
\begin{align*}
	\snorm{\hat{g}^A_U}_{\mathcal{B}(\Hil)} = \snorm{\hat{g}^A}_{\mathcal{B}(L^2(\R^d))} 
	. 
\end{align*}
With this notation in hand, we can summarize the covariance condition as 
\begin{align}
	W^A_U(\Xbf) \, \hat{g}^A_U :& \negmedspace = w^A_U(X_L) \, \hat{g}^A_U \, w^A_U(X_R) 
	= U \, w^A(X_L) \, U^{-1} \, \hat{g}^A_U \, U \, w^A(X_R) \, U^{-1} 
	\notag \\
	&= U \, w^A(X_L) \, \hat{g}^A \, w^A(X_R) \, U^{-1} 
	= \Ad_U \, W^A(\Xbf) \, \hat{g}^A
	\notag \\
	&= \Ad_U \, W^A(\Xbf) \, \Ad_U^{-1} \, \hat{g}^A_U 
	= \Ad_{\Ad_U} \, W^A(\Xbf) \, \hat{g}^A_U
	\label{rigorous_definition_supercalculus_on_S:eqn:super_Weyl_system_new_representation}
	. 
\end{align}
where the magnetic Weyl system in the new representation is 
\begin{align*}
	w^A_U(X_L) :& \negmedspace = \e^{- \ii \sigma(X_L,(Q_U,P^A_U))}
	\\
	&= U \, \e^{- \ii \sigma(X_L,(Q,P^A))} \, U^{-1} 
	= \Ad_U \, w^A(X_L) 
	. 
\end{align*}
In contrast to ordinary magnetic Weyl calculus, when we change the gauge with $U = \e^{+ \ii \lambda \chi(Q)}$, we need to keep track of this change of gauge on the level of the operator $\hat{g}^A$ the magnetic super Weyl system acts on. 

Importantly, the composition law of both, the Weyl and the super Weyl system are identical. This reflects the fact that the commutation relations of $Q_U$ and $P^A_U$ are identical to those of $Q$ and $P^A$. 

The action of the gauge transformation on $\hat{g}^A$ simplifies when we use the Schwartz kernel theorem to write 
\begin{align*}
	\hat{g}^A = \op^A(g)
\end{align*}
as the magnetic pseudodifferential operator associated to some tempered distribution $g \in \Schwartz'(\pspace)$. 
\medskip

\noindent
To anticipate our discussion of the semi-super and the super product, it is helpful to study the composition laws of the magnetic super Weyl system. Given that the magnetic super Weyl system is defined in terms of the ordinary magnetic Weyl system and that we can write any bounded operator $\hat{g}^A = \op^A(g)$ as the magnetic pseudodifferential operator for some distribution $g$, we will also make use of a second composition law. The advantage now is that the magnetic phase factors neatly combine with one another. 
\begin{lemma}\label{magnetic_super_PsiDOs:lem:composition_super_Weyl_system_Weyl_system}
	For any $\Xbf \in \Pspace$ and $Y \in \pspace$ the magnetic super Weyl system and the magnetic Weyl system satisfy the composition law 
	\begin{align*}
		W^A(\Xbf) \, w^A(Y) &= \e^{\ii \frac{\eps}{2} \sigma(X_L+Y,X_R+Y)} \, \omega^B(Q;x_L,y) \, \omega^B(Q;x_L+y,x_R) \, w^A(Y+X_L+X_R) 
		.
	\end{align*}
\end{lemma}
Likewise, the product of two magnetic super Weyl systems also simplifies to a super Weyl system operator adjoined by a magnetic and a non-magnetic phase. 
%

\subsubsection{The magnetic super Weyl quantization} 
\label{rigorous_definition_supercalculus_on_S:construction_quantization:Op_A}
Just like before, the magnetic \emph{super} Weyl calculus is constructed from the magnetic super Weyl system, the only difference being that we are promoting functions on $\Pspace$ (rather than $\pspace$) to operators. 
\begin{definition}[Magnetic super Weyl quantization]\label{magnetic_super_PsiDOs:defn:super_Weyl_quantization}
	The magnetic super Weyl quantization of a Schwartz function $F \in \Schwartz(\Pspace)$ is defined in the strong sense as the Bochner integral 
	\begin{align}
		\Op^A(F) \, \hat{g}^A := \frac{1}{(2\pi)^{2d}} \int_{\Pspace} \dd \Xbf \, (\Fourier_{\Sigma} F)(\Xbf) \, W^A(\Xbf) \, \hat{g}^A 
		. 
		\label{magnetic_super_PsiDOs:eqn:definition_super_Weyl_quantization}
	\end{align}
\end{definition}
The above definition intentionally does not specify the space from which we take $\hat{g}^A$. One common choice is $\mathcal{B} \bigl ( L^2(\R^d) \bigr )$, but we might also want to consider the $p$-Schatten classes $\mathfrak{L}^p \bigl ( \mathcal{B} \bigl ( L^2(\R^d) \bigr ) \bigr )$ [cite] or operators on $L^2(\R^d) \otimes \mathfrak{h}$ where $\mathfrak{h}$ is some separable Hilbert space. Let us collect some basic facts. 
\begin{lemma}\label{rigorous_definition_supercalculus_on_S:lem:properties_super_Weyl_quantization}
	Suppose Assumption~\ref{intro:assumption:polynomially_bounded_magnetic_field} on the magnetic field $B = \dd A$ and the vector potential $A$ are satisfied, and let $F \in \Schwartz(\Pspace)$. Then the following holds: 
	\begin{enumerate}[(1)]
		\item $\Op^A(F)$ gives rise to a bounded linear operator from $\bounded \bigl ( L^2(\R^d) \bigr )$ to itself.
		\item For every $p \geq 1$, $\Op^A(F)$ gives rise to a bounded linear operator from $\mathfrak{L}^p \bigl ( \mathcal{B} \bigl ( L^2(\R^d) \bigr ) \bigr )$ to itself, where $\mathfrak{L}^p \bigl ( \mathcal{B} \bigl ( L^2(\R^d) \bigr ) \bigr )$ is the ideal of $p$th Schatten-class operators.
	\end{enumerate}
\end{lemma}
\begin{proof}
	\begin{enumerate}[(1)]
		\item This assertion follows easily by replacing the Schatten $p$-norm $\norm{\, \cdot \,}_p$ by the operator norm $\norm{\, \cdot \,}$ in the below arguments that prove (2). 
		\item Let $p \geq 1$, $\hat{g}^A \in \mathfrak{L}^p \bigl ( \mathcal{B} \bigl ( L^2(\R^d) \bigr ) \bigr )$ and denote the Schatten $p$-norm of $\mathfrak{L}^p \bigl ( \mathcal{B} \bigl ( L^2(\R^d) \bigr ) \bigr )$ by $\norm{\, \cdot \,}_p$ and the operator norm of $\mathcal{B} \bigl ( L^2(\R^d) \bigr )$ by $\norm{\, \cdot \,}$. We know that $\mathfrak{L}^p \bigl ( \mathcal{B} \bigl ( L^2(\R^d) \bigr ) \bigr )$ is the two-sided ideal of $\mathcal{B} \bigl ( L^2(\R^d) \bigr )$ and $\bnorm{\hat{f} \, \hat{g}^A \, \hat{h}}_p \leq \snorm{\hat{f}} \, \snorm{\hat{g}^A}_p \, \snorm{\hat{h}}$ for all $\hat{f} , \hat{h} \in \mathcal{B} \bigl ( L^2(\R^d) \bigr )$ and $\hat{g}^A \in \mathfrak{L}^p \bigl ( \mathcal{B} \bigl ( L^2(\R^d) \bigr ) \bigr )$ (see, \eg \cite{Gohberg_Krein:linear_nonselfadjoint_operators:1969}). Therefore, as the product of two unitaries and an operator from the ideal $\mathfrak{L}^p \bigl ( \mathcal{B} \bigl ( L^2(\R^d) \bigr ) \bigr )$, we conclude 
		\begin{align*}
			W^A(\Xbf) \, \hat{g}^A = w^A(X_L) \, \hat{g}^A \, w^A(X_R) \in \mathfrak{L}^p \bigl ( \mathcal{B} \bigl ( L^2(\R^d) \bigr ) \bigr ) 
		\end{align*}
		and its Schatten $p$-norm 
		\begin{align*}
			\bnorm{W^A(\Xbf) \, \hat{g}^A}_p \leq \bnorm{w^A(X_L)} \, \snorm{\hat{g}^A}_p \, \bnorm{w^A(X_R)} = \snorm{\hat{g}^A}_p 
		\end{align*}
		reduces to that of $\hat{g}^A$. 
		
		Combined with the fact that the symplectic Fourier transform $\Fourier_{\Sigma}$ maps Schwartz functions onto Schwartz functions, and Schwartz functions are integrable, we deduce that the Schatten $p$-norm
		\begin{align*}
			\bnorm{\Op^A(F) \, \hat{g}^A}_p &\leq \frac{1}{(2\pi)^{2d}} \int_{\Xi^2} \dd \Xbf \, \Bnorm{ (\Fourier_{\Sigma} F)(\Xbf) \, W^A(\Xbf) \, \hat{g}^A}_p
			\\
			&= \left ( \frac{1}{(2\pi)^{2d}} \int_{\Xi^2} \dd \Xbf \, \babs{(\Fourier_{\Sigma} F)(\Xbf)} \right ) \snorm{\hat{g}^A}_p 
			\\
			&= (2\pi)^{-2d} \, \bnorm{\Fourier_{\Sigma} F}_{L^1(\Xi^2)} \, \snorm{\hat{g}^A}_p
			.
		\end{align*}
		is finite. This shows that $\Op^A(F)$ gives rise to a continuous linear operator on $\mathfrak{L}^p \bigl ( \mathcal{B} \bigl ( L^2(\R^d) \bigr ) \bigr )$. This proves the assertion~(2).
	\end{enumerate}
\end{proof}
In a future work, we plan to generalize this formalism by replacing $\mathcal{B} \bigl ( L^2(\R^d) \bigr )$ with some other $C^*$- or von Neumann algebra and replace the usual trace with a different trace. Of course, we then need to impose additional assumptions, \eg that the Weyl system is compatible with the algebra and the trace akin to Hypothesis~3 from \cite{DeNittis_Lein:linear_response_theory:2017}. 

\subsubsection{The adjoint} 
\label{rigorous_definition_supercalculus_on_S:construction_quantization:adjoint}
Ordinary (magnetic) pseudodifferential operators naturally come furnished with an involution as most of the time they are considered as operators on the Hilbert space $L^2(\R^d)$. Certainly, it is possible to use the inclusions 
\begin{align*}
	\Schwartz(\R^d) \subseteq L^p(\R^d) \subseteq \Schwartz'(\R^d) 
\end{align*}
to extend $\op^A(f)$ to Banach spaces where $1 \leq p < \infty$ and not just $p = 2$. 

The situation for super operators is different, though. Here, it is much more common to study $\Op^A(F)$ on \eg $\mathfrak{L}^1 \bigl ( \mathcal{B} \bigl ( L^2(\R^d) \bigr ) \bigr )$ or $\mathfrak{L}^{\infty} \bigl ( \mathcal{B} \bigl ( L^2(\R^d) \bigr ) \bigr ) = \mathcal{B} \bigl ( L^2(\R^d) \bigr )$ than $p = 2$. And for $p \neq 2$ usually one chooses the linear (as opposed to antilinear) map $\Op^A(F) \mapsto \Op^A(F)'$ where the prime denotes the adjoint with respect to the dual pairing 
\begin{align*}
	\bigl ( \hat{f}^A \, , \, \hat{g}^A \bigr ) \mapsto \mathrm{Tr}_{L^2(\R^d)} \bigl ( \hat{f}^A \, \hat{g}^A \bigr )
	, 
\end{align*}
where $\hat{g}^A \in \mathfrak{L}^p \bigl ( \mathcal{B}(L^2(\R^d) \bigr ) \bigr )$, $\hat{f}^A \in \mathfrak{L}^q \bigl ( \mathcal{B}(L^2(\R^d) \bigr ) \bigr )$ and $q$ is dual to $p$, that is, $\nicefrac{1}{p} + \nicefrac{1}{q} = 1$. 

When $p = 2$ it is customary to pick the \emph{anti}linear Hilbert adjoint $\Op^A(F) \mapsto \Op^A(F)^*$ since then the dual pairing 
\begin{align*}
	\bscpro{\hat{f}^A}{\hat{g}^A}_{\mathfrak{L}^2(\mathcal{B}(L^2(\R^d)))} := \mathrm{Tr}_{L^2(\R^d)} \bigl ( \hat{f}^{A \, \ast} \, \hat{g}^A \bigr )
\end{align*}
that is customarily chosen contains the operator adjoint of the first argument and defines a scalar product on $\mathfrak{L}^2 \bigl ( \mathcal{B}(L^2(\R^d) \bigr ) \bigr )$. 

Not just for the case under consideration, but much more broadly is the symbol of the Hilbert adjoint the complex conjugate of the symbol. 
\begin{proposition}
	Suppose $F \in \Cont^{\infty}(\Pspace)$ is a function quantizes to a bounded operator 
	\begin{align*}
		\Op^A(F) \in \mathcal{B} \Bigl ( \mathfrak{L}^2 \bigl ( \mathcal{B} \bigl ( L^2(\R^d) \bigr ) \bigr ) \Bigr ) 
		. 
	\end{align*}
	Then its adjoint $\Op^A(F)^* = \Op^A(\overline{F})$ in $\mathcal{B} ( \mathfrak{L}^2 ( \mathcal{B} ( L^2(\R^d) ) ) )$ is the magnetic pseudodifferential super operator associated to the complex conjugate function $F^*(X_L,X_R) := \overline{F(X_L,X_R)}$. 
\end{proposition}
\begin{proof}
	The building blocks of a magnetic pseudodifferential super operator is the action of the Weyl system. Since the magnetic Weyl system $w^A(X) \in \mathcal{B} \bigl (L^2(\R^d) \bigr )$ is bounded and the $p$-Schatten classes are ideals (\cf \eg \cite[Theorem~VI.19~(b) and VI.22~(a)]{Reed_Simon:M_cap_Phi_1:1972} for $p = 1 , 2$), we deduce 
	\begin{align*}
		w^A(X_L) \, \mathfrak{L}^2 \bigl ( \mathcal{B} \bigl ( L^2(\R^d) \bigr ) \bigr ) \, w^A(X_R) = \mathfrak{L}^2 \bigl ( \mathcal{B} \bigl ( L^2(\R^d) \bigr ) \bigr ) 
		. 
	\end{align*}
	A straightforward computation using the cyclicity of the trace shows that we can push the magnetic Weyl systems from the right argument to the left argument, at the expense of inverting $X_L$ and $X_R$: 
	\begin{align*}
		\bscpro{\hat{f}^A \, }{ \, w^A(X_L) \, \hat{g}^A \, w^A(X_R)}_{\mathfrak{L}^2(\mathcal{B}(L^2(\R^d)))} &= \mathrm{Tr}_{L^2(\R^d)} \Bigl ( \hat{f}^{A \, \ast} \, w^A(X_L) \, \hat{g}^A \, w^A(X_R) \Bigr ) 
		\\
		&= \mathrm{Tr}_{L^2(\R^d)} \Bigl ( w^A(X_R) \, \hat{f}^{A \, \ast} \, w^A(X_L) \, \hat{g}^A \Bigr ) 
		\\
		&= \mathrm{Tr}_{L^2(\R^d)} \Bigl ( \bigl ( w^A(-X_L) \, \hat{f}^A \, w^A(-X_R) \bigr )^* \, \hat{g}^A \Bigr ) 
		\\
		&= \bscpro{w^A(-X_L) \, \hat{f}^A \, w^A(-X_R) \, }{ \, \hat{g}^A}_{\mathfrak{L}^2(\mathcal{B}(L^2(\R^d)))}
	\end{align*}
	Once we combine that with 
	\begin{align*}
		\overline{\Fourier_{\Sigma} F(X_L,X_R)} &= \bigl ( \Fourier_{\Sigma} \overline{F} \bigr )(-X_L,-X_R)
		, 
	\end{align*}
	we see that the $-$ signs in the arguments cancel and we obtain the claim. 
\end{proof}
\begin{remark}
	The theorem can be generalized to suitable von Neumann sub alegebras $\mathcal{A} \subset \mathcal{B} \bigl ( L^2(\R^d) \bigr )$ of the bounded operators and f.n.s.\ traces. For instance, we must ensure 
	\begin{align*}
		w^A(X_L) \, \mathfrak{L}^2(\mathcal{A}) \, w^A(X_R) = \mathfrak{L}^2(\mathcal{A}) 
		&&
		\forall X_L , X_R \in \pspace 
		, 
	\end{align*}
	which is no longer obvious as $w^A(X_L)$ need not be an element of $\mathcal{A}$. For the same reason, we must assume that adjoining with the magnetic Weyl system is compatible with the trace, \ie the cyclicity condition 
	\begin{align*}
		\mathcal{T} \bigl ( w^A(X) \, \hat{f}^A \, w^A(X)^{-1} \bigr ) = \mathcal{T}(\hat{f}^A) 
		&&
		\forall X \in \pspace
		. 
	\end{align*}
\end{remark}
%

\subsection{The magnetic semi-super Weyl product} 
\label{rigorous_definition_supercalculus_on_S:semi_super_weyl_product}
Next up is the magnetic semi-super product, which describes the action of a magnetic pseudodifferential \emph{super} operator on a magnetic pseuododifferential operator that belongs to some $\mathfrak{L}^p \bigl ( \mathcal{B}(\Hil) \bigr )$.

\subsubsection{Definition of the product} 
\label{rigorous_definition_supercalculus_on_S:semi_super_weyl_product:definition}
To express the action of the Weyl product in a more explicit fashion, we write operators $\hat{g}^A = \op^A(g)$ we are interested as the magnetic Weyl quantization of a distribution $g \in \Schwartz'(\pspace)$. The semi-super product~\eqref{formal_super_calculus:eqn:definition_semisuper_product} then by definition replicates the operator product on the level of (generalized) functions on phase space. As usual we need to deal with the case where $g \in \Schwartz(\pspace)$ first, though, before we can extend these ideas to more general classes of functions and distributions. 
\begin{proposition}\label{magnetic_super_PsiDOs:prop:semi_super_product_S}
	Suppose Assumption~\ref{intro:assumption:polynomially_bounded_magnetic_field} on the magnetic field $B = \dd A$ and the vector potential $A$ are satisfied. Then the following holds: 
	\begin{enumerate}[(1)]
		\item The map $(F,f) \mapsto F \semisuper^B f$ gives rise to a continuous bilinear map from $\Schwartz(\Pspace) \times \Schwartz(\pspace)$ to $\Schwartz(\pspace)$.
		\item The magnetic semi-super product of two Schwartz functions $F \in \Schwartz(\Pspace)$ and $g \in \Schwartz(\pspace)$ is given in terms of the integral 
		\begin{align}
			F \semisuper^B g(X) = \frac{1}{(2\pi)^{3d}} \int_{\Pspace} \dd \Ybf \int_{\pspace} \dd Z \, &\e^{+ \ii \sigma(X , Y_L + Y_R + Z)} \, \e^{+\ii \frac{\eps}{2} \sigma(Y_L + Z , Y_R + Z)} \, \e^{- \ii \lambda \recfluxp(x,y_L,y_R,z)} \, \cdot \notag \\
			& (\Fourier_{\Sigma} F)(\Ybf) \, (\Fourier_{\sigma} g)(Z) 
			.
			\label{magnetic_super_PsiDOs:eqn:semi_super_product_integral_expression}
		\end{align}
		\item Alternatively, this integral can be recast into the form 
		\begin{align}
			F \semisuper^B g(X) = \frac{1}{(2\pi)^{3d}} \int_{\Pspace} \dd \Ybf \int_{\pspace} \dd Z \, &\e^{+ \ii (z \cdot \zeta + y_L \cdot \eta_L + y_R \cdot \eta_R)} \, \e^{- \ii \lambda \recfluxp(x,y_L,y_R,z)} \, \cdot  \notag \\
			& F \bigl ( x - \tfrac{\eps}{2} (y_R + z) , \xi - \eta_L , x + \tfrac{\eps}{2}(y_L + z) , \xi - \eta_R \bigr ) \, \cdot \notag \\
			& g \bigl ( x + \tfrac{\eps}{2}(y_L - y_R) , \xi - \zeta \bigr ) 
			. 
			\label{magnetic_super_PsiDOs:eqn:semi_super_product_integral_expression_no_Fourier}
		\end{align}
		\item For the special case $F = f_L \otimes f_R$, $f_L , f_R \in \Schwartz(\pspace)$ the semi-super product reduces to 
		\begin{align}
			F \semisuper^B g = f_L \weyl^B g \weyl^B f_R .
			\label{magnetic_super_PsiDOs:eqn:semisuper_product}
		\end{align}
	\end{enumerate}
\end{proposition}
\begin{proof}
	\begin{enumerate}[(1)]
		\item This is the content of Corollary~\ref{magnetic_super_PsiDOs:cor:kernel_map_semisuper_product}; its proof rests on the so-called kernel map that we will study in Section~\ref{rigorous_definition_supercalculus_on_S:semi_super_weyl_product:integral_kernel_map} below. 
		\item The proof of this result is a mere elaboration of the proofs of~\cite[eq.\ (22)]{Mueller:product_rule_gauge_invariant_Weyl_symbols:1999}, \cite[Proposition~13]{Mantoiu_Purice:magnetic_Weyl_calculus:2004} and~\cite[Theorem~2.2.17]{Lein:progress_magWQ:2010}. Combining with definition of the magnetic super Weyl quantization (Definition~\ref{magnetic_super_PsiDOs:defn:super_Weyl_system}) with the composition law between super and regular Weyl system (Lemma~\ref{magnetic_super_PsiDOs:lem:composition_super_Weyl_system_Weyl_system}), we obtain
		\begin{align*}
			\Op^A(F) \, \op^A(g) &= \frac{1}{(2\pi)^{3d}} \int_{\Pspace} \dd \Ybf \int_{\pspace} \dd Z \, (\Fourier_{\Sigma} F)(\Ybf) \, (\Fourier_{\sigma} g)(Z) \, W^A(\Ybf) \, w^A(Z) 
			\\
			&= \frac{1}{(2\pi)^{3d}} \int_{\Pspace} \dd \Ybf \int_{\pspace} \dd Z \, \e^{+ \ii \frac{\eps}{2}\sigma(Y_L+Z,Y_R+Z)} \, (\Fourier_{\Sigma} F)(\Ybf) \, (\Fourier_{\sigma} g)(Z) \, 
			\cdot \\
			&\qquad\qquad\qquad\quad\quad 
			\cocyp(Q;y_L,z) \, \cocyp(Q;y_L+z,y_R) \, w^A(Z+Y_L+Y_R) \\
			&= \frac{1}{(2\pi)^{3d}} \int_{\Pspace} \dd \Ybf \int_{\pspace} \dd Z \, \e^{+ \ii \frac{\eps}{2}\sigma(Z-Y_R,Z-Y_L)} \, (\Fourier_{\Sigma} F)(\Ybf) \, (\Fourier_{\sigma} g)(Z-Y_L-Y_R) \, 
			\cdot \\
			&\qquad\qquad\qquad\quad\quad 
			\cocyp(Q;y_L,z-y_L-y_R) \, \cocyp(Q;z-y_R,y_R) \, w^A(Z) 
			.
		\end{align*}
		The symbol of this operator can be computed by finding the kernel for the operator depending on the parameters $y_L$, $y_R$ and $Z$,
		\begin{align*}
			\hat{L}(y_L,y_R,Z) := \cocyp(Q;y_L,z-y_L-y_R) \, \cocyp(Q;z-y_R,y_R) \, w^A(Z) ,
		\end{align*}
		and then applying the magnetic Wigner transform~\eqref{magnetic_Weyl_calculus:eqn:magnetic_Wigner_transform} to it. Given $\varphi \in L^2(\R^d)$, we have
		\begin{align*}
			\bigl ( \hat{L} & (y_L,y_R,Z) \, \varphi \bigr )(q) \\
		 	&= \cocyp(\eps q;y_L,z-y_L-y_R) \, \cocyp(\eps q;z-y_R,y_R) \, 
			\e^{- \ii \eps \left ( q+\frac{z}{2}\right )\cdot \zeta} \, \e^{- \ii \frac{\lambda}{\eps}\Gamma^A([\eps q,\eps q+\eps z])} \, \varphi(q+z) 
			\\
			&= \int_{\R^d} \dd q' \, \cocyp(\eps(q'-z);y_L,z-y_L-y_R) \, \cocyp(\eps(q'-z);z-y_R,y_R) 
			\, \cdot \\
			& \quad \quad
			\e^{- \ii \eps \left ( q'-\frac{z}{2}\right )\cdot\zeta} \, \e^{- \ii \frac{\lambda}{\eps}\Gamma^A([\eps(q'-z),\eps q'])} \, \delta \bigl ( q'-(q+z) \bigr ) \, \varphi(q') \\
			&=: \frac{1}{(2\pi)^{d/2}} \int_{\R^d} \dd q' \, K_{\hat{L}}(y_L,y_R,Z;q,q') \, \varphi(q') .
		\end{align*}
		The magnetic Wigner transform~\eqref{magnetic_Weyl_calculus:eqn:magnetic_Wigner_transform} of the kernel $K_{\hat{L}}(y_L,y_R,Z;q,q')$ is
		\begin{align*}
			\wigner^A K_{\hat{L}} &(y_L,y_R,Z;\cdot,\cdot)(X) 
			\\
			&=
			\frac{1}{(2\pi)^{\nicefrac{d}{2}}} \int_{\R^d} \dd q \, \e^{- \ii q\cdot\xi} \, \e^{- \ii \frac{\lambda}{\eps}\Gamma^A\left ( [ x-\frac{\eps}{2}q , x+\frac{\eps}{2}q ] \right )} \, 
			K_{\hat{L}} \bigl ( y_L,y_R,Z ; \tfrac{x}{\eps} + \tfrac{q}{2} , \tfrac{x}{\eps} - \tfrac{q}{2} \bigr )
			\\
			&= \e^{+ \ii \sigma(X,Z)} \, \cocyp \bigl ( x-\tfrac{\eps}{2} z ; y_L , z - y_L - y_R \bigr ) \; \cocyp \bigl ( x-\tfrac{\eps}{2}z ; z-y_R,y_R \bigr ) 
			\\
			&=: L(y_L,y_R,Z;X) .
		\end{align*}
		By plugging this back into the above operator equation of $\Op^A(F) \, \op^A(g)$ we see that the operator $\Op^A(F) \, \op^A(g)$ is the magnetic Weyl quantization of the symbol 
		\begin{align}
			\frac{1}{(2\pi)^{3d}} &\int_{\Pspace} \dd \Ybf \int_{\pspace} \dd Z \, \e^{+ \ii \sigma(X,Z)} \, \e^{+ \ii \frac{\eps}{2}\sigma(Z-Y_R,Z-Y_L)} \, \cocyp \bigl ( x-\tfrac{\eps}{2}z;y_L,z-y_L-y_R \bigr ) \, \cdot
			\notag \\
			&\qquad \qquad
			\cocyp \bigl ( x-\tfrac{\eps}{2}z;z-y_R,y_R \bigr ) \, (\Fourier_{\Sigma} F)(\Ybf) \, (\Fourier_{\sigma} g)(Z-Y_L-Y_R) \notag \\
			&= \frac{1}{(2\pi)^{3d}} \int_{\Pspace} \dd \Ybf \int_{\pspace} \dd Z \, \e^{+ \ii \sigma(X,Y_L+Y_R+Z)} \, \e^{+ \ii \frac{\eps}{2}\sigma(Y_L+Z,Y_R+Z)} \, \cdot
			\notag \\
			&\qquad \qquad
			\cocyp \bigl ( x-\tfrac{\eps}{2}(z+y_L+y_R);y_L,z \bigr ) \, \cdot
			\notag \\
			&\qquad \qquad
			\cocyp \bigl ( x-\tfrac{\eps}{2}(z+y_L+y_R);z+y_L,y_R \bigr ) \, (\Fourier_{\Sigma} F)(\Ybf) \, (\Fourier_{\sigma} g)(Z) 
			.
			\label{magnetic_super_PsiDOs:eqn:proof_semi_super_product_symbol_before_plugging_in_magnetic_phase}
		\end{align}
		It follows from~\eqref{magnetic_Weyl_calculus:eqn:definition_scaled_magnetic_phase_factor} and~\eqref{formal_super_calculus:eqn:definition_magnetic_phase_factor_super_product} that we have
		\begin{align*}
			\e^{- \ii \lambda \recfluxp(x,y_L,y_R,z)} = \cocyp \bigl ( x-\tfrac{\eps}{2}(z+y_L+y_R);y_L,z \bigr ) \, \cocyp \bigl ( x-\tfrac{\eps}{2}(z+y_L+y_R);z+y_L,y_R \bigr ) .
		\end{align*}
		By using this we see that the symbol~\eqref{magnetic_super_PsiDOs:eqn:proof_semi_super_product_symbol_before_plugging_in_magnetic_phase} agrees with $F\semisuper^B g$ given in~\eqref{magnetic_super_PsiDOs:eqn:semi_super_product_integral_expression}, and hence we have $\Op^A(F)\op^A(g) = \op^A(F\semisuper^B g)$. 
		\item By writing the symplectic Fourier transforms in~\eqref{magnetic_super_PsiDOs:eqn:semi_super_product_integral_expression} explicitly and sorting all terms containing $\zeta$, $\eta_L$ and $\eta_R$, we see that $F \semisuper^B g(X)$ is equal to
		\begin{align*}
			\frac{1}{(2\pi)^{6d}} \int_{\Pspace} &\dd \mathbf{Y} \int_{\Pspace} \dd \mathbf{Y}' \int_{\pspace} \dd Z \int_{\pspace} \dd Z' 
			\e^{- \ii (x - z' + \frac{\eps}{2} (y_L - y_R)) \cdot \zeta} \, \e^{- \ii (x - y'_L - \frac{\eps}{2} (y_R + z)) \cdot \eta_L} \, 
			\cdot \\
			& \e^{- \ii (x - y'_R + \frac{\eps}{2} (y_L + z)) \cdot \eta_R} \, \e^{+ \ii (\xi-\zeta') \cdot z} \, \e^{+ \ii (\xi - \eta'_L) \cdot y_L} \, \e^{+ \ii (\xi - \eta'_R) \cdot y_R} \, \e^{- \ii \lambda \recfluxp(x,y_L,y_R,z)} \, F(\Ybf') \, g(Z') 
			.
		\end{align*}
		Integrating over $\zeta$, $\eta_L$ and $\eta_R$ and then integrating over $z'$, $y'_L$ and $y'_R$ yields
		\begin{align*}
			F &\semisuper^B g(X) 
			\\
			&= \frac{1}{(2\pi)^{3d}} \int_{\pspace} \dd y_L \, \dd \eta'_L \int_{\pspace} \dd y_R \, \dd \eta'_R \int_{\pspace} \dd z \, \dd \zeta' \, 
			\e^{+ \ii (\xi-\zeta') \cdot z} \, \e^{+ \ii (\xi-\eta'_L) \cdot y_L} \, \e^{+ \ii (\xi-\eta'_R) \cdot y_R} \, 
			\cdot \\
			&\qquad \e^{- \ii \lambda \recfluxp(x,y_L,y_R,z)} \, F \bigl ( x - \tfrac{\eps}{2} (y_R + z) , \eta'_L , x + \tfrac{\eps}{2} (y_L + z) , \eta'_R \bigr ) \, g \bigl ( x + \tfrac{\eps}{2} (y_L-y_R) , \zeta' \bigr ) 
			.
		\end{align*}
		Making the change of variables $\zeta = \xi - \zeta'$, $\eta_L = \xi - \eta'_L$ and $\eta_R = \xi - \eta'_R$, we get~\eqref{magnetic_super_PsiDOs:eqn:semi_super_product_integral_expression_no_Fourier}. 
		\item Once we plug $F = f_L \otimes f_R$ into equation~\eqref{magnetic_super_PsiDOs:eqn:proof_semi_super_product_symbol_before_plugging_in_magnetic_phase}, we get 
		\begin{align*}
			[ (f_L &\otimes f_R) \semisuper^B g ](X)
			\\
			&= \frac{1}{(2\pi)^{3d}} \int_{\Pspace} \dd \Ybf \int_{\pspace} \dd Z \, \e^{+ \ii \sigma(X,Y_L+Y_R+Z)} \, \e^{+ \ii \frac{\eps}{2}\sigma(Y_L+Z,Y_R+Z)} \, 
			\cdot \notag \\
			&\qquad \qquad
			\cocyp \bigl ( x-\tfrac{\eps}{2}(z+y_L+y_R);y_L,z \bigr ) \, \cocyp \bigl ( x-\tfrac{\eps}{2}(z+y_L+y_R);z+y_L,y_R \bigr ) \, 
			\cdot \notag \\
			&\qquad \qquad
			(\Fourier_{\sigma} f_L)(Y_L) \, (\Fourier_{\sigma} f_R)(Y_R) \, (\Fourier_{\sigma} g)(Z)
		\end{align*}
		And comparing that to the double magnetic Weyl product $f_L \weyl^B g \weyl^B f_R$, we see that the two integrals agree. 
	\end{enumerate}
\end{proof}
%

\subsubsection{The integral kernel map of the magnetic semi-super Weyl product} 
\label{rigorous_definition_supercalculus_on_S:semi_super_weyl_product:integral_kernel_map}
When extending ordinary magnetic Weyl calculus from Section~\ref{magnetic_Weyl_calculus} by duality, the so-called kernel map was tremendously useful. We shall make use of a slightly different kernel map here, namely $F \mapsto K_F^B$ that arises from 
\begin{align}
	\Int(K_F^B) \, g := F \semisuper^B g
	. 
	\label{rigorous_definition_supercalculus_on_S:eqn:definition_kernel_map_semi_super_product}
\end{align}
The integral map is defined in the obvious way, 
\begin{align}
	\bigl ( \Int(K) \, g \bigr )(X) := \int_{\pspace} \dd Y \, K(X,Y) \, g(Y) 
	, 
	&&
	g \in \Schwartz(\pspace)
	. 
	\label{rigorous_definition_supercalculus_on_S:eqn:definition_integral_map}
\end{align}
From equations~\eqref{magnetic_super_PsiDOs:eqn:semi_super_product_integral_expression} or \eqref{magnetic_super_PsiDOs:eqn:semi_super_product_integral_expression_no_Fourier} we have just proven we can not only infer the form of $K_F^B$, but that $F \mapsto K_F^B$ continuously maps Schwartz functions onto Schwartz functions. Before we proceed to prove these facts, we will first provide some additional context and show in what ways the kernel map is useful. 

Once we relate integral and kernel map on $\Schwartz$ to the magnetic super Weyl quantization 
\begin{align}
	\Op^A(F) \, \op^A(g) &= \op^A \bigl ( F \semisuper^B g \bigr )
	= \op^A \bigl ( \Int(K_F^B) \, g \bigr )
	\; \; \Longleftrightarrow \; \; 
	\Op^A(F) = \op^A \; \Int(K_F^B) \; {\op^A}^{-1} 
	,
	\label{magnetic_super_PsiDOs:eqn:writing_super_PsiDO_with_kernel_map}
\end{align}
their utility becomes immediately obvious: not only does it provides us with a way to dequantize magnetic Weyl super operators. But the duality bracket on $\pspace$ then gives a natural path to extend the semi-super product — and therefore magnetic super Weyl quantization — to larger classes of functions and distributions than $\Schwartz(\Pspace)$ and $\Schwartz(\pspace)$, respectively. Simply rewriting 
\begin{align*}
	\bigl ( F \semisuper^B g \, , \, h \bigr )_{\Schwartz(\pspace)} = \bigl ( \Int(K_F^B) g \, , \, h \bigr )_{\Schwartz(\pspace)} 
	= \bigl ( K_F^B \, , \, h \otimes g \bigr )_{\Schwartz(\Pspace)}
\end{align*}
explains why we may study $K_F^B$ rather than $g \mapsto F \semisuper^B g$. 

But we now have the added advantage that we may tap into the standard theory. For example, the Schwartz Kernel Theorem (\cf \cite[Section~51]{Treves:topological_vector_spaces:1967}) tells us that $\Int$ induces a topological vector space isomorphism from $\Schwartz'(\Pspace)$ to $\bounded \bigl ( \Schwartz(\pspace) , \Schwartz'(\pspace) \bigr )$. Combining this with the arguments from the proof of Proposition~\ref{magnetic_super_PsiDOs:prop:kernel_map_to_symbol_extension_distributions} below shows that, given a tempered distribution $F \in \Schwartz'(\Pspace)$, we obtain a continuous linear operator $\Int(K_F^B) : \Schwartz(\pspace) \longrightarrow \Schwartz'(\pspace)$. 

The starting point for the extension of $\semisuper^B$ is to make sure that kernels of Schwartz functions are also Schwartz class. 
\begin{proposition}\label{magnetic_super_PsiDOs:prop:kernel_map_semisuper_product}
	The explicit action of the map $F \mapsto K_F^B$ defined by the relation~\eqref{rigorous_definition_supercalculus_on_S:eqn:definition_kernel_map_semi_super_product} is
	\begin{align}
		K_F^B(X,Z) &= \frac{\e^{+ \ii \frac{2}{\eps} (x \cdot \xi - z \cdot \zeta)}}{(\pi \eps)^{3d}} \int_{\Pspace} \dd \Ybf \, \e^{+ \ii \frac{2}{\eps} y_R \cdot (\zeta - \xi)} \, \e^{+ \ii \frac{2}{\eps} y_L \cdot (\xi + \zeta)} \, \e^{-\ii\frac{2}{\eps}y_L\cdot\eta_L} \, \e^{-\ii\frac{2}{\eps}(y_L-z+x)\cdot\eta_R} \, 
		\cdot \notag \\
		&\qquad \qquad \,
		\e^{- \ii \lambda \recfluxp \left ( x , \frac{2}{\eps} y_L , \frac{2}{\eps} (x + y_L - z) , \frac{2}{\eps} (z - y_L - y_R) \right )} \, F(y_R , \eta_L , x + z - y_R , \eta_R) 
		,
		\label{magnetic_super_PsiDOs:eqn:kernel_semisuper_product_operator}
	\end{align}
	and gives rise to a topological vector space isomorphism from $\Schwartz(\Pspace)$ to itself. 
\end{proposition}
A quick corollary is the proof of Proposition~\ref{magnetic_super_PsiDOs:prop:semi_super_product_S}~(1) that we still owed to the readers, to the above we just need to add the observation that $\Schwartz(\Pspace) \times \Schwartz(\pspace) \ni (K,f) \mapsto \Int(K) \, f \in \Schwartz(\pspace)$ is a bilinear, continuous map. 
\begin{corollary}\label{magnetic_super_PsiDOs:cor:kernel_map_semisuper_product}
	The semisuper product $\semisuper^B : \Schwartz(\Pspace) \times \Schwartz(\pspace) \longrightarrow \Schwartz(\pspace)$ is a bilinear, continuous map. 
\end{corollary}
\begin{proof}[Proposition~\ref{magnetic_super_PsiDOs:prop:kernel_map_semisuper_product}]
	First, let us begin by proving that $F \mapsto K_F^B$ is a topological vector space isomorphism on $\Schwartz(\Pspace)$. And then in a second step we will connect the integral kernel $K_F^B$ to the semi-super product via $F \semisuper^B g = \Int(K_F^B) \, g$. 
	
	Let $F \in \Schwartz(\Pspace)$ be a Schwartz function, and set
	\begin{align}
	    G(\Zbf) &= \frac{\e^{- \frac{\ii}{\eps} (z_L \cdot \zeta_R + z_R \cdot \zeta_L)}}{(\pi\eps)^{3d}} \int_{\Pspace} \dd \Ybf \, \e^{+ \ii \frac{2}{\eps} \zeta_L \cdot y_L} \, \e^{+\ii\frac{2}{\eps}\zeta_R\cdot y_R} \, \e^{- \ii \frac{2}{\eps} y_L \cdot \eta_L} \, \e^{-\ii \frac{2}{\eps} (y_L - z_R) \cdot \eta_R} \, 
		\cdot \notag \\
		& \qquad \qquad \e^{-\ii \lambda \recfluxp \left ( \frac{1}{2} (z_L-z_R) , \frac{2}{\eps} y_L , \frac{2}{\eps} (y_L-z_R) , \frac{1}{\eps} (z_L+z_R-2y_L-2y_R) \right )} \, F(y_R , \eta_L , z_L - y_R , \eta_R) 
		.
		\label{eq:supercalculus.auxiliary-function-for-deriving-kernel}
	\end{align}
	Note that we have $G(x+z , \xi + \zeta , -x + z , -\xi + \zeta) = K_F^B(x,\xi,z,\zeta)$. Using this we can see that the map $F\mapsto K_F^B$ assigning the kernel to a symbol $F$ can be decomposed into the following seven topological vector space isomorphisms on the space of Schwartz functions.
	\begin{enumerate}[(1)]
		\item Compose $F$ with the invertible linear map $\Pspace \ni (y_R,\eta_L,z_L,\eta_R) \mapsto (y_R , \eta_L , z_L - y_R , \eta_R) \in \Pspace$.
		\item Take the partial Fourier transform of this function with respect to $\eta_L$ and $\eta_R$. Then we get the function,
		\begin{align*}
			(y_R,p,z_L,q) \longmapsto \frac{1}{(2\pi)^d} \int \dd \eta_L \dd \eta_R \, \e^{-\ii p\cdot\eta_L} \, \e^{-\ii q\cdot\eta_R} \, F(y_R,\eta_L,z_L-y_R,\eta_R) .
		\end{align*}
		\item Compose this function with the invertible linear map 
		\begin{align*}
			(z_L,y_L,z_R,y_R) \mapsto \bigl ( y_R , \tfrac{2}{\eps} y_L , z_L , \tfrac{2}{\eps}(y_L-z_R) \bigr )
			. 
		\end{align*}
		This yields the function,
		\begin{align*}
			(z_L,y_L,z_R,y_R) \longmapsto \frac{1}{(2\pi)^d} \int \dd \eta_L \dd \eta_R \, \e^{- \ii \frac{2}{\eps} y_L \cdot \eta_L} \, \e^{- \ii \frac{2}{\eps} (y_L - z_R) \cdot \eta_R} \, F(y_R , \eta_L , z_L - y_R , \eta_R) 
			.
		\end{align*}
		\item Multiply the magnetic flux factor,
		\begin{align}
			\e^{- \ii \lambda \recfluxp \left ( \frac{1}{2} (z_L - z_R) , \frac{2}{\eps} y_L , \frac{2}{\eps}(y_L - z_R) , \frac{1}{\eps}(z_L + z_R - 2y_L - 2y_R) \right )} 
			,
			\label{eq:supercalculus.magnetic-flux-factor}
		\end{align}
		which is given in~\eqref{eq:supercalculus.auxiliary-function-for-deriving-kernel}. Then we obtain
		\begin{align*}
			(z_L,y_L,z_R,y_R) \longmapsto \frac{1}{(2\pi)^d} \int \dd \eta_L \, \dd \eta_R & \; \e^{- \ii \frac{2}{\eps} y_L \cdot \eta_L} \, \e^{- \ii \frac{2}{\eps} (y_L - z_R) \cdot \eta_R} \, 
			\cdot \\
			& \e^{- \ii \lambda \recfluxp \left ( \frac{1}{2} (z_L - z_R) , \frac{2}{\eps} y_L , \frac{2}{\eps} (y_L - z_R) , \frac{1}{\eps} (z_L + z_R - 2y_L - 2y_R) \right )} \, 
			\cdot \\
			& F(y_R , \eta_L , z_L - y_R , \eta_R) 
			.
		\end{align*}
		\item Take the partial Fourier transform with respect to $y_L$ and $y_R$. Then we get the function,
		\begin{align*}
			(z_L,\zeta_L,z_R,\zeta_R) \longmapsto \frac{1}{(2\pi)^{2d}} \int \dd \Ybf \; &\e^{+ \ii \frac{2}{\eps} \zeta_L \cdot y_L} \, \e^{+ \ii \frac{2}{\eps} \zeta_R \cdot y_R} \, \e^{- \ii \frac{2}{\eps} y_L \cdot \eta_L} \, \e^{- \ii \frac{2}{\eps} (y_L - z_R) \cdot \eta_R} \, 
			\cdot \\
			& \e^{- \ii \lambda \recfluxp \left ( \frac{1}{2} (z_L - z_R) , \frac{2}{\eps} y_L , \frac{2}{\eps} (y_L - z_R) , \frac{1}{\eps} (z_L + z_R - 2y_L - 2y_R) \right )} \, 
			\cdot \\
			& F(y_R,\eta_L,z_L-y_R,\eta_R) .
		\end{align*}
		\item Multiply the function $\frac{2^{2d}}{\pi^d\eps^{3d}}\e^{-\frac{\ii}{\eps}(z_L\cdot\zeta_R+z_R\cdot\zeta_L)}$. We then obtain the function $G(\Zbf)$ in~\eqref{eq:supercalculus.auxiliary-function-for-deriving-kernel}.
		\item Compose $G$ with the invertible linear map $(x,\xi,z,\zeta)\mapsto(x+z,\xi+\zeta,-x+z,-\xi+\zeta)$. This yields the kernel $K_F^B$ given in~\eqref{magnetic_super_PsiDOs:eqn:kernel_semisuper_product_operator}.
	\end{enumerate}
	We know by Lemma~\ref{appendix:oscillatory_integrals:lem:magnetic_phase_factor_estimate} that the magnetic flux factor~\eqref{eq:supercalculus.magnetic-flux-factor} is an invertible function in $\smoothpoly(\R^{4d})$. Moreover, the function $\frac{2^{2d}}{\pi^d\eps^{3d}} \, \e^{-\frac{\ii}{\eps} (z_L \cdot \zeta_R + z_R \cdot \zeta_L)}$ is also invertible and belongs to $\smoothpoly(\Pspace)$. As the multiplication of an invertible $\smoothpoly$-class function induces a topological vector space isomorphism on the space of Schwartz functions, we see that the steps~(4) and~(6) induce topological vector space isomorphisms on the space of Schwartz functions. Since the partial Fourier transform gives rise to a topological vector space isomorphism on the space of Schwartz functions, we also see that the steps~(2) and~(5) induce topological vector space isomorphisms on the space of Schwartz functions. Furthermore, the coordinate transformation steps~(1), (3) and~(7) induce topological vector space isomorphisms on the space of Schwartz functions as well. All this shows that the map $F \mapsto K_F^B$ decomposes into isomorphisms on the space of Schwartz functions, and so it gives rise to a topological vector space isomorphism from $\Schwartz(\Pspace)$ to itself. 
	\medskip
	
	\noindent
	It remains to prove the connection $\Int(K_F^B) \, g = F \semisuper^B g$ between integral kernel map and semi-super product $\semisuper^B$. So let $F \in \Schwartz(\Pspace)$ and $g \in \Schwartz(\pspace)$ be Schwartz functions. Starting from~\eqref{magnetic_super_PsiDOs:eqn:semi_super_product_integral_expression_no_Fourier}, we can easily extract $K_F^B$ after a change of variables,
	\begin{equation}
	    ( Y_L, Y_R , Z ) \longmapsto \left( \tfrac{\eps}{2}y_L, \xi-\eta_L, x-\tfrac{\eps}{2}(z+y_R), \xi-\eta_R, x+\tfrac{\eps}{2}(y_L-y_R), \xi-\zeta \right)
		, 
		\label{rigorous_definition_supercalculus_on_S:eqn:change-of-variable-for-deriving-kernel}
	\end{equation}
	and see that the result agrees with \eqref{magnetic_super_PsiDOs:eqn:kernel_semisuper_product_operator}. The proof is complete. 
\end{proof}
%

\subsection{Magnetic super Wigner transform} 
\label{rigorous_definition_supercalculus_on_S:super_wigner_transform}
The introduction of the semi-super product and the kernel map $F \mapsto K_F^B$ pave the way for the analog of the Wigner transform. Recall that the map $F\mapsto K_F^B$~(\ref{magnetic_super_PsiDOs:eqn:kernel_semisuper_product_operator}) gives rise to a topological vector space isomorphism from $\Schwartz(\Pspace)$ to itself (\cf Proposition~\ref{magnetic_super_PsiDOs:prop:kernel_map_semisuper_product}).

\begin{definition}
We define the magnetic super Wigner transform $\Schwartz(\Pspace)\ni K\mapsto\Wigner^B \, K\in\Schwartz(\Pspace)$ as the inverse of the kernel map $F\mapsto K_F^B$ given in~(\ref{magnetic_super_PsiDOs:eqn:kernel_semisuper_product_operator}).
\end{definition}

As the magnetic super Wigner transform is the \emph{inverse of the kernel map} $F \mapsto K_F^B$, it conceptually does the same thing as the usual magnetic Wigner transform~\eqref{magnetic_Weyl_calculus:eqn:magnetic_Wigner_transform}. Thereby it provides a dequantization of superoperators, at least assuming for now that the super operator $\hat{F}^A$ has a Schwartz class operator kernel; we shall extend this relation later on in Section~\ref{super_calculus_extension_by_duality:super_weyl_quantization} to operators with distributional kernels.
\begin{proposition}\label{magnetic_super_PsiDOs:prop:super_Wigner_transform_Schwartz_functions}
	\begin{enumerate}[(1)]
		\item $\Wigner^B : \Schwartz(\Pspace) \longrightarrow \Schwartz(\Pspace)$ is the topological vector space isomorphism given by 
\begin{align}
	\Wigner^B \, K(\Xbf) &= \frac{\eps^{3d}}{2^{4d}\pi^d} \int_{\Pspace} \dd \Ybf \, \e^{+ \ii \xi_L \cdot y_L} \, \e^{+ \ii \xi_R \cdot y_R} \, \e^{+ \frac{\ii}{2} \left ( (-x_L + x_R) \cdot \eta_R - \frac{\eps}{2} (y_L + y_R) \cdot \eta_L \right )} \, 
	\cdot \notag \\
	&\qquad \qquad \qquad \quad 
	\e^{+ \ii \lambda \recfluxp \left ( \frac{1}{2} (x_L+x_R) - \frac{\eps}{4} (y_L - y_R) , y_L , y_R , \frac{1}{\eps}(x_R-x_L) - \frac{1}{2} (y_L + y_R) \right )} \, 
	\cdot \notag \\
	&\qquad \qquad \qquad \quad 
	K \bigl ( \tfrac{1}{2} (x_L + x_R) - \tfrac{\eps}{4} (y_L - y_R) , \tfrac{\eps}{4} (\eta_L - \eta_R) , 
	\bigr . \notag \\
	&\qquad \qquad \qquad \qquad \; \bigl . 
	\tfrac{1}{2}(x_L + x_R) + \tfrac{\eps}{4} (y_L - y_R) , \tfrac{\eps}{4} (\eta_L + \eta_R) \bigr )
	.
	\label{magnetic_super_PsiDOs:eqn:super_Wigner_transform}
\end{align}
		\item $\Wigner^B$ provides a way to dequantize super operators, namely 
		\begin{align*}
			\Op^A \bigl ( \Wigner^B K_{\hat{F}^A} \bigr ) = \hat{F}^A
			\; \; \Longleftrightarrow \; \; 
			{\Op^A}^{-1} \hat{F}^A = \Wigner^B K_{\hat{F}^A} 
			, 
		\end{align*}
		provided $\hat{F}^A$ is such that $K_{\hat{F}^A} \in \Schwartz(\Pspace)$. 
	\end{enumerate}
\end{proposition}
\begin{proof}
	\begin{enumerate}[(1)]
		\item It is immediate from the definition that the map $\Schwartz(\Pspace)\ni K\mapsto\Wigner^B K\in\Schwartz(\Pspace)$ is a topological vector space isomorphism.
		
		Furthermore, we have seen in the proof of Proposition~\ref{magnetic_super_PsiDOs:prop:kernel_map_semisuper_product} that the kernel map $F\mapsto K_F^B$ decomposes into seven topologival vector space isomorphisms on the space of Schwartz functions. Thus, we can explicitly compute the magnetic super Wigner transform by taking inverses of these seven isomorphisms. This yields the formula~\eqref{magnetic_super_PsiDOs:eqn:super_Wigner_transform}. 
		\item For Schwartz functions $F \in \Schwartz(\Pspace)$ the relation between $\Op^A(F) = \op^A \; \Int(K_F^B) \; {\op^A}^{-1}$ and its kernel $K^B_F$ is one-to-one (see Propositions~\ref{magnetic_super_PsiDOs:prop:semi_super_product_S} and \ref{magnetic_super_PsiDOs:prop:kernel_map_semisuper_product}). 
		
		Thus, it suffices to prove that it is the inverse of the kernel map $F \mapsto K_F^B$ given by~\eqref{magnetic_super_PsiDOs:eqn:kernel_semisuper_product_operator}. But that can be checked by direct computation. This finishes the proof. 
	\end{enumerate}
\end{proof}
There are some key conceptual differences, though. Most importantly, unlike its ordinary counterpart, it only depends on the magnetic field $B$ as opposed to the magnetic vector potential. The reason is that the ordinary dequantization ${\op^A}^{-1}$ maps operators onto tempered distributions. The $C^*$-algebraic structure of $\mathcal{B} \bigl ( \Hil^A \bigr )$ is inherited by the algebra of distributions 
\begin{align*}
	\mathfrak{C}^B := {\op^A}^{-1} \bigl ( \mathcal{B} \bigl ( L^2(\R^d) \bigr ) \bigr )
	, 
\end{align*}
where \eg the norm is $\norm{f}_{\mathfrak{C}^B} := \bnorm{\op^A(f)}_{\mathcal{B}(\Hil^A)}$. The conceptual advantage, though, is that gauge-covariance $\op^{A + \dd \chi}(f) = \e^{- \ii \chi(Q)} \, \op^A(f) \, \e^{+ \ii \chi(Q)}$ implies the pulled back algebra $\mathfrak{C}^B$ only depends on the magnetic field. Indeed, for suitable elements of $\mathfrak{C}^B$ the operator adjoint maps to complex conjugation and the product reduces to the magnetic Weyl product~\eqref{magnetic_Weyl_calculus:eqn:magnetic_Weyl_product_formula}. 

\subsection{The magnetic super Weyl product} 
\label{ssub:the_magnetic_super_weyl_product}
The utility of a magnetic super Weyl calculus is that the product of two magnetic super pseudodifferential operators yields another magnetic super pseudodifferential operator. 
\begin{proposition}
	Suppose Assumption~\ref{intro:assumption:polynomially_bounded_magnetic_field} on the magnetic field $B = \dd A$ and the vector potential $A$ are satisfied. Given any $F , G \in \Schwartz(\Pspace)$ the product of two magnetic super Weyl quantizations implicitly defined through $\Op^A(F) \, \Op^A(G) = \Op^A \bigl ( F \super^B G \bigr )$ is the Schwartz function given by 
	\begin{align}
		F \super^B G(\Xbf) = \frac{1}{(2\pi)^{4d}} \int_{\Pspace} \dd \Ybf & \int_{\Pspace} \dd \Zbf \; \e^{+ \ii \Sigma(\Xbf,\Ybf+\Zbf)} \, \e^{+\ii\frac{\eps}{2}\Sigma(r(\Ybf),\Zbf)} \, 
		\cdot \notag \\ 
		& \e^{- \ii \lambda \trifluxp(x_L,y_L,z_L)} \, \e^{- \ii \lambda \trifluxp(x_R,z_R,y_R)} \, (\Fourier_{\Sigma} F)(\Ybf) \, (\Fourier_{\Sigma} G)(\Zbf) 
		,
		\label{magnetic_super_PsiDOs:eqn:magnetic_super_Weyl_product}
	\end{align}
	where $r(\Ybf) \equiv r(Y_L,Y_R) := (Y_L,-Y_R)$. Furthermore, $\super^B$ gives rise to a continuous bilinear map from $\Schwartz(\Pspace)\times\Schwartz(\Pspace)$ to $\Schwartz(\Pspace)$.
\end{proposition}
\begin{proof}
	Let $F , G \in \Schwartz(\Pspace)$ be two Schwartz functions. Then the product of two magnetic super Weyl quantizations 
	\begin{align*}
		\Op^A(F) \, \Op^A(G) = \op^A \; \Int(K_F^B) \; \Int(K_G^B) \; {\op^A}^{-1} = \op^A \, \Int \bigl ( K_F^B \kerprod K_G^B \bigr ) \, {\op^A}^{-1} 
	\end{align*}
	can be recast in terms of the integral map $\Int$ for the kernels~\eqref{magnetic_super_PsiDOs:eqn:writing_super_PsiDO_with_kernel_map}. Recall that the product of two integral operators is the integral operator for the kernel $K_F^B \kerprod K_G^B \in \Schwartz(\Pspace)$ given by $K_F^B\kerprod K_G^B (X,Z) := \int_\pspace dY \, K_F^B(X,Y) \, K_G^B(Y,Z)$.
	
	The operator kernel and the symbol are related via the magnetic super Wigner transform, Proposition~\ref{magnetic_super_PsiDOs:prop:super_Wigner_transform_Schwartz_functions}, which means we can use 
	\begin{align*}
		F \super^B G &= \Wigner^B \bigl ( K_F^B \kerprod K_G^B \bigr ) 
	\end{align*}
	to find an explicit expression for the magnetic super Weyl product. Not only that, since the magnetic super Wigner transform maps Schwartz functions onto Schwartz functions (\cf Proposition~\ref{magnetic_super_PsiDOs:prop:super_Wigner_transform_Schwartz_functions}), we conclude that $F \super^B G \in \Schwartz(\Pspace)$ holds. 
	
	In principle, we could proceed and obtain an explicit expression of $F \super^B G$ from the kernel and the Wigner transform. But there is a second pathway, which mimics the proof of \cite[Theorem~2.10]{Lein:two_parameter_asymptotics:2008} where we compute $\super^B$ by hand. Since we have already established that $F \super^B G \in \Schwartz(\Pspace)$, we just note that all the integrals below exist in the $L^1$ sense with absolutely integrable integrands. 
	
	The product applied to some $\hat{h}^A \in \bounded(\Hil^A)$ is equal to
	\begin{align*}
		\Op^A(F) \, &\Op^A(G) \, \hat{h}^A
		\\
		&= 
		\frac{1}{(2\pi)^{4d}} \int_{\Pspace} \dd \Ybf \, (\Fourier_{\Sigma} F)(\Ybf) \, w^A(Y_L) \, \left ( \int_{\Pspace} \dd \Zbf \, (\Fourier_{\Sigma} G)(\Zbf) \, w^A(Z_L) \, \hat{h}^A \, w^A(Z_R) \right ) \, w^A(Y_R) 
		\\
		&= \frac{1}{(2\pi)^{4d}} \int_{\Pspace} \dd \Ybf \int_{\Pspace} \dd \Zbf \, (\Fourier_{\Sigma} F)(\Ybf) \, (\Fourier_{\Sigma} G)(\Zbf) \, w^A(Y_L) \, w^A(Z_L) \, \hat{h}^A \, w^A(Z_R) \, w^A(Y_R) 
		.
	\end{align*}
	The composition law~\eqref{magnetic_Weyl_calculus:eqn:composition_magnetic_Weyl_system} of the (ordinary) magnetic Weyl system allows us to rewrite the above as 
	\begin{align} 
		& \frac{1}{(2\pi)^{4d}} \int_{\Pspace} \dd \Ybf \int_{\Pspace} \dd \Zbf \, (\Fourier_{\Sigma} F)(\Ybf) \, (\Fourier_{\Sigma} G)(\Zbf) \, \e^{+ \ii \frac{\eps}{2} \sigma(Y_L,Z_L)} \, \cocyp (Q;y_L,z_L) \, w^A(Y_L+Z_L) \, 
		\cdot \notag \\ 
		& 
		\qquad \qquad \hat{h}^A \, \e^{+ \ii \frac{\eps}{2} \sigma(Z_R,Y_R)} \, \cocyp (Q;z_R,y_R) \, w^A(Y_R + Z_R) \,
		\notag \\
		&
		= \frac{1}{(2\pi)^{4d}} \int_{\Pspace} \dd \Ybf \int_{\Pspace} \dd \Zbf \, (\Fourier_{\Sigma} F)(\Ybf) \, (\Fourier_{\Sigma} G)(\Zbf-\Ybf) \, \e^{+ \ii \frac{\eps}{2} \sigma(Y_L,Z_L)} \, \cocyp (Q;y_L,z_L-y_L) \, w^A(Z_L) \, 
		\cdot \notag \\ 
		& 
		\qquad \qquad \hat{h}^A \, \e^{+ \ii \frac{\eps}{2} \sigma(Z_R,Y_R)} \, \cocyp (Q;z_R-y_R,y_R) \, w^A(Z_R)
		.
		\label{magnetic_super_PsiDOs:eqn:proof_magnetic_super_Weyl_product_after_composing_Weyl_systems}
	\end{align}
	The operators to the left and to the right of $\hat{h}^A$ can be written as the Weyl quantization of another function. That allows us to replace 
	\begin{align*}
		\cocyp &(Q;y_L,z_L-y_L) \, w^A(Z_L) 
		\\
		&= \frac{1}{(2\pi)^{2d}} \int_{\pspace} \dd Y_L' \int_{\pspace} \dd Z_L' \, \e^{+ \ii \sigma(Z_L',Y_L')} \, \e^{+ \ii\sigma(Y_L',Z_L)} \, \cocyp \bigl ( y_L'-\tfrac{\eps}{2} z_L ; y_L , z_L-y_L \bigr ) \, w^A(Z_L') 
	\end{align*}
	with an integral expression over a magnetic Weyl system; to verify that left- and right-hand side are indeed one and the same operator, we compare the actions of both on some $\varphi \in \Schwartz(\R^d) \subset L^2(\R^d)$. Analogously, we can replace operator to the right of $\hat{h}^A$ by 
	\begin{align*}
		\cocyp &(Q;z_R-y_R,y_R) \, w^A(Z_R) 
		\\
		 &= \frac{1}{(2\pi)^{2d}} \int_{\pspace} \dd Y_R' \int_{\pspace} \dd Z_R' \, \e^{+ \ii \sigma(Z_R',Y_R')} \, \e^{+ \ii \sigma(Y_R',Z_R)} \, \cocyp \bigl ( y_R'-\tfrac{\eps}{2}z_R ; z_R-y_R , y_R \bigr ) \, w^A(Z_R') 
		 .
	\end{align*}
	Therefore, after we combine the $\eps$-dependent phase factors 
	\begin{align*}
		\e^{+ \ii \frac{\eps}{2} \sigma(Y_L,Z_L)} \, \e^{+ \ii \frac{\eps}{2} \sigma(Z_R,Y_R)} &= \e^{+ \ii \frac{\eps}{2} \Sigma(r(\Ybf),\Zbf)} 
		, 
	\end{align*}
	we can bring the product into the form 
	\begin{align*}
		\Op^A(F) \, &\Op^A(G) \, \hat{h}^A 
		\\
		&= \frac{1}{(2\pi)^{8d}} \int_{\Pspace} \dd \Ybf \int_{\Pspace} \dd \Zbf \int_{\Pspace} \dd \Ybf' \int_{\Pspace} \dd \Zbf' \, 
		\cdot \\
		&\qquad \qquad 
		\e^{+ \ii \frac{\eps}{2} \Sigma(r(\Ybf),\Zbf)} \, \e^{+ \ii \sigma(Z_L',Y_L')} \, \e^{+ \ii \sigma(Y_L',Z_L)} \, \e^{+ \ii \sigma(Z_R',Y_R')} \, \e^{+ \ii \sigma(Y_R',Z_R)} \, 
		\cdot \\
		&\qquad \qquad
		\cocyp \bigl ( y_L'-\tfrac{\eps}{2}z_L ; y_L , z_L-y_L \bigr ) \, \cocyp \bigl ( y_R'-\tfrac{\eps}{2}z_R ; z_R-y_R , y_R \bigr ) \, 
		\cdot \\
		&\qquad \qquad 
		(\Fourier_{\Sigma} F)(\Ybf) \, (\Fourier_{\Sigma} G)(\Zbf-\Ybf) \, w^A(Z_L') \, \hat{h}^A \, w^A(Z_R')
		\\
		&= \frac{1}{(2\pi)^{2d}} \int_{\Pspace} \dd \Zbf' \, \biggl ( \frac{1}{(2\pi)^{6d}} \int_{\Pspace} \dd \Ybf \int_{\Pspace} \dd \Zbf \int_{\Pspace} \dd \Ybf' \, \e^{+ \ii \Sigma(\Zbf',\Ybf')} \, \e^{+ \ii \Sigma(\Ybf',\Zbf)} \, \e^{+ \ii \frac{\eps}{2} \Sigma(r(\Ybf),\Zbf)} \, 
		\cdot \\
		&\qquad \qquad 
		\cocyp \bigl ( y_L'-\tfrac{\eps}{2}z_L ; y_L , z_L-y_L \bigr ) \, \cocyp \bigl ( y_R'-\tfrac{\eps}{2}z_R ; z_R-y_R , y_R \bigr ) \, 
		\cdot \\
		&\qquad \qquad 
		(\Fourier_{\Sigma} F)(\Ybf) \, (\Fourier_{\Sigma} G)(\Zbf-\Ybf) \biggr ) \, w^A(Z_L') \, \hat{h}^A \, w^A(Z_R') 
		.
	\end{align*}
	The inner integral can be recognized as $\bigl ( \Fourier_{\Sigma}(F\super^B G) \bigr )(\Zbf')$. Thus, we obtain the final form~\eqref{magnetic_super_PsiDOs:eqn:magnetic_super_Weyl_product} of the product formula, 
	\begin{align*}
		F \super^B G(\Xbf) &= \frac{1}{(2\pi)^{8d}} \int_{\Pspace} \dd \Ybf \int_{\Pspace} \dd \Zbf \int_{\Pspace} \dd \Ybf' \int_{\Pspace} \dd \Zbf' \, \e^{+ \ii \Sigma(\Xbf,\Zbf')} \, \e^{+ \ii \Sigma(\Zbf',\Ybf')} \, \e^{+ \ii \Sigma(\Ybf',\Zbf)} \, \e^{+ \ii \frac{\eps}{2}\Sigma(r(\Ybf),\Zbf)} \, 
		\cdot \\
		&\qquad \qquad 
		\cocyp \bigl ( y_L'-\tfrac{\eps}{2}z_L ; y_L , z_L-y_L \bigr ) \, \cocyp \bigl ( y_R'-\tfrac{\eps}{2}z_R ; z_R-y_R , y_R \bigr ) \, 
		\cdot \\
		&\qquad \qquad 
		(\Fourier_{\Sigma} F)(\Ybf)\, (\Fourier_{\Sigma} G)(\Zbf-\Ybf) 
		\\
		&= \frac{1}{(2\pi)^{4d}} \int_{\Pspace} \dd \Ybf \int_{\Pspace} \dd \Zbf \, \e^{+ \ii \Sigma(\Xbf,\Zbf)} \, \e^{+ \ii \frac{\eps}{2}\Sigma(r(\Ybf),\Zbf)} \, \cocyp \bigl ( x_L-\tfrac{\eps}{2}z_L ; y_L , z_L-y_L \bigr ) \, 
		\cdot \\
		&\qquad \qquad
		\cocyp \bigl ( x_R-\tfrac{\eps}{2}z_R ; z_R-y_R , y_R \bigl ) \, (\Fourier_{\Sigma} F)(\Ybf)\, (\Fourier_{\Sigma} G)(\Zbf-\Ybf) 
		\\
		&= \frac{1}{(2\pi)^{4d}} \int_{\Pspace} \dd \Ybf \int_{\Pspace} \dd \Zbf \, \e^{+ \ii \Sigma(\Xbf,\Ybf+\Zbf)} \, \e^{+ \ii \frac{\eps}{2} \Sigma(r(\Ybf),\Zbf)} \e^{- \ii \lambda \trifluxp(x_L,y_L,z_L)} \, 
		\cdot \\
		&\qquad \qquad 
		\e^{- \ii \lambda \trifluxp(x_R,z_R,y_R)} \, (\Fourier_{\Sigma} F)(\Ybf) \, (\Fourier_{\Sigma} G)(\Zbf) 
		,
	\end{align*}
	which completes the proof. 
\end{proof}
%

\subsection{Changes of representation} 
\label{rigorous_definition_supercalculus_on_S:changes_of_representation}
We will close this section with an important consideration. In quantum mechanics, physicists interpret unitaries are facilitating changes of representation. The unitarity implies that probabilities computed in either representation necessarily need to coincide, and therefore, both representations provide equivalent descriptions of the physics. Consequently, all essential properties of a system are those, which are independent of the choice of representation. 

One such property are the commutation relations~\eqref{magnetic_Weyl_calculus:eqn:commutation_relations} between position $Q_j$ and magnetic momentum $P^A_j$. Indeed, continuing the discussion we have started in Section~\ref{rigorous_definition_supercalculus_on_S:construction_quantization:super_weyl_system}, if 
\begin{align*}
	U : L^2(\R^d) \longrightarrow \Hil
\end{align*}
is any unitary and the position and momentum operators $(Q_U,P^A_U) = \bigl ( \Ad_U(Q) , \Ad_U(P^A) \bigr )$ in the new representation are given by \eqref{rigorous_definition_supercalculus_on_S:eqn:unitarily_equivalent_building_block_operators}. Then the commutation relations of $Q_U$ and $P^A_U$ are identical to those of $Q$ and $P^A$, namely those in \eqref{magnetic_Weyl_calculus:eqn:commutation_relations}. And the punchline of this subsection is that the pseudodifferential calculus, specifically the products $\weyl^B$, $\semisuper^B$ and $\super^B$, are independent of the choice of representation. 
\begin{theorem}\label{rigorous_definition_supercalculus_on_S:thm:products_independent_of_representation}
	The products $\weyl^B$, $\semisuper^B$ and $\super^B$ are independent of the choice of representation. 
\end{theorem}
Potential choices for $U$ could be the continuous Fourier transform $\Fourier$ to change to momentum representation or a change of gauge $\e^{+ \ii \lambda \chi(Q)}$. 

While at the end of the day our observations are somewhat trivial, this has not always been exploited in the literature. E.~g.~in \cite[Appendix~A]{PST:effective_dynamics_Bloch:2003} unnecessarily introduces a pseudodifferential calculus on weighted Hörmander classes by hand, where position and momentum have traded places, since they work in the momentum representation where the position operator $Q = \ii \nabla_k$ is given by a derivative and momentum $P = \hat{k}$ is a multiplication. Instead, as explained in \cite[Section~2.2]{DeNittis_Lein:Bloch_electron:2009} the (magnetic) pseudodifferential calculus only depends on the commutation relations and Panati et al.\ could have based their calculus on ordinary Hörmander classes instead. The point of this subsection is to extend the arguments from \cite[Section~2.2]{DeNittis_Lein:Bloch_electron:2009} to the pseudodifferential calculus of super operators constructed here. 

Assume we are given some unitary $U : L^2(\R^d) \longrightarrow \Hil$. Then adjoining with the unitary 
\begin{align*}
	\Ad_U : \mathcal{B} \bigl ( L^2(\R^d) \bigr ) \longrightarrow \mathcal{B}(\Hil)
\end{align*}
defines an isomorphism between the Banach spaces of operators on the respective Hilbert spaces. 

Moreover, because unitaries map orthonormal bases onto orthonormal bases, $\Ad_U$ relates the two traces, 
\begin{align*}
	\mathrm{Tr}_{\Hil} = \mathrm{Tr}_{L^2(\R^d)} \circ \Ad_U^{-1} 
	= \mathrm{Tr}_{L^2(\R^d)} \circ \Ad_{U^{-1}} 
	. 
\end{align*}
Consequently, $\Ad_U$ restricts to isomorphisms between the $p$-Schatten classes 
\begin{align*}
	\Ad_U : \mathfrak{L}^p \bigl ( \mathcal{B} \bigl ( L^2(\R^d) \bigr ) \bigr ) \longrightarrow \mathfrak{L}^p \bigl ( \mathcal{B}(\Hil) \bigr ) 
	. 
\end{align*}
Operators between these non-commutative $L^p$ spaces now transform via the “iterated adjoin operation”, and we set 
\begin{align}
	\Op^A_U = \Ad_{\Ad_U} \, \Op^A 
	= \Ad_U \, \Op^A( \, \cdot \, ) \, \Ad_U^{-1} 
	. 
	\label{rigorous_definition_supercalculus_on_S:eqn:transformation_super_Weyl_quantization}
\end{align}
To derive this from first principles is to see how the magnetic super Weyl system~\eqref{formal_super_calculus:eqn:super_Weyl_system} transforms and then define $\Op^A_U$ from equation~\eqref{formal_super_calculus:eqn:Op_A_formal_definition} and the transformed magnetic super Weyl system 
\begin{align*}
	W^A_U(\Xbf) = \Ad_{\Ad_U} \, W^A(\Xbf) = \Ad_U \, W^A(\Xbf) \, \Ad_U^{-1} 
	. 
\end{align*}
Fortunately, we have already verified the above relationship in equation~\eqref{rigorous_definition_supercalculus_on_S:eqn:super_Weyl_system_new_representation}. 
\begin{proof}[Theorem~\ref{rigorous_definition_supercalculus_on_S:thm:products_independent_of_representation}]
	To showcase how we can deduce that \eg the products $\weyl^B$, $\semisuper^B$ and $\super^B$ do not depend on our choice of representation, we will content ourselves proving this only for the magnetic super Weyl product $\super^B$. The other proofs are analogous. 
	
	Let $F , G \in \Schwartz(\Pspace)$ be two Schwartz functions; by Lemma~\ref{rigorous_definition_supercalculus_on_S:lem:properties_super_Weyl_quantization} the associated magnetic pseudodifferential super operators $\Op^A(F)$ and $\Op^A(G)$ define bounded maps between the $p$-Schatten classes. Let us pretend that there are two super Weyl products $\super^B_U$ and $\super^B$, and then prove with the following computation that they are in fact one and the same: 
	\begin{align}
		\Op_U^A(F \super^B_U G) &= \Op^A_U(F) \, \Op^A_U(G) 
		= \Ad_U \, \Op^A(F) \, \Ad_U^{-1} \, \Ad_U \, \Op^A(G) \, \Ad_U^{-1} 
		\notag 
		\\
		&= \Ad_U \, \Op^A(F \super^B G) \, \Ad_U^{-1} 
		= \Op^A_U(F \super^B G)
		.
		\label{rigorous_definition_supercalculus_on_S:eqn:super_B_independent_of_representation}
	\end{align}
	The above holds in the sense of bounded operators on $\mathcal{B}(\Hil)$ or $\mathfrak{L}^p \bigl ( \mathcal{B}(\Hil) \bigr )$ for any $1 \leq p < \infty$. This finishes the proof. 
\end{proof}
\begin{remark}
	Theorem~\ref{rigorous_definition_supercalculus_on_S:thm:products_independent_of_representation} can be generalized in any number of ways. For example, in the next two sections, we will extend $\semisuper^B$ and $\super^B$ to tempered distributions and Hörmander class symbols, and the above computation, suitably reinterpreted, applies. 
	
	Likewise, we can consider von Neumann sub algebras $\mathcal{A} \subseteq \mathcal{B} \bigl ( L^2(\R^d) \bigr )$: if $\super^B_U = \super^B$ holds on a larger set, then it also holds on a subset as well. 
	
	Lastly, with more effort, we can replace the usual trace $\mathrm{Tr}_{L^2(\R^d)}$ with other f.n.s.\ traces. The difficulty then consists of proving that \eg for suitable functions or distributions $F$ the operator $\Op^A(F)$ defines a map between the non-commutative $L^p$ spaces. Once that has been ensured, computation~\eqref{rigorous_definition_supercalculus_on_S:eqn:super_B_independent_of_representation} holds verbatim in the strong sense on $\mathfrak{L}^p(\mathcal{A})$. 
\end{remark}
%
\section{Extension of super calculus by duality} 
\label{super_calculus_extension_by_duality}
The extension by duality is now rather straight-forward. We mostly adapt the strategy nicely outlined in \cite{Mantoiu_Purice:magnetic_Weyl_calculus:2004}: we extend the kernel map and its inverse to tempered distributions, and introduce the notion of Moyal space and Moyal algebra. Even when we deviate, we still employ the same ideas.

\subsection{The magnetic super Weyl quantization and the magnetic super Wigner transform} 
\label{super_calculus_extension_by_duality:super_weyl_quantization}
There were two ways essentially to extend ordinary (magnetic) Weyl calculus: either one looks at matrix elements $\scpro{\varphi \, }{ \, \op^A(f) \psi}$ for Schwartz functions (see \eg \cite{Cornean_Helffer_Purice:simple_proof_Beals_criterion_magnetic_PsiDOs:2018}); or one uses the kernel map as in \cite{Mantoiu_Purice:magnetic_Weyl_calculus:2004}. We will follow the latter approach. 

When we apply $\Op^A(F)$ to operators $\hat{g}^A = \op^A(g)$ that are the magnetic Weyl quantization of $g \in \Schwartz(\pspace)$, we can use the kernel map 
\begin{align*}
	\Op^A(F) \, \op^A(g) &= \op^A \bigl ( \Int(K_F^B) \, g \bigr ) 
\end{align*}
to describe the action. Hence, the first step is to extend the kernel map. 
\begin{proposition}\label{magnetic_super_PsiDOs:prop:kernel_map_to_symbol_extension_distributions}
	The kernel map $F \mapsto K^B_F$ defined through~\eqref{magnetic_super_PsiDOs:eqn:kernel_semisuper_product_operator} extends to a continuous vector space isomorphism 
	\begin{align*}
		\Schwartz'(\Pspace) \longrightarrow \Schwartz'(\Pspace)
	\end{align*}
	on the tempered distributions. 
\end{proposition}
\begin{proof}
	Peeking at the proof of Proposition~\ref{magnetic_super_PsiDOs:prop:kernel_map_semisuper_product} we remind ourselves that the kernel map is a combination of linear, continuous, invertible variable transformations, multiplication with $\Cont^{\infty}_{\mathrm{pol}}$ functions and partial Fourier transform. Each of these operations extends to a continuous map on $\Schwartz'(\Pspace)$ and as a composition of continuous maps also $F \mapsto K_F^B$ is a continuous map.

	Moreover, all of the inverse maps are continuous as well so that $F \mapsto K_F^B$ is a continuous linear isomorphism between distributions.
\end{proof}
Since the magnetic super Wigner transform is nothing but the inverse of the kernel map, we have also just extended the magnetic super Wigner transform.
\begin{corollary}
	The magnetic super Wigner transform~\eqref{magnetic_super_PsiDOs:eqn:super_Wigner_transform} extends to a continuous linear isomorphism
	\begin{align*}
		\Wigner^B : \Schwartz'(\Pspace) \longrightarrow \Schwartz'(\Pspace)
		.
	\end{align*}
\end{corollary}
Recall that to every $K\in\Schwartz'(\Pspace)$ we can associate a continuous linear operator $\Int(K) : \Schwartz(\pspace)\rightarrow\Schwartz'(\pspace)$, where $\Int(K)$ is defined by
\begin{equation*}
    \bigl ( \Int(K)g \, , \, h \bigr )_{\Schwartz(\pspace)} := \bigl ( K \, , \, h\otimes g \bigr )_{\Schwartz(\Pspace)} \qquad \forall g,h\in\Schwartz(\pspace) .
\end{equation*}
This generalizes the integral operator given in~\eqref{rigorous_definition_supercalculus_on_S:eqn:definition_integral_map}. We know by the Schwartz kernel theorem (\cf \cite[Section 51]{Treves:topological_vector_spaces:1967}) that $\Int$ induces a topological vector space isomorphism from $\Schwartz'(\Pspace)$ to $\bounded \bigl ( \Schwartz(\pspace) , \Schwartz'(\pspace) \bigr )$. Combining this with Proposition~\ref{magnetic_super_PsiDOs:prop:kernel_map_to_symbol_extension_distributions} shows that, given $F \in \Schwartz'(\Pspace)$, we obtain a continuous linear operator $\Int(K_F^B) : \Schwartz(\pspace) \longrightarrow \Schwartz'(\pspace)$. With this in hand, the basis for the extension of $\Op^A(F)$ is the equation~\eqref{rigorous_definition_supercalculus_on_S:eqn:definition_kernel_map_semi_super_product}. Namely, for $F \in \Schwartz'(\Pspace)$ and $g\in\Schwartz(\pspace)$ we define a continuous linear map $g\mapsto F\semisuper^B g$ from $\Schwartz(\pspace)$ to $\Schwartz'(\pspace)$ by letting
\begin{align*}
    F\semisuper^B g := \Int(K_F^B) \, g .
\end{align*}
%

\subsection{The magnetic semi-super product} 
\label{super_calculus_extension_by_duality:semi_super_product}
Throughout the paper, given a function $G\in\Schwartz(\Pspace)$, we shall denote by $G^{\mathrm{t}}$ the function 
\begin{align*}
	G^{\mathrm{t}}(X_L,X_R) := G(X_R,X_L) \qquad \forall X_L,X_R\in\pspace .
\end{align*}
Given $F\in\Schwartz'(\Pspace)$, let $F^{\mathrm{t}}$ be the extension of this transpose map to $\Schwartz'(\Pspace)$ via 
	\begin{align*}
		( F^{\mathrm{t}} , G )_{\Schwartz(\Pspace)} := (F , G^{\mathrm{t}})_{\Schwartz(\Pspace)} 
		&&
		\forall G \in \Schwartz(\Pspace) 
		. 
	\end{align*}
Since $F^{\mathrm{t}}\in\Schwartz'(\Pspace)$, the above arguments show that $h\mapsto F^{\mathrm{t}}\semisuper^B h$ gives rise to a continuous linear map from $\Schwartz(\pspace)$ to $\Schwartz'(\pspace)$.
\begin{definition}[Magnetic semi-super Moyal space]\label{magnetic_super_PsiDOs:defn:semi_super_Moyal_space}
	We say $F \in \Schwartz'(\Pspace)$ belongs to the semi-super Moyal space $\ssMoyalSpace \subseteq \Schwartz'(\Pspace)$ if
	\begin{align*}
		h \mapsto F^{\mathrm{t}} \semisuper^B h
	\end{align*}
	induces a linear, continuous homomorphism from $\Schwartz(\pspace)$ to \emph{itself}. Then for all $F \in \ssMoyalSpace$ and $g \in \Schwartz'(\pspace)$ we set
	\begin{align}
		\bigl ( F \semisuper^B g \, , \, h \bigr )_{\Schwartz(\pspace)} := \bigl ( g \, , \, F^{\mathrm{t}} \semisuper^B h \bigr )_{\Schwartz(\pspace)}
		&&
		\forall h \in \Schwartz(\pspace)
		.
		\label{magnetic_super_PsiDOs:eqn:extension_semi_super_product}
	\end{align}
\end{definition}
\begin{remark}\label{magnetic_super_PsiDOs:rem:remark:topology_semi_super_Moyal_space}
	If $B$ is a bounded subset of $\Schwartz(\pspace)$ and $F \in \ssMoyalSpace$, then since $h \mapsto F^{\mathrm{t}} \semisuper^B h$ induces a continuous linear map from $\Schwartz(\pspace)$ to itself, the set $\bigl \{ F^{\mathrm{t}} \semisuper^B h \; \vert \; h \in B \bigr \}$ is bounded in $\Schwartz(\pspace)$ as well. Combining this with the definition~\eqref{magnetic_super_PsiDOs:eqn:extension_semi_super_product} we get $\sup_{h\in B}|(F\semisuper^B g , h )| = \sup_{h\in B}|(g , F^{\mathrm{t}}\semisuper^B h)|<\infty$, which shows that $g \mapsto F \semisuper^B g$ gives rise to a continuous linear map from $\Schwartz'(\pspace)$ to itself with respect to the topology of uniform convergence on bounded subsets.
\end{remark}
Naturally, the semi-super Moyal space contains the standard Hörmander symbol spaces as well as those defined in Definition~\ref{symbol_super_calculus:defn:Hoermander_super_symbols}. 
\begin{lemma}\label{super_calculus_extension_by_duality:lem:Hoermander_symbol_classes_contained_in_semi_super_Moyal_space}
	Suppose the magnetic field is polynomially bounded in the sense of Assumption~\ref{intro:assumption:polynomially_bounded_magnetic_field}, $0 \leq \rho \leq 1$ and $0 < \eps , \lambda \leq 1$. Then the following holds: 
	\begin{enumerate}[(1)]
	    \item For any $m \in \R$, we have $S_{\rho,0}^m(\Pspace) \subset \ssMoyalSpace$.
	    \item For any $m_L , m_R \in \R$, we have $S_{\rho,0}^{m_L,m_R}(\Pspace) \subset \ssMoyalSpace$.
	\end{enumerate}
\end{lemma}
Because the proof is quite laborious, yet standard, we have relegated it to Appendix~\ref{appendix:oscillatory_integrals:semi_super_product}. 

\subsection{The magnetic super Weyl product} 
\label{super_calculus_extension_by_duality:super_Weyl_product}
The last ingredient is an extension of the product itself. Here, we follow the playbook of \cite{Mantoiu_Purice:magnetic_Weyl_calculus:2004}. We exploit the following to extend the magnetic super Weyl product via the duality bracket: 
\begin{lemma}
	Suppose $F , G , H \in \Schwartz(\Pspace)$ are Schwartz functions. Then we the following holds: 
	\begin{enumerate}[(1)]
		\item $\displaystyle \int_{\Pspace} \dd \Xbf \, F \super^B G(\Xbf) = \int_{\Pspace} \dd \Xbf \, F(\Xbf)\, G(\Xbf) = {\sscpro {\bar{F}} G}_{L^2(\Pspace)} = (F,G)$
		\item $\bigl ( F \super^B G \, , \, H \bigr ) = \bigl ( F \, , \, G \super^B H \bigr ) = \bigl ( H \super^B F \, , \, G \bigr )$
	\end{enumerate}
\end{lemma}
\begin{proof}
	Clearly, (2) is a direct consequence of (1). 
	
	So let us turn to the proof of (1): we plug in the explicit expression~\eqref{magnetic_super_PsiDOs:eqn:magnetic_super_Weyl_product} for the magnetic super Weyl product and first integrate over $\Xbf$, which produces $(2\pi)^{4d} \, \delta(\Ybf + \Zbf)$, 
	\begin{align*}
		\int_{\Pspace} \dd \Xbf \, F \super^B G(\Xbf) &= \frac{1}{(2\pi)^{4d}} \int_{\Pspace} \dd \Xbf \int_{\Pspace} \dd \Ybf \int_{\Pspace} \dd \Zbf \, \e^{+ \ii \Sigma(\Xbf,\Ybf+\Zbf)} \, \e^{+ \ii \frac{\eps}{2} \Sigma(r(\Ybf),\Zbf)} \, 
		\cdot \\
		&\qquad \qquad \cocyp \bigl ( x_L - \tfrac{\eps}{2}(y_L+z_L) , y_L , z_L \bigr ) \; \cocyp \bigl ( x_R-\tfrac{\eps}{2}(y_R+z_R) , z_R , y_R \bigr) \, 
		\cdot \\
		&\qquad \qquad (\Fourier_{\Sigma} F)(\Ybf) \, (\Fourier_{\Sigma} G)(\Zbf)
		\\
		&= \int_{\Pspace} \dd \Ybf \, \e^{+ \ii \frac{\eps}{2} \Sigma(r(\Ybf),-\Ybf)} \, \cocyp \bigl ( x_L - \tfrac{\eps}{2}(y_L-y_L) , y_L , -y_L \bigr ) \,
 		\cdot \\
		&\qquad \qquad \cocyp \bigl ( x_R - \tfrac{\eps}{2} (y_R - y_R) , -y_R , y_R \bigr ) \; (\Fourier_{\Sigma} F)(\Ybf) \; (\Fourier_{\Sigma} G)(-\Ybf) 
		. 
	\end{align*}
	Because the two magnetic flux triangles collapse and have zero area, the two magnetic phase factors are in fact $1$, and after writing out the symplectic Fourier transform and integrating over the free variables, we are left with 
	\begin{align*}
		\ldots &= \int_{\Pspace} \dd \Ybf \, (\Fourier_{\Sigma} F)(\Ybf)\, (\Fourier_{\Sigma} G)(-\Ybf) 
		\\
		&= \frac{1}{(2\pi)^{4d}} \int_{\Pspace} \dd \Ybf \int_{\Pspace} \dd \Ybf' \int_{\Pspace} \dd \Zbf' \, \e^{+ \ii \Sigma(\Ybf,\Ybf')} \, \e^{+ \ii \Sigma(-\Ybf,\Zbf')} \, F(\Ybf')\, G(\Zbf') 
		\\
		&= \int_{\Pspace} \dd \Ybf \, F(\Ybf) \, G(\Ybf) = (F,G) 
		= {\sscpro{\bar{F}} G}_{L^2(\Pspace)} 
		.
	\end{align*}
	This finishes the proof. 
\end{proof}
The extension by duality is now straightforward. 
\begin{definition}[Extension via duality]
	We extend the magnetic super Weyl product of $F \in \Schwartz'(\Pspace)$ and $G \in \Schwartz(\Pspace)$ by setting 
	\begin{align*}
		\bigl (F \super^B G \, , \, H \bigr )_{\Schwartz(\Pspace)} := \bigl ( F \, , \, G \super^B H \bigr )_{\Schwartz(\Pspace)} 
		&&
		\forall H \in \Schwartz(\Pspace) 
		, 
		\\
		\bigl ( G \super^B F \, , \, H \bigr )_{\Schwartz(\Pspace)} := \bigl ( F \, , \, H \super^B G \bigr )_{\Schwartz(\Pspace)} 
		&&
		\forall H \in \Schwartz(\Pspace) 
		.
	\end{align*}
\end{definition}
Of course, we are primarily interested in subspaces of distributions, which form an algebra with respect to the super Weyl product. Unlike for the semi-super product, we need to distinguish between distributions that are nicely behaved under multiplication from the left and from the right. 
\begin{definition}[Magnetic super Moyal algebra]
	We define
	\begin{align*}
		\MoyalAlgebra_L^B(\Pspace) &:= \bigl \{ F\in\Schwartz'(\Pspace) \; \; \vert \; \; F \super^B G \in \Schwartz(\Pspace) \; \forall G\in\Schwartz(\Pspace) \bigr \} 
		\\
		\MoyalAlgebra_R^B(\Pspace) &:= \bigl \{ F \in \Schwartz'(\Pspace) \; \; \vert \; \; G \super^B F \in \Schwartz(\Pspace) \; \forall G \in \Schwartz(\Pspace) \bigr \} 
		,
	\end{align*}
	and then we define $\sMoyalAlg = \MoyalAlgebra_L^B(\Pspace) \cap \MoyalAlgebra_R^B(\Pspace)$.
\end{definition}
The product of $F,G\in\sMoyalAlg$ can be defined in the same way as in~\cite{Mantoiu_Purice:magnetic_Weyl_calculus:2004} via duality,
\begin{align}
    \bigl ( F\super^B G \, , \, H \bigr ) := \bigl ( F \, , \, G\super^B H \bigr ) \qquad \forall H\in\Schwartz(\Pspace) .
    \label{super_calculus_extension_by_duality:eqn:product_on_semi_super_Moyal_space}
\end{align}
From the very definition of $\sMoyalAlg$ we can deduce that $F\super^B G\in\sMoyalAlg$.

The magnetic super Moyal algebra contains many classes of functions, including the two types of Hörmander symbols we have introduced earlier and in  Definition~\ref{symbol_super_calculus:defn:Hoermander_super_symbols}. 
\begin{lemma}\label{symbol_super_calculus:lem:Hoermander_symbols_in_super_Moyal_algebra}
	Suppose the magnetic field $B$ satisfies Assumption~\ref{intro:assumption:bounded_magnetic_field} and that $0 \leq \rho \leq 1$, $0 < \eps , \lambda \leq 1$. Then the following holds:
	\begin{enumerate}[(1)]
		\item For any $m \in \R$, we have $S_{\rho,0}^m(\Pspace) \subset \sMoyalAlg$.
		\item For any $m_L , m_R \in \R$, we have $S_{\rho,0}^{m_L,m_R}(\Pspace) \subset \sMoyalAlg$.
	\end{enumerate}
\end{lemma}
The readers can find the proof in Appendix~\ref{appendix:oscillatory_integrals:super_weyl_product}. 

Of course, $\sMoyalAlg$ contains other topological vector spaces made up of functions, \eg the uniformly polynomially bounded functions $\Cont^{\infty}_{\mathrm{pol,u}}(\Pspace)$ (compare with \cite[Proposition~23]{Mantoiu_Purice:magnetic_Weyl_calculus:2004}). But we shall not give a proof here. 
\section{A magnetic pseudodifferential super operator calculus for Hörmander symbols} 
\label{symbol_super_calculus}
We left off by extending magnetic super Weyl calculus by duality to distributions. Of particular importance are spaces of distributions which have good composition properties; that is the magnetic semi-super Moyal space $\ssMoyalSpace$ and the magnetic super Moyal algebra $\sMoyalAlg$. But for the most part, the purpose of this whole procedure was to extend magnetic super Weyl calculus to functions which are not necessarily Schwartz function. Typically, we are interested in either Hörmander symbol classes $S^m_{\rho,\delta}(\Pspace)$, which is defined just as in Definition~\ref{magnetic_Weyl_calculus:defn:Hoermander_classes} but with $\pspace$ being replaced by $\Pspace$; we will re-use the notation introduced there. For some of our results, we will also introduce 
\begin{definition}[Hörmander super symbol classes $S^{m_L,m_R}_{\rho,\delta}(\Pspace)$]\label{symbol_super_calculus:defn:Hoermander_super_symbols}
	Let $m_L , m_R \in \R$, $0 \leq \delta \leq \rho \leq 1$ and $\delta<1$. The topological vector space $S_{\rho,\delta}^{m_L,m_R}(\Pspace)$ consists of functions $F \in \Cont^{\infty}(\Pspace)$ such that, for all $a_L , a_R , \alpha_L , \alpha_R \in \N_0^d$, there exists $C_{a_L a_R \alpha_L \alpha_R} > 0$ such that, for all $\Xbf = (X_L,X_R) \in \Pspace$, we have 
	\begin{align*}
		\babs{\partial_{x_L}^{a_L} \partial_{\xi_L}^{\alpha_L} \partial_{x_R}^{a_R} \partial_{\xi_R}^{\alpha_R} F(X_L,X_R)} \leq C_{a_L a_R \alpha_L \alpha_R} \; \sexpval{\xi_L}^{m_L - \sabs{\alpha_L} \rho + \sabs{a_L} \delta} \; \sexpval{\xi_R}^{m_R - \sabs{\alpha_R} \rho + \sabs{a_R} \delta} 
		. 
	\end{align*}
	The smallest such constants $C_{a_L a_R \alpha_L \alpha_R}$ are the seminorms 
	\begin{align*}
		\snorm{F}_{m_L m_R,a_L a_R \alpha_L \alpha_R} := \sup_{\Xbf \in \Pspace} \Bigl ( \sexpval{\xi_L}^{- m_L + \sabs{\alpha_L} \rho - \sabs{a_L} \delta} \; \sexpval{\xi_R}^{- m_R + \sabs{\alpha_R} \rho - \sabs{a_R} \delta} \; \babs{\partial_{x_L}^{a_L} \partial_{\xi_L}^{\alpha_L} \partial_{x_R}^{a_R} \partial_{\xi_R}^{\alpha_R} F(X_L,X_R)} \Bigr ) 
		. 
	\end{align*}
\end{definition}
These differ from regular Hörmander symbols in that we can keep track of the behavior in left and right momenta separately. They have been used occasionally in the literature before, albeit not in the context of a pseudodifferential \emph{super} calculus. Kumano-go has introduced a more general version of $S^{m_L,m_R}_{\rho,\delta}(\Pspace)$ (\cf \cite[Definition~2.1]{Kumanogo:PsiDOs_multiple_symbols_L2_boundedness:1975}) and used these generalized symbol classes to efficiently study the formulas for the adjoint and the analog of the Weyl product for the Kohn-Nierenberg quantization in Section~2.2 of his book \cite{Kumanogo:pseudodiff:1981}; the Kohn-Nierenberg quantization is equivalent to Weyl quantization (\cf \cite[Section~3.2]{Mantoiu_Purice_Richard:twisted_X_products:2004}) and corresponds to a different operator ordering. At roughly the same time Rodino \cite{Rodino:PsiDOs_manifolds:1975} introduced these symbol classes on open sets in Euclidean spaces. 
\begin{remark}[Alternate set of seminorms]
	For our estimates, we will instead use the seminorms
	\begin{align}
		q_N^{m_L,m_R}(F) := \max_{\sabs{a_L} + \sabs{a_R} + \sabs{\alpha_L} + \sabs{\alpha_R} \leq N} \snorm{F}_{m_L m_R,a_L a_R \alpha_L \alpha_R} 
		, 
		&&
		m_L , m_R \in \R , \; N \in \N_0 
		. 
		\label{symbol_super_calculus:eqn:max_seminorms_super_Hoermander_symbols}
	\end{align}
\end{remark}
\begin{remark}[Hörmander super symbol classes do not nest into $S^m_{\rho,\delta}(\Pspace)$]\label{symbol_super_calculus:rem:nesting_Hoermander_super_Hoermander_classes}
	It is tempting to think that Hörmander super symbol classes are nested, but with a little bit of thought we see that 
	\begin{align}
		S^{m_L,m_R}_{\rho,\delta}(\Pspace) \not\subseteq S^m_{\rho,\delta}(\Pspace) 
		\label{symbol_super_calculus:eqn:inclusion_Hoermander_spaces}
	\end{align}
	holds even if we try to choose $m$ suitably. Clearly, at first glance $m = m_L + m_R$ seems like a sensible option. 
	
	Let us deal with the obvious, easy attempt $m = m_L + m_R$. We will give a simple counter example where $m_L = -1 = m_R$: if we \emph{could} include $S^{-1,-1}_{\rho,\delta}(\Pspace)$ into $S^{-2}_{\rho,\delta}(\Pspace)$ ($m = m_L + m_R = -2$), then this would imply the estimate 
	\begin{align*}
		\sexpval{\xi_L}^{-1} \, \sexpval{\xi_R}^{-1} \leq  C \, \bexpval{(\xi_L,\xi_R)}^{-2} 
	\end{align*}
	for some constant $C > 0$. But this estimate cannot be true, no matter the value of $C$: the left-hand side decays like $\nicefrac{1}{\sabs{\xi_L}}$ as $\sabs{\xi_L} \to \infty$, which is slower than the quadratic decay in the supposed upper bound. 
	
	Another attempt would be to opt for $m = \sabs{m_L} + \sabs{m_R}$, in order to deal with symbols of negative order. But even that would not work when $\rho \gneq 0$: take a function $F \in S^{1,0}_{1,0}(\Pspace)$. Then our definition of Hörmander super symbol classes implies that for any $j = 1 , \ldots , d$ the first-order partial derivatives lie in $\partial_{\xi_{L,j}} F \in S^{0,0}_{1,0}(\Pspace)$ and $\partial_{\xi_{R,j}} F \in S^{1,-1}_{1,0}(\Pspace)$. 
	
	If $F \in S^1_{1,0}(\Pspace)$ were to hold true, then independently of whether we take partial derivatives with respect to the left or right momentum variables, we would deduce $\partial_{\xi_{L,j}} F , \partial_{\xi_{R,j}} F \in S^0_{1,0}(\Pspace)$. But the boundedness of the partial derivatives is not compatible with $\partial_{\xi_{R,j}} F \in S^{1,-1}_{1,0}(\Pspace)$ since $\partial_{\xi_{R,j}} F$ is allowed to grow linearly in $\sabs{\xi_L}$ at $\infty$. 
	
	The only nesting we do have is $S^{m_L,m_R}_{\rho,0}(\Pspace) \subset S^{\sabs{m_L} + \sabs{m_R}}_{0,0}(\Pspace)$, which states that Hörmander super symbols have uniform polynomial growth in momentum. But the utility of this inclusion is rather limited and it would make more sense then to think of 
	\begin{align*}
		S^{m_L,m_R}_{\rho,\delta}(\Pspace) \subset \Cont^{\infty}_{\mathrm{pol,u}}(\Pspace)
	\end{align*}
	instead, where the space on the right consists of smooth, uniformly polynomially bounded functions. But then we are outside of the realm of Hörmander symbols. 
	
	The opposite inclusion 
	\begin{align*}
		S^m_{\rho,\delta}(\Pspace) \subseteq S^{m,m}_{\rho,\delta}(\Pspace)
	\end{align*}
	does hold for $m \geq 0$ and $\rho = 0 = \delta$ — and for those \emph{only}, though. To see that, we need to establish that there exists a constant $C_m > 0$ for which 
	\begin{align}
		\sexpval{(\xi_L,\xi_R)}^{\sabs{m}} \leq C_m \, \sexpval{\xi_L}^{\sabs{m}} \, \sexpval{\xi_R}^{\sabs{m}} 
		\label{symbol_super_calculus:eqn:inequality_nesting_Hoermander_super_Hoermander_classes_m_geq_0}
	\end{align}
	is satisfied; the absolute value in the exponent is for emphasis. Excluding the trivial case $m = 0$, we need to consider when $m > 0$ is positive. Then the desired inequality follows from the elementary inequality 
	\begin{align*}
		\sexpval{(\xi_L,\xi_R)}^2 &= 1 + \sabs{\xi_L}^2 + \sabs{\xi_R}^2 \leq (1 + \sabs{\xi_L}^2) \, (1 + \sabs{\xi_R}^2) 
		\\
		&\leq \sexpval{\xi_L}^2 \, \sexpval{\xi_R}^2 
	\end{align*}
	and the fact that the function $s \mapsto s^m$ is monotonically increasing. That is, $C_m = 1$ is the best constant. 
	
	For $m < 0$ the map $s \mapsto s^m$ is monotonically decreasing, and therefore, taking a negative power of it flips the inequality. Alternatively, the interested reader can convince themselves directly for the case $m = -2$ that the inequality fails. 
\end{remark}
In addition to showing that the products $\semisuper^B$ and $\super^B$ map two Hörmander symbols onto a Hörmander symbol in a continuous fashion, we will also study their asymptotic expansions in a future work. Note that to ensure that the magnetic phase factor do not kick the products out of the symbol classes, we must suppose that the magnetic field $B$ satisfies the stricter Assumption~\ref{intro:assumption:bounded_magnetic_field}. Otherwise derivatives of the product in $x$ may grow polynomially, which is incompatible with either definition of Hörmander (super) symbol classes. 
\medskip

\noindent
Let us begin by investigating the products of Hörmander symbols first, which are the basis for any calculus. In future works, this will be the basis for proving that \eg super Moyal resolvents of magnetic pseudodifferential super operators associated to Hörmander symbols are again Hörmander symbols.

\subsection{The magnetic semi-super Weyl product $\semisuper^B$} 
\label{symbol_super_calculus:composition:semi_super_product}
The first order of business is to show that Hörmander symbols have good composition properties under the magnetic semi-super product. We have fortunately already deal with the first step, Lemma~\ref{super_calculus_extension_by_duality:lem:Hoermander_symbol_classes_contained_in_semi_super_Moyal_space}, which states that the two types of Hörmander symbols lie in the magnetic semi-super Moyal space $\ssMoyalSpace$. Combined with the observation that Hörmander symbols $S^m_{\rho,\delta}(\pspace)$ lie in $\Schwartz'(\pspace)$, we can define $F \semisuper^B g$ in case $F$ and $g$ are Hörmander symbols in the sense of~\eqref{magnetic_super_PsiDOs:eqn:extension_semi_super_product}. However, to show that the resulting tempered distribution $F\semisuper^B g$ is in fact another Hörmander symbol takes a bit more effort. 
\begin{proposition}[$\semisuper^B$ composition of Hörmander symbols]\label{symbol_super_calculus:prop:semi_super_composition}
	Suppose the magnetic field $B$ satisfies Assumption~\ref{intro:assumption:bounded_magnetic_field} and that $0 \leq \rho \leq 1$. Then the following holds:
	\begin{enumerate}[(1)]
		\item Let $m , m' \in \R$. Then the map $(F,g) \mapsto F \semisuper^B g$ gives rise to a continuous bilinear map 
		\begin{align*}
			\semisuper^B : S_{\rho,0}^m(\Pspace) \times S_{\rho,0}^{m'}(\pspace) \longrightarrow S_{\rho,0}^{m+m'}(\pspace)
			. 
		\end{align*}
		\item Let $m , m_L , m_R \in \R$. Then $\semisuper^B$ gives rise to a continuous bilinear map 
		\begin{align*}
			\semisuper^B : S_{\rho,0}^{m_L,m_R}(\Pspace) \times S_{\rho,0}^m(\pspace) \longrightarrow S_{\rho,0}^{m + m_L + m_R}(\pspace)
			. 
		\end{align*}
	\end{enumerate}
\end{proposition}
We will only give the structure of the proof. It rests on standard oscillatory integral techniques, which are straightforward. But since they are standard and quite tedious, we have moved them to Appendix~\ref{appendix:oscillatory_integrals:semi_super_product}. 
\begin{proof}
	To prove both~(1) and~(2) we need to show that the integral~\eqref{magnetic_super_PsiDOs:eqn:semi_super_product_integral_expression} exists as an oscillatory integral and belongs to the correct symbol classes.

	We know by Corollary~\ref{appendix:oscillatory_integrals:cor:magnetic_flux_estimate_semi_super_product} that the factor $\e^{-\ii\lambda\recfluxp(x,y_L,y_R,z)}$ satisfies the assumptions on $G_{\tau'}$ of Lemma~\ref{appendix:oscillatory_integrals:lem:semi_super_product_existence_oscillatory_integral} with $\tau'=1$. Therefore, the oscillatory integral~\eqref{magnetic_super_PsiDOs:eqn:semi_super_product_integral_expression} defining the semi-super product $F \semisuper^B g$ satisfies the assumptions of Lemma~\ref{appendix:oscillatory_integrals:lem:semi_super_product_existence_oscillatory_integral} with $\tau=\tau'=1$ and $a=b=c=\alpha=\beta=\gamma=0$. Thus, it follows from Lemma~\ref{appendix:oscillatory_integrals:lem:semi_super_product_existence_oscillatory_integral} that the map $(F,g) \mapsto F \semisuper^B g$ gives rise to continuous bilinear maps from $S_{\rho,0}^m(\Pspace)\times S_{\rho,0}^{m'}(\pspace)$ to $S_{\rho,0}^{m+m'}(\pspace)$ and from $S_{\rho,0}^{m_L,m_R}(\Pspace) \times S_{\rho,0}^m(\pspace)$ to $S_{\rho,0}^{m+m_L+m_R}(\pspace)$. This proves both~(1) and~(2). The proof is complete.
\end{proof}
%

\subsection{The magnetic super Weyl product $\super^B$} 
\label{symbol_super_calculus:composition:super_Weyl_product}
Analogously, we proceed with the magnetic super Weyl product. When $F$ and $G$ are Hörmander symbols, we know by Lemma~\ref{symbol_super_calculus:lem:Hoermander_symbols_in_super_Moyal_algebra} that $F$ and $G$ belong to the magnetic super Moyal algebra $\sMoyalAlg$, and hence we can make sense of $F\super^B G$ as an element of $\sMoyalAlg$ by utilizing~\eqref{super_calculus_extension_by_duality:eqn:product_on_semi_super_Moyal_space}. Furthermore, if $F$ and $G$ are Hörmander symbols, then $F \super^B G$ is not merely an element of the super Moyal algebra, but actually another Hörmander symbol.
\begin{proposition}\label{symbol_super_calculus:prop:magnetic_super_composition_Hoermander_symbols}
	Suppose the magnetic field $B$ satisfies Assumption~\ref{intro:assumption:bounded_magnetic_field} and that $0 \leq \rho \leq 1$. Then the following holds:
	\begin{enumerate}[(1)]
		\item Let $m , m' \in \R$. Then the map $(F,G) \mapsto F \super^B G$ gives rise to a continuous bilinear map 
		\begin{align*}
			\super^B : S_{\rho,0}^m(\Pspace) \times S_{\rho,0}^{m'}(\Pspace) \longrightarrow S_{\rho,0}^{m + m'}(\Pspace)
			. 
		\end{align*}
		\item Let $m_L , m_L' , m_R , m_R' \in \R$. Then $\super^B$ gives rise to a continuous bilinear map 
		\begin{align*}
			\super^B : S_{\rho,0}^{m_L,m_R}(\Pspace) \times S_{\rho,0}^{m_L',m_R'}(\Pspace) \longrightarrow S_{\rho,0}^{m_L+m_L',m_R+m_R'}(\Pspace)
			. 
		\end{align*}
	\end{enumerate}
\end{proposition}
As before, we will factor out the tedious, but necessary book-keeping that is involved in proving the existence of oscillatory integrals to Appendix~\ref{appendix:oscillatory_integrals:super_weyl_product}. 
\begin{proof}
	For proofs of both~(1) and~(2), we need to show that the integral~\eqref{formal_super_calculus:eqn:super_Weyl_product_formula} exists as an oscillatory integral and belongs to the correct symbol classes.

	Thanks to Corollary~\ref{appendix:oscillatory_integrals:cor:magnetic_flux_estimate_small_Weyl_product} we see that the magnetic flux factor $\e^{- \ii \lambda \trifluxp(x_L,y_L,z_L)} \, \e^{- \ii \lambda \trifluxp(x_R,z_R,y_R)}$ in~\eqref{formal_super_calculus:eqn:super_Weyl_product_formula} satisfies the assumptions on $G_{\tau'}$ of Lemma~\ref{appendix:oscillatory_integrals:lem:existence_oscillatory_integral_super_Weyl_product} with $\tau'=1$. Thus, the oscillatory integral~\eqref{formal_super_calculus:eqn:super_Weyl_product_formula} defining the super product $F\super^B G$ satisfies the assumptions of Lemma~\ref{appendix:oscillatory_integrals:lem:existence_oscillatory_integral_super_Weyl_product} with $\tau = \tau' = 1$ and $a = b = c = l = \alpha = \beta = \gamma = \delta = 0$. Therefore, it follows from Lemma~\ref{appendix:oscillatory_integrals:lem:existence_oscillatory_integral_super_Weyl_product} that the map $(F,G)\mapsto F\super^B G$ gives rise to continuous bilinear maps from $S_{\rho,0}^m(\Pspace)\times S_{\rho,0}^{m'}(\Pspace)$ to $S_{\rho,0}^{m+m'}(\Pspace)$ and from $S_{\rho,0}^{m_L,m_R}(\Pspace)\times S_{\rho,0}^{m_L',m_R'}(\Pspace)$ to $S_{\rho,0}^{m_L+m_L',m_R+m_R'}(\Pspace)$. This proves both~(1) and~(2) and completes the proof.
\end{proof}
%
%
\begin{appendix}
	\section{Estimates on the magnetic phase factor} 
	\label{appendix:oscillatory_integrals:magnetic_phase_factor}
	For our results we first need to establish some estimates on the magnetic phase and its exponential. The following lemma is the first step to obtain necessary estimates.
	\begin{lemma}[\cite{Iftimie_Mantoiu_Purice:magnetic_psido:2006}; see also~\cite{Lein:two_parameter_asymptotics:2008, Lein:progress_magWQ:2010}] \label{appendix:oscillatory_integrals:lem:magnetic_phase_factor_estimate}
		Suppose that the magnetic field $B = \dd A$ satisfies Assumption~\ref{intro:assumption:bounded_magnetic_field}. Then there are functions $D_{jk}, D_{jk}', D_{jk}'', E_{jk}, E_{jk}', E_{jk}''$, $1 \leq j , k \leq d$, in $\Cont^\infty_{\mathrm{b}}(\R^{3d})$ such that
		\begin{align*}
			\partial_{x_j} \trifluxp(x,y,z) &= \sum_{k=1}^d \bigl ( D_{jk}^{(x)}(x,\eps y,\eps z) \, y_k + E_{jk}^{(x)}(x,\eps y,\eps z) \, z_k \bigr )
			, 
			\\
			\partial_{y_j} \trifluxp(x,y,z) &= \sum_{k=1}^d \bigl ( D_{jk}^{(y)}(x,\eps y,\eps z) \, \eps y_k + E_{jk}^{(y)}(x,\eps y,\eps z) \, \eps z_k \bigr )
			, 
			\\
			\partial_{z_j} \trifluxp(x,y,z) &= \sum_{k=1}^d \bigl ( D_{jk}^{(z)}(x,\eps y,\eps z) \, \eps y_k + E_{jk}^{(z)}(x,\eps y,\eps z) \, \eps z_k \bigr ) 
			.
		\end{align*}
		Here the functions $D_{jk}^{(x)}, D_{jk}^{(y)}, D_{jk}^{(z)}, E_{jk}^{(x)}, E_{jk}^{(y)}, E_{jk}^{(z)}$ are given by
		\begin{align*}
			D_{jk}^{(x)}(x,y,z) &= \frac{1}{2} \int_{-1}^1 \dd t \, \Bigl ( B_{jk} \bigl ( x + \tfrac{1}{2} (ty-z) \bigr ) - B_{jk} \bigl ( x + \tfrac{1}{2} t (y+z) \bigr ) \Bigr )
			, 
			\\
			E_{jk}^{(x)}(x,y,z) &= \frac{1}{2} \int_{-1}^1 \dd t \, \Bigl ( B_{jk} \bigl ( x + \tfrac{1}{2} (y+tz) \bigr ) - B_{jk} \bigl ( x + \tfrac{1}{2} t (y+z) \bigr ) \Bigr )
			, 
			\\
			D_{jk}^{(y)}(x,y,z) &= \frac{1}{4} \int_{-1}^1 \dd t \, t \, \Bigl ( B_{jk} \bigl ( x + \tfrac{1}{2} (ty-z) \bigr ) - B_{jk} \bigl ( x + \tfrac{1}{2} t (y+z) \bigr ) \Bigr )
			, 
			\\
			E_{jk}^{(y)}(x,y,z) &= \frac{1}{4} \int_{-1}^1 \dd t \, \Bigl ( B_{jk} \bigl ( x + \tfrac{1}{2} (y+tz) \bigr ) - t \, B_{jk} \bigl ( x + \tfrac{1}{2} t (y+z) \bigr ) \Bigr )
			, 
			\\
			D_{jk}^{(z)}(x,y,z) &= \frac{1}{4} \int_{-1}^1 \dd t \, \Bigl ( B_{jk} \bigl ( x + \tfrac{1}{2} (ty-z) \bigr ) - t \, B_{jk} \bigl ( x + \tfrac{1}{2} t (y+z) \bigr ) \Bigr )
			, 
			\\
			E_{jk}^{(z)}(x,y,z) &= \frac{1}{4} \int_{-1}^1 \dd t \, t \, \Bigl ( B_{jk} \bigl ( x + \tfrac{1}{2} (y+tz) \bigr ) - B_{jk} \bigl ( x + \tfrac{1}{2} t (y+z) \bigr ) \Bigr ) 
			.
		\end{align*}
	\end{lemma}
	This lemma can be proved by associating the small parameter $\eps$ to the proof in~\cite{Iftimie_Mantoiu_Purice:magnetic_psido:2006}. As a direct consequence of the above lemma we obtain the following result on the magnetic phase factor. 
	\begin{corollary}[\cite{Iftimie_Mantoiu_Purice:magnetic_psido:2006}; see also~\cite{Lein:two_parameter_asymptotics:2008, Lein:progress_magWQ:2010}]\label{appendix:oscillatory_integrals:cor:magnetic_flux_estimate_small_Weyl_product}
		Suppose that the magnetic field $B = \dd A$ satisfies Assumption~\ref{intro:assumption:bounded_magnetic_field} and $0 < \eps , \lambda \leq 1$. Then, for all $a,b,c \in \N_0^d$, there exists $C_{abc} , \tilde{C}_{abc} > 0$ such that
		\begin{align*}
			\abs{ \partial_x^a \partial_y^b \partial_z^c \e^{- \ii \lambda \trifluxp(x,y,z)}} \leq C_{abc} \, \bigl ( \sexpval{y} + \sexpval{z} \bigr )^{\abs{a} + \abs{b} + \abs{c}} 
			\leq \tilde{C}_{abc} \, \sexpval{y}^{\abs{a} + \abs{b} + \abs{c}} \, \sexpval{z}^{\abs{a} + \abs{b} + \abs{c}} 
			.
		\end{align*}
	\end{corollary}
	Furthermore, we also get the following result on the phase factor~\eqref{formal_super_calculus:eqn:definition_magnetic_phase_factor_super_product}, 
	\begin{align*}
		\recfluxp(x,y_L,y_R,z) := \trifluxp(x,y_L,z) + \trifluxp(x,y_L+z,y_R)
		. 
	\end{align*}
	\begin{corollary}\label{appendix:oscillatory_integrals:cor:magnetic_flux_estimate_semi_super_product}
		Suppose that the magnetic field $B = \dd A$ satisfies Assumption~\ref{intro:assumption:bounded_magnetic_field} and $0 < \eps , \lambda \leq 1$. Then, for all $a , b_L , b_R , c \in \N_0^d$, there exists $C_{a b_L b_R c} , \tilde{C}_{a b_L b_R c} > 0$ such that
		\begin{align*}
			\abs{\partial_x^a \partial_{y_L}^{b_L} \partial_{y_R}^{b_R} \partial_z^c \e^{- \ii \lambda \recfluxp(x,y_L,y_R,z)} } &\leq C_{a b_L b_R c} \, \bigl ( \sexpval{y_L} + \sexpval{y_R} + \sexpval{z} \bigr )^{\abs{a} + \sabs{b_L} + \abs{b_R} + \abs{c}} 
			\\
			&\leq \tilde{C}_{a b_L b_R c} \, \sexpval{y_L}^{\abs{a} + \abs{b_L} + \abs{b_R} + \abs{c}} \,  \sexpval{y_R}^{\abs{a} + \abs{b_L} + \abs{b_R} + \abs{c}} \, \sexpval{z}^{\abs{a} + \abs{b_L} + \abs{b_R} + \abs{c}}  
			.
		\end{align*}
	\end{corollary}
	%

	\section{Existence of oscillatory integrals} 
	\label{appendix:oscillatory_integrals}
	The purpose of this appendix is to provide some elementary, albeit tedious proofs establishing the existence of certain oscillatory integrals.

	\subsection{Magnetic semi-super product $\semisuper^B$} 
	\label{appendix:oscillatory_integrals:semi_super_product}
	We begin with the oscillatory integrals that enter the magnetic semi-super product. 
	
	First, we will need to furnish the proof that Hörmander symbols 
	\begin{align*}
		S^m_{\rho,0}(\Pspace) , S^{m_L,m_R}_{\rho,0}(\Pspace) \subseteq \ssMoyalSpace 
	\end{align*}
	lie in the magnetic semi-super Moyal space, which is the basis for the extension of the semi-super product to Hörmander symbols. 
	\begin{remark}[Precise notation needed for derivatives]\label{appendix:oscillatory_integrals:rem:derivatives}
		In the following, we will sometimes need to be precise in our notation when taking certain derivatives. For example, the transpose in $F^{\mathrm{t}}(X_L,X_R) = F(X_R,X_L)$ swaps the argument, and in general left- and right-hand side of 
		\begin{align*}
			\partial_{x_L}^{a_L} \partial_{\xi_L}^{\alpha_L} \partial_{x_R}^{a_R} \partial_{\xi_R}^{\alpha_R} \bigl ( F(X_R,X_L) \bigr ) &= \bigl ( \partial_{x_L}^{a_R} \partial_{\xi_L}^{\alpha_R} \partial_{x_R}^{a_L} \partial_{\xi_R}^{\alpha_L} F \bigr )(X_R,X_L) 
			\\
			&\neq \bigl ( \partial_{x_L}^{a_L} \partial_{\xi_L}^{\alpha_L} \partial_{x_R}^{a_R} \partial_{\xi_R}^{\alpha_R} F \bigr )(X_R,X_L) 
		\end{align*}
		will disagree. So to do proper book keeping, we will add brackets to emphasize that 
		\begin{align*}
			\bigl ( \partial_{x_L}^{a_L} \partial_{\xi_L}^{\alpha_L} \partial_{x_R}^{a_R} \partial_{\xi_R}^{\alpha_R} F \bigr )(X_R,X_L) 
		\end{align*}
		is the derivative of $F$ evaluated at $X_R$ and $X_L$ (\ie the arguments are swapped). 
		
		This level of precision is also needed when we consider derivatives of 
		\begin{align*}
			g \bigl ( x + \tfrac{\eps}{2}(y_L - y_R) , \xi - \zeta \bigr )
		\end{align*}
		in $x$, $y_L$, $y_R$, $\xi$ and $\zeta$; to characterize the behavior of the function, we need to collect derivatives in the spatial and momentum variables, and we will also add brackets here to avoid ambiguity. 
	\end{remark}
	\begin{proof}[Lemma~\ref{super_calculus_extension_by_duality:lem:Hoermander_symbol_classes_contained_in_semi_super_Moyal_space}]
		\begin{enumerate}[(1)]
			\item Let $F \in S_{\rho,0}^m(\Pspace)$ and $g \in \Schwartz(\pspace)$. We shall prove that, for every $N , K \in \R$ and $r , \varrho \in \N_0^d$ we can estimate the corresponding Schwartz seminorm 
			\begin{align*}
				\bnorm{F^{\mathrm{t}} \semisuper^B g}_{\Schwartz , N K r \varrho} := \sup_{X \in \pspace} \, \sexpval{x}^N \, \sexpval{\xi}^K \babs{\partial_x^r \partial_{\xi}^{\varrho} F^{\mathrm{t}} \semisuper^B g(X)} 
			\end{align*}
			by some constant $C_{N K r \varrho} > 0$ times a finite number of seminorms of $S^m_{\rho,0}(\Pspace)$ and $\Schwartz(\pspace)$, 
			\begin{align*}
				\bnorm{F^{\mathrm{t}} \semisuper^B g}_{\Schwartz , N K r \varrho} \leq C_{n \nu r \varrho} \, p^m_J(F) \, \max_{N' , K' , \sabs{r'} , \sabs{\varrho'} \leq Q} \snorm{g}_{\Schwartz , N' K' r' \varrho'} 
			\end{align*}
			for some appropriately chosen integers $J , Q \geq 0$. 
		
			That being said, we proceed to estimate the derivative $\partial_x^n \partial_{\xi}^{\nu} F^{\mathrm{t}} \semisuper^B g(X)$ of the magnetic semi-super product for $n , \nu \in \N_0^d$. Writing out the symplectic Fourier transforms in the magnetic semi-super product and distributing the derivatives across the different terms gives a sum 
			\begin{align*}
				\partial_x^n \partial_{\xi}^{\nu} &F^{\mathrm{t}} \semisuper^B g(X) 
				\\
				&= \sum_{\substack{a_L + a_R + b + c = n \\ \alpha_L + \alpha_R + \beta = \nu}} \frac{n!}{a_L! \, a_R! \, b! \, c!} \frac{\nu!}{\alpha_L! \, \alpha_R! \, \beta!} \frac{1}{(2\pi)^{6d}} \int_{\Pspace} \dd \Ybf \int_{\pspace} \dd Z \int_{\Pspace} \dd \Ybf' \int_{\pspace} \dd Z'
				\, \cdot \\
				&\quad \quad 
				\bigl ( \partial_x^{a_L} \partial_{\xi}^{\alpha_L} \e^{+ \ii \sigma (X-Y_L',Y_L)} \bigr ) \, \bigl ( \partial_x^{a_R} \partial_{\xi}^{\alpha_R} \e^{+ \ii \sigma (X-Y_R',Y_R)} \bigr ) \, \bigl ( \partial_x^b \partial_{\xi}^{\beta} \e^{+ \ii \sigma (X-Z',Z)} \bigr )
				\, \cdot \\
				&\quad \quad 
				\e^{+ \ii \frac{\eps}{2} \sigma(Y_L+Z,Y_R+Z)} \, \partial_x^c \bigl ( \e^{- \ii \lambda \recfluxp(x,y_L,y_R,z)} \bigr ) \, F(Y_R',Y_L') \, g(Z') 
				\\
				&= \sum_{\substack{a_L + a_R + b + c = n \\ \alpha_L + \alpha_R + \beta = \nu}} \frac{n!}{a_L! \, a_R! \, b! \, c!} \frac{\nu!}{\alpha_L! \, \alpha_R! \, \beta!} \frac{1}{(2\pi)^{6d}} \int_{\Pspace} \dd \Ybf \int_{\pspace} \dd Z \int_{\Pspace} \dd \Ybf' \int_{\pspace} \dd Z'
				\, \cdot \\
				&\quad \quad 
				\e^{+ \ii \sigma (X-Y_L',Y_L)} \, \e^{+ \ii \sigma (X-Z',Z)} \, \e^{+ \ii \sigma (X-Y_R',Y_R)} \, \e^{+ \ii \frac{\eps}{2} \sigma(Y_L+Z,Y_R+Z)} \, \partial_x^c \bigl ( \e^{- \ii \lambda \recfluxp(x,y_L,y_R,z)} \bigr )
				\, \cdot \\
				&\quad \quad 
				\bigl ( \partial_{y_L'}^{a_R} \partial_{\eta_L'}^{\alpha_R} \partial_{y_R'}^{a_L} \partial_{\eta_R'}^{\alpha_L} F \bigr ) (Y_R',Y_L') \; \partial_{z'}^b \partial_{\zeta'}^{\beta} g(Z') 
				.
			\end{align*}
			Keep in mind Remark~\ref{appendix:oscillatory_integrals:rem:derivatives},  
			\begin{align*}
				\bigl ( \partial_{y_L'}^{a_R} F \bigr ) (Y_R',Y_L') &= \partial_{y_R'}^{a_R} \bigl ( F(Y'_R,Y'_L) \bigr ) 
			\end{align*}
			means we take the $a_R$th derivative in the left spatial variable of $F(X_L,X_R)$ and then evaluate the function at $X_L = Y_R'$ and $X_R = Y_L'$. 
			
			To simplify our presentation, we will omit sums and most prefactors; further, we will also exchange the labels of the multi indices $a_L$ and $\alpha_L$ with $a_R$ and $\alpha_R$. What is important is that the sums in the relevant multi indices are all finite and of a specific order. 
		
			We now collect the phase factors so as to factor out the free variables $y_L$, $y_R$, $z$, $\eta_L$, $\eta_R$ and $\zeta$, integrate them out and then change momentum variables to $\eta_L = \xi-\eta_L'$, $\eta_R = \xi-\eta_R'$ and $\zeta = \xi-\zeta'$ in the same way as in the proof of Proposition~\ref{magnetic_super_PsiDOs:prop:semi_super_product_S} (3). That gives us a finite sum of integrals of the type 
			\begin{align}
				\frac{1}{(2\pi)^{3d}} \int_{\Pspace} \dd \Ybf \int_{\pspace} \dd Z 
				\, &\e^{+ \ii (y_L \cdot \eta_L+y_R \cdot \eta_R+z\cdot\zeta)} \, \partial_x^c \bigl ( \e^{- \ii \lambda \recfluxp(x,y_L,y_R,z)} \bigr )
				\, \cdot 
				\notag \\
				&
				\bigl ( \partial_{x_L}^{a_L} \partial_{\xi_L}^{\alpha_L} \partial_{x_R}^{a_R} \partial_{\xi_R}^{\alpha_R} F \bigr ) \bigl ( x+\tfrac{\eps}{2}(y_L+z) , \xi - \eta_R , x - \tfrac{\eps}{2}(y_R + z) , \xi - \eta_L \bigr )
				\, \cdot \notag \\
				&
				\bigl ( \partial_x^b \partial_{\xi}^{\beta} g \bigr ) \bigl ( x+\tfrac{\eps}{2}(y_L-y_R) , \xi-\zeta \bigr ) 
				. \label{appendix:oscillatory_integrals:eqn:prototypical_oscillatory_integral_extension_via_duality}
			\end{align}
			To show the existence of this oscillatory integral we rely on $L$ operators of the form $L_{\xi} := \sqrt{1 - \Delta_{\xi}}$ that satisfy 
			\begin{align}
				 \sexpval{x}^{-1} \, L_{\xi} \e^{+ \ii x \cdot \xi} = \e^{+ \ii x \cdot \xi} 
				 . \label{appendix:oscillatory_integrals:eqn:L_operator}
			\end{align}
			We first add sufficient even powers of $y_L, y_R$ and $z$ by using this. Then we get
			\begin{align*}
				\int_{\Pspace} \dd \Ybf \int_{\pspace} \dd Z \, &\bigl ( \sexpval{y_L}^{-2N_L} L_{\eta_L}^{2N_L} \e^{+ \ii y_L \cdot \eta_L} \bigr ) \, \bigl ( \sexpval{z}^{-2N} L_{\zeta}^{2N} \e^{+ \ii z\cdot\zeta} \bigr ) \, \bigl ( \sexpval{y_R}^{-2N_R} L_{\eta_R}^{2N_R} \e^{+ \ii y_R \cdot \eta_R} \bigr )
				\, \cdot \notag \\
				&
				\partial_x^c \bigl ( \e^{- \ii \lambda \recfluxp(x,y_L,y_R,z)} \bigr ) 
				\, \cdot \notag \\
				&
				\bigl ( \partial_{x_L}^{a_L} \partial_{\xi_L}^{\alpha_L} \partial_{x_R}^{a_R} \partial_{\xi_R}^{\alpha_R} F \bigr ) \bigl ( x+\tfrac{\eps}{2}(y_L+z) ,\xi-\eta_R , x - \tfrac{\eps}{2}(y_R+z) , \xi-\eta_L \bigr )
				\, \cdot \notag \\
				&
				\bigl ( \partial_x^b \partial_{\xi}^{\beta} g \bigr ) \bigl ( x+\tfrac{\eps}{2}(y_L-y_R) , \xi-\zeta \bigr ) .
			\end{align*}
			Next, we partially integrate this with respect to the associated momentum variables $\eta_L, \eta_R$ and $\zeta$. And then, we insert powers in the momentum variables $\eta_L, \eta_R$ and $\zeta$ by applying $L_{y_L}, L_{y_R}$ and $L_z$ to the factors $\e^{+\ii y_L\cdot\eta_L}, \e^{+\ii y_R\cdot\eta_R}$ and $\e^{+\ii z\cdot\zeta}$, respectively, and partially integrate with respect to $y_L,y_R$ and $z$. The last partial integration step gives rise to derivatives of $\sexpval{y_L}^{-2N_L}$ and in the other spatial variables, which we will label as 
			\begin{align}
				\sexpval{x}^{-2N} \, \varphi_{N a}(x) := \partial_x^a \bigl ( \sexpval{x}^{-2N} \bigr ) 
				. 
				\label{appendix:oscillatory_integrals:eqn:derivatives_Japanese_bracket}
			\end{align}
			Importantly for us, the functions $\varphi_{N a} \in \Cont^{\infty}(\R^d) \cap L^{\infty}(\R^d)$ are smooth and bounded. 
		
			After all that we obtain another finite sum of terms of the type 
			\begin{align}
				&\int_{\Pspace} \dd \Ybf \int_{\pspace} \dd Z \, \e^{+ \ii (y_L \cdot \eta_L+y_R \cdot \eta_R+z\cdot\zeta)}
				\, \cdot \notag \\
				&\,\,\,\,
				\sexpval{y_L}^{-2N_L} \, \sexpval{z}^{-2N} \, \sexpval{y_R}^{-2N_R} \, \sexpval{\eta_L}^{-2K_L} \, \sexpval{\zeta}^{-2K} \, \sexpval{\eta_R}^{-2K_R}
				\, \cdot \notag \\
				&\,\,\,\, 
				\varphi_{N_L a'_L}(y_L) \, \varphi_{N b'}(z) \, \varphi_{N_R a'_R}(y_R) \, \partial_x^c \partial_{y_L}^{c'_L} \partial_{y_R}^{c'_R} \partial_z^{c'} \bigl ( \e^{- \ii \lambda \recfluxp(x,y_L,y_R,z)} \bigr )
				\, \cdot \notag \\
				&\,\,\,\,
				\bigl ( \partial_{x_L}^{a_L + a_L'' + a_L'''} \partial_{\xi_L}^{\alpha_L + \alpha'_L} \partial_{x_R}^{a_R + a''_R + a'''_R} \partial_{\xi_R}^{\alpha_R + \alpha'_R} F \bigr ) \bigl ( x + \tfrac{\eps}{2} (y_L+z) , \xi - \eta_R , x - \tfrac{\eps}{2} (y_R+z) , \xi - \eta_L \bigr )
				\, \cdot \notag \\
				&\,\,\,\,
				\bigl ( \partial_x^{b+b''+b'''} \partial_{\xi}^{\beta+\beta'} g \bigr ) \bigl ( x+\tfrac{\eps}{2}(y_L-y_R) , \xi-\zeta \bigr ) 
				.
				\label{appendix:oscillatory_integrals:eqn:prototypical_oscillatory_integral_to_be_estimated_extension_via_duality}
			\end{align}
			Here, the prefactors contain non-negative powers of $\eps$ that are harmless since we have assumed $\eps \in (0,\eps_0]$ for some $\eps_0 > 0$. The sum is over all multi indices satisfying 
			\begin{align*}
				\sabs{\alpha'_R} &\leq 2N_L 
				, 
				&
				\sabs{\alpha'_L} &\leq 2N_R 
				, 
				&
				\sabs{\beta'} &\leq 2N
				, 
				\\
				\abs{a'_L + b'' + a''_L + c'_L} &\leq 2K_L 
				, 
				&
				\abs{a'''_R + b''' + a'_R + c'_R} &\leq 2K_R 
				, 
				&
				\abs{a''_R + b' + a'''_L + c'} &\leq 2K 
			\end{align*}
			as well as the previously introduced multi indices. Note that some of the left and right indices mix, because \eg we evaluate $F$ in the \emph{left} momentum variable at $\xi - \eta_R$. 
		
			For the purposes of our estimates we introduce the constant 
			\begin{align*}
				N(n , \nu) := \sabs{n} + \sabs{\nu} + 2(N_L + N_R + N + K_L + K_R + K) 
			\end{align*}
			that will become the order of the seminorm~\eqref{magnetic_Weyl_calculus:eqn:max_seminorm_Hoermander_symbols}, and we pick some arbitrary integers $J , Q > 0$ that will quantify the polynomial decay in $x$ and $\xi$. 
		
			With the help of our estimates on the exponential of the magnetic phase factor, Corollary~\ref{appendix:oscillatory_integrals:cor:magnetic_flux_estimate_semi_super_product}, the integrand of \eqref{appendix:oscillatory_integrals:eqn:prototypical_oscillatory_integral_to_be_estimated_extension_via_duality} can be estimated in absolute value by 
			\begin{align}
				\babs{\mbox{$\displaystyle \int$grand}}
				&\leq
				\sexpval{y_L}^{-2N_L} \, \sexpval{y_R}^{-2N_R} \, \sexpval{z}^{-2N} \, \sexpval{\eta_L}^{-2K_L} \, \sexpval{\eta_R}^{-2K_R} \, \sexpval{\zeta}^{-2K} 
				\, \cdot \notag \\
				&\quad \quad
				\Babs{\varphi_{N_L a'_L}(y_L) \, \varphi_{N_R a'_R}(y_R) \, \varphi_{N b'}(z)} \, \Babs{\partial_x^c \partial_{y_L}^{c'_L} \partial_{y_R}^{c'_R} \partial_z^{c'} \bigl ( \e^{- \ii \lambda \recfluxp(x,y_L,y_R,z)} \bigr )}
				\, \cdot \notag \\
				&\quad \quad
				\Bigl \lvert \bigl ( \partial_{x_L}^{a_L + a''_L + a'''_L} \partial_{\xi_L}^{\alpha_L + \alpha'_L} \partial_{x_R}^{a_R + a''_R + a'''_R} \partial_{\xi_R}^{\alpha_R + \alpha'_R} F \bigr ) \Bigr . 
				\notag \\
				&\qquad \qquad \qquad \qquad 
				\Bigl . \bigl ( x + \tfrac{\eps}{2} (y_L+z) , \xi - \eta_R , x - \tfrac{\eps}{2} (y_R+z) , \xi - \eta_L \bigr ) \Bigr \rvert 
				\, \cdot \notag \\
				&\quad \quad
				\Babs{\bigl ( \partial_x^{b + b'' + b'''} \partial_{\xi}^{\beta + \beta'} g \bigr ) \bigl ( x + \tfrac{\eps}{2} (y_L-y_R) , \xi - \zeta \bigr )} 
				\notag \\
				&\leq C' \, p_{N(n,\nu)}^m(F) \; \Bigl ( \max_{\sabs{b} + \sabs{\beta} \leq N(n,\nu)} \snorm{g}_{\Schwartz,J Q' b \beta} \Bigr ) \, 
				\, \cdot \notag \\
				&\quad \quad
				\sexpval{y_L}^{-2N_L + \sabs{c} + \sabs{c'_L} + \sabs{c'_R} + \sabs{c'}} \, \sexpval{y_R}^{-2N_R+\sabs{c} + \sabs{c'_L} + \sabs{c'_R} + \sabs{c'}} \, \sexpval{z}^{-2N + \sabs{c} + \sabs{c'_L} + \sabs{c'_R} + \sabs{c'}} 
				\, \cdot \notag \\
				&\quad \quad
				\bexpval{x + \tfrac{\eps}{2}(y_L-y_R)}^{-J} 
				\, \cdot \notag \\
				&\quad \quad
				\sexpval{\eta_L}^{-2K_L} \, \sexpval{\eta_R}^{-2K_R} \, \sexpval{\zeta}^{-2K}
				\, \cdot \notag \\
				&\quad \quad
				\sexpval{(\xi-\eta_R,\xi-\eta_L)}^{m - \rho (\sabs{\alpha_L} + \sabs{\alpha'_L} + \sabs{\alpha_R} + \sabs{\alpha'_R})} \, \sexpval{\xi-\zeta}^{-Q-m} 
				,
				\label{appendix:oscillatory_integrals:eqn:proof_Hoermander_classes_in_Moyal_space_estimate_oscillatory_integrand}
			\end{align}
			where we have set $Q' = Q + m$. In the last step, we need to find lower bounds on $N_L$, $N_R$, $N$, $K_L$, $K_R$ and $K$ that ensure not only ensure integrability, but also give us the claimed decay. To arrive at the suitable estimate, we will employ Peetre's Inequality 
			\begin{align*}
				\sexpval{\xi - \eta}^m \leq 2^{\nicefrac{\sabs{m}}{2}} \, \sexpval{\xi}^m \, \sexpval{\eta}^{\sabs{m}} 
			\end{align*}
			and combine it with $\sexpval{\epsilon \xi}^{\sabs{m}} \leq \sexpval{\xi}^{\sabs{m}}$ for any $m \in \R$ and $\epsilon \in (0,1]$: 
			\begin{align*}
				\bexpval{x + \tfrac{\eps}{2}(y_L-y_R)}^{-J} \, \sexpval{\xi-\zeta}^{-Q-m} \leq 2^{J + \frac{1}{2} (Q + \sabs{m})} \, \sexpval{x}^{-J} \, \sexpval{\xi}^{-Q - m} \, \sexpval{y_L}^J \, \sexpval{y_R}^J \, \sexpval{\zeta}^{Q + \abs{m}} .
			\end{align*}
			We will need two further inequalities: the first is \eqref{symbol_super_calculus:eqn:inequality_nesting_Hoermander_super_Hoermander_classes_m_geq_0}, which we have proven in Remark~\ref{symbol_super_calculus:rem:nesting_Hoermander_super_Hoermander_classes}. 
		
			The second one reads 
			\begin{align}
				\sexpval{(\xi,\xi)}^m \leq 2^{\nicefrac{\sabs{m}}{2}} \, \sexpval{\xi}^m
				, 
				\label{appendix:oscillatory_integrals:eqn:Japanese_bracket_xi_xi_squeezing_estimate_Japanese_bracket_xi}
			\end{align}
			which is a consequence of 
			\begin{align*}
				\sexpval{\xi}^2 = 1 + \sabs{\xi}^2 
				\leq 1 + 2 \sabs{\xi}^2 
				= \sexpval{(\xi,\xi)}^2 
				\leq 2 \sexpval{\xi}^2 
				. 
			\end{align*}
			All of these can be combined to bound \eqref{appendix:oscillatory_integrals:eqn:proof_Hoermander_classes_in_Moyal_space_estimate_oscillatory_integrand} from above by 
			\begin{align}
				&C'' \, p_{N(n,\nu)}^m(F) \; \Bigl ( \max_{\sabs{b} + \sabs{\beta} \leq N(n,\nu)} \snorm{g}_{\Schwartz,J Q' b \beta} \Bigr ) \, \sexpval{x}^{-J} \, \sexpval{\xi}^{-Q} 
				\, \cdot \notag \\
				&\quad \quad
				\sexpval{y_L}^{-2N_L + \abs{n} + 2(K_L+K+K_R) + J} \, \sexpval{y_R}^{-2N_R + \abs{n} + 2(K_L+K+K_R) + J} \, \sexpval{z}^{-2N+\abs{n}+2(K_L+K+K_R)} 
				\, \cdot \notag \\
				&\quad \quad 
				\sexpval{\eta_L}^{-2K_L + \abs{m}} \, \sexpval{\eta_R}^{-2K_R + \abs{m}} \, \sexpval{\zeta}^{-2K + Q + \abs{m}} 
				, 
				\label{appendix:oscillatory_integrals:eqn:proof_Hoermander_classes_in_Moyal_space_estimate_oscillatory_integrand-all-variables-separated}
			\end{align}
			where $C'' > 0$ is a suitable constant. To ensure integrability, we first need to pick constants $K_L$, $K_R$ and $K$ that are at least as large as 
			\begin{align*}
				K_L , K_R &> \tfrac{1}{2} (\sabs{m} + d) 
				,
				\\
				K &> \tfrac{1}{2} (\sabs{m} + d + Q)
				%
				.
			\end{align*}
			That is because the lower bounds on $N_L$, $N_R$ and $N$ will contain these constants as well, 
			\begin{align*}
				N_L , N_R &> \tfrac{1}{2} (\sabs{n} + d + J) + K_L + K_R + K 
				, 
				\\
				N &> \tfrac{1}{2} (\sabs{n} + d) + K_L + K_R + K 
				. 
			\end{align*}
			Once we have chosen those six constants thusly, the right-hand side of \eqref{appendix:oscillatory_integrals:eqn:proof_Hoermander_classes_in_Moyal_space_estimate_oscillatory_integrand-all-variables-separated} is integrable with respect to $y_L$, $y_R$, $z$, $\eta_L$, $\eta_R$ and $\zeta$. 
		
			This last estimate is independent of all multi indices save for $n$ and $\nu$. Once we sum all of the finitely many terms, we have demonstrated with the above arguments that there is a constant $C(n,\nu,J,Q) > 0$ such that all derivatives 
			\begin{align*}
				\babs{\partial_x^n \partial_{\xi}^{\nu} F^{\mathrm{t}} \semisuper^B g(X)} \leq 
				C(n,\nu,J,Q) \; p_{N(n,\nu)}^m(F) \; \Bigl ( \max_{\sabs{b} + \sabs{\beta} \leq N(n,\nu)} \snorm{g}_{\Schwartz,J Q' b \beta} \Bigr ) \, \sexpval{x}^{-J} \, \sexpval{\xi}^{-Q} 
			\end{align*}
			decay faster than any inverse polynomial in $X \in \pspace$. This shows that the map $g \mapsto F^{\mathrm{t}} \semisuper^B g$ gives rise to a continuous linear map from $\Schwartz(\pspace)$ to itself, and hence $F \in \ssMoyalSpace$. Therefore, we conclude that for any $m \in \R$ the Hörmander class $S_{\rho,0}^m(\Pspace) \subset \ssMoyalSpace$ lies in the magnetic semi-super Moyal space. This proves~(1).
			\item In principle, the proof is mostly identical, we just need to use different Hörmander seminorms~\eqref{symbol_super_calculus:eqn:max_seminorms_super_Hoermander_symbols} in our estimates of the integrand of \eqref{appendix:oscillatory_integrals:eqn:prototypical_oscillatory_integral_to_be_estimated_extension_via_duality}. So let $F \in S^{m_L,m_R}_{\rho,0}(\Pspace)$ be from a Hörmander super symbol class and $g \in \Schwartz(\pspace)$. The only difference is that this time we need to keep track of left- and right-variables separately for some of our arguments. Specifically, the estimates analogously to to \eqref{appendix:oscillatory_integrals:eqn:proof_Hoermander_classes_in_Moyal_space_estimate_oscillatory_integrand} and \eqref{appendix:oscillatory_integrals:eqn:proof_Hoermander_classes_in_Moyal_space_estimate_oscillatory_integrand-all-variables-separated} are 
			\begin{align*}
				\babs{\mbox{$\displaystyle \int$grand}} &\leq C \, q^{m_L,m_R}_{N(n,\nu)}(F) \; \Bigl ( \max_{\sabs{b} + \sabs{\beta} \leq N(n,\nu)} \snorm{g}_{\Schwartz,J \tilde{Q} b \beta} \Bigr ) 
				\, \cdot \notag \\
				&\quad \quad
				\sexpval{y_L}^{-2N_L + \sabs{c} + \sabs{c'_L} + \sabs{c'_R} + \sabs{c'}} \, \sexpval{y_R}^{-2N_R+\sabs{c} + \sabs{c'_L} + \sabs{c'_R} + \sabs{c'}}
				\sexpval{z}^{-2N + \sabs{c} + \sabs{c'_L} + \sabs{c'_R} + \sabs{c'}} 
				\, \cdot \notag \\
				&\quad \quad
				\bexpval{x + \tfrac{\eps}{2}(y_L-y_R)}^{-J} \, \sexpval{\eta_L}^{-2K_L} \, \sexpval{\eta_R}^{-2K_R} \, \sexpval{\zeta}^{-2K} \, 
				\, \cdot \notag \\
				&\quad \quad
				\sexpval{\xi - \eta_L}^{m_R - \rho \, (\sabs{\alpha_R} + \sabs{\alpha'_R})} \, \sexpval{\xi - \eta_R}^{m_L - \rho \, (\sabs{\alpha_L} + \sabs{\alpha'_L})} \, \sexpval{\xi - \zeta}^{- Q - m_L - m_R} 
				\\
				&\leq C' \, q^{m_L,m_R}_{N(n,\nu)}(F) \; \Bigl ( \max_{\sabs{b} + \sabs{\beta} \leq N(n,\nu)} \snorm{g}_{\Schwartz,J \tilde{Q} b \beta} \Bigr ) \, \sexpval{x}^{-J} \, \sexpval{\xi}^{-Q} 
				\, \cdot \notag \\
				&\quad \quad 
				\sexpval{y_L}^{-2N_L + \abs{n} + 2(K_L+K+K_R) + J} \, \sexpval{y_R}^{-2N_R + \abs{n} + 2(K_L+K+K_R) + J} \, \sexpval{z}^{-2N+\abs{n}+2(K_L+K+K_R)} 
				\, \cdot \notag \\
				&\quad \quad 
				\sexpval{\eta_L}^{-2K_L + \abs{m_R}} \, \sexpval{\eta_R}^{-2K_R + \abs{m_L}} \, \sexpval{\zeta}^{-2K + Q + \abs{m_L} + \abs{m_R}} 
			\end{align*}
			for the same $N(n,\nu)$ that we used in the proof of (1) and suitable constants $C , C' > 0$ that can be chosen independently of $\eps \in [0,1]$. Here we have set $\tilde{Q} = Q + m_L + m_R$. Note that it is \eg $m_L$ that appears in the exponent of $\sexpval{\eta_R}$, because the roles of left and right variables in $F^{\mathrm{t}}(X_L,X_R) = F(X_R,X_L)$ are reversed. 
		
			Summing over all terms, there are only finitely many of them, we arrive at
			\begin{align*}
				\babs{\partial_x^n \partial_{\xi}^{\nu} F^{\mathrm{t}} \semisuper^B g(X)} \leq C(n,\nu,J,Q) \, q^{m_L,m_R}_{N(n,\nu)}(F) \; \Bigl ( \max_{\sabs{b} + \sabs{\beta} \leq N(n,\nu)} \snorm{g}_{\Schwartz,J \tilde{Q} b \beta} \Bigr ) \, \sexpval{x}^{-J} \, \sexpval{\xi}^{-Q} 
			\end{align*}
			where the constant only depends on $n$, $\nu$, $J$ and $Q$. That shows $g \mapsto F^{\mathrm{t}} \semisuper^B g$ is a continuous linear map $\Schwartz(\pspace) \longrightarrow \Schwartz(\pspace)$, and therefore $S^{m_L,m_R}_{\rho,0}(\Pspace)$ lies in the magnetic semi-super Moyal space $\ssMoyalSpace$. 
		\end{enumerate}
	\end{proof}
	The next lemma covers the oscillatory integral that defines the semi-super product $F\semisuper^B g$ of Hörmander symbols $F$ and $g$. But before we proceed, we make the following remark.
	\begin{remark}
		In fact, this lemma actually contains a stronger result than necessary in this paper. The case $a_L = a_R = b = \alpha_L = \alpha_R = \beta$ and $\tau = \tau' = 1$ in Lemma~\ref{appendix:oscillatory_integrals:lem:semi_super_product_existence_oscillatory_integral} is enough for our purpose in this paper, and we will use the full strength of this lemma to derive the asymptotic expansion of $F\semisuper^B g$ in the forthcoming work. 
	\end{remark}
	\begin{lemma}\label{appendix:oscillatory_integrals:lem:semi_super_product_existence_oscillatory_integral}
		Suppose the magnetic field satisfies Assumption~\ref{intro:assumption:bounded_magnetic_field}, and that the constants satisfy $\rho , \tau , \tau' \in [0,1]$, $\eps \in (0,1]$ and $m , m_L , m_R , m' \in \R$. Moreover, assume that the function 
		\begin{align*}
			[0,1] \ni \tau' \mapsto G_{\tau'} \in \Cont^{\infty}_{\mathrm{b}} \bigl ( \R_x^d \, , \, \Cont^{\infty}_{\mathrm{pol}}(\R_{y_L}^d \times \R_{y_R}^d \times \R_z^d) \bigr ) 
		\end{align*}
		depends on $\tau'$ in a continuous fashion and that for all $c , c_L , c_R , c'\in \N_0^d$, there exists $C_{c c_L c_R c'} > 0$ such that
		\begin{align} 
			\abs{\partial_x^c \partial_{y_L}^{c_L} \partial_{y_R}^{c_R} \partial_z^{c'} G_{\tau'}(x,y_L,y_R,z)} \leq C_{c c_L c_R c'} \bigl ( \sexpval{y_L} + \sexpval{y_R} + \sexpval{z} \bigr )^{\sabs{c} + \sabs{c_L} + \sabs{c_R} + \sabs{c'}} 
			.
			\label{appendix:oscillatory_integrals:eqn:semi_super_product_estimate_oscillatory_integral}
		\end{align}
		Then for a function $F$ on $\Pspace$ and a function $g$ on $\pspace$, $a_L , a_R , b , \alpha_L , \alpha_R , \beta \in \N_0^d$ and $\tau , \tau' \in [0,1]$, we define the oscillatory integral as
		\begin{align}
			I_{\tau\tau'}(X) := \frac{1}{(2\pi)^{3d}} \int_{\Pspace} \dd \Ybf \int_{\pspace} \dd Z \, &\e^{+ \ii \sigma(X,Z+Y_L+Y_R)} \, \e^{+ \ii \tau \frac{\eps}{2} \sigma(Y_L+Z,Y_R+Z)}
			\, \cdot 
			\label{appendix:oscillatory_integrals:eqn:semi_super_product_oscillatory_integral}
			\\
			& 
			G_{\tau'}(x,y_L,y_R,z) \, y_L^{a_L} \, \eta_L^{\alpha_L} \, y_R^{a_R} \, \eta_R^{\alpha_R} \, z^b \, \zeta^{\beta} \, (\Fourier_{\Sigma} F)(\Ybf) \, (\Fourier_{\sigma} g)(Z) 
			. 
			\notag 
		\end{align}
		Then in the following circumstances the oscillatory integral exists: 
		\begin{enumerate}[(1)]
			\item The map 
			\begin{align*}
				(F,g) \mapsto I_{\tau \tau'} : S_{\rho,0}^m(\Pspace) \times S_{\rho,0}^{m'}(\pspace) \longrightarrow S_{\rho,0}^{m + m' - \rho (\abs{a_L} + \abs{a_R} + \abs{b})}(\pspace)
			\end{align*}
			is a continuous bilinear map. Furthermore, given symbols $F\in S_{\rho,0}^m(\Pspace)$ and $g\in S_{\rho,0}^{m'}(\pspace)$, we obtain a continuous map $[0,1] \times [0,1] \ni (\tau,\tau')\mapsto I_{\tau\tau'}\in S_{\rho,0}^{m + m' - \rho (\abs{a_L} + \abs{a_R} + \abs{b})}(\pspace)$. 
			\item The map 
			\begin{align*}
				(F,g) \mapsto I_{\tau\tau'} : S_{\rho,0}^{m_L,m_R}(\Pspace) \times S_{\rho,0}^m(\pspace) \longrightarrow S_{\rho,0}^{m + m_L + m_R - \rho (\sabs{a_L} + \sabs{a_R} + \abs{b})}(\pspace)
			\end{align*}
			is a continuous bilinear map. Furthermore, given symbols $F\in S_{\rho,0}^{m_L,m_R}(\Pspace)$ and $g\in S_{\rho,0}^m(\pspace)$, we obtain a continuous map $[0,1] \times [0,1] \ni (\tau,\tau')\mapsto I_{\tau\tau'}\in S_{\rho,0}^{m + m_L + m_R - \rho (\sabs{a_L} + \sabs{a_R} + \abs{b})}(\pspace)$.
		\end{enumerate}
	\end{lemma}
	\begin{proof}
		\begin{enumerate}[(1)]
			\item Let $F\in S_{\rho,0}^m(\Pspace)$ and $g\in S_{\rho,0}^{m'}(\pspace)$. The given oscillatory integral~\eqref{appendix:oscillatory_integrals:eqn:semi_super_product_oscillatory_integral} can be rewritten as
			\begin{align*}
				I_{\tau\tau'}(X) &= \frac{1}{(2\pi)^{6d}} \int_{\Pspace} \dd \Ybf \int_{\pspace} \dd Z \int_{\Pspace} \dd \Ybf' \int_{\pspace} \dd Z' \, \bigl ( (+ \ii \partial_{\eta_L'})^{a_L} (- \ii \partial_{y_L'})^{\alpha_L} \e^{+ \ii \sigma(X-Y_L',Y_L)} \bigr )
				\, \cdot \\
				&\quad \quad
				\bigl ( (+\ii\partial_{\zeta'})^b (-\ii\partial_{z'})^{\beta} \e^{+ \ii \sigma(X-Z',Z)} \bigr ) \, \bigl ( (+ \ii \partial_{\eta_R'})^{a_R} (- \ii \partial_{y_R'})^{\alpha_R} \e^{+ \ii \sigma(X-Y_R',Y_R)} \bigr )
				\, \cdot \\
				&\quad \quad
				\e^{+ \ii \tau\frac{\eps}{2}\sigma(Y_L+Z,Y_R+Z)} G_{\tau'}(x,y_L,y_R,z) \, F(\Ybf') \, g(Z') 
				\\
				&= \frac{1}{(2\pi)^{6d}} \int_{\Pspace} \dd \Ybf \int_{\pspace} \dd Z \int_{\Pspace} \dd \Ybf' \int_{\pspace} \dd Z' \, \e^{+ \ii \sigma(X-Y_L',Y_L)} \, \e^{+ \ii \sigma(X-Z',Z)} \, \e^{+ \ii \sigma(X-Y_R',Y_R)}
				\, \cdot \\
				&\quad \quad
				\e^{+ \ii \tau \frac{\eps}{2} \sigma(Y_L+Z,Y_R+Z)} \, G_{\tau'}(x,y_L,y_R,z)
				\, \cdot \\
				&\quad \quad
				\bigl ( (-\ii\partial_{\eta_L'})^{a_L} (+ \ii \partial_{y_L'})^{\alpha_L} (- \ii \partial_{\eta_R'})^{a_R} (+ \ii \partial_{y_R'})^{\alpha_R} F \bigr )(\Ybf') \; \bigl ( (- \ii \partial_{\zeta'})^b (+ \ii \partial_{z'})^{\beta} g \bigr )(Z') 
				.
			\end{align*}
			Since the functions $\partial_{x_L}^{\alpha_L} \partial_{\xi_L}^{a_L} \partial_{x_R}^{\alpha_R} \partial_{\xi_R}^{a_R} F \in S_{\rho,0}^{m - \rho ( \sabs{a_L} + \abs{a_R} )}(\Pspace)$ and $\partial_x^{\beta} \partial_{\xi}^b g \in S_{\rho,0}^{m' - \rho \abs{b}}(\pspace)$ are just Hörmander symbols of lower ($\rho > 0$) or the same order ($\rho = 0$), without loss of generality we may set the multi indices to $a_L = a_R = b = \alpha_L = \alpha_R = \beta = 0 \in \N_0^d$. 
			
			We need to show the existence of not just $I_{\tau\tau'}(X)$, but all of its partial derivatives as well. For $n , \nu \in \N_0^d$ the partial derivatives $\partial_x^n \partial_{\xi}^{\nu} I_{\tau\tau'}(X)$ distribute and give a finite sum of terms of the type 
			\begin{align*}
				\frac{1}{(2\pi)^{6d}} &\int_{\Pspace} \dd \Ybf \int_{\pspace} \dd Z \int_{\Pspace} \dd \Ybf' \int_{\pspace} \dd Z'
				\, \cdot 
				\notag \\
				&\quad \quad
				\bigl ( \partial_x^{a_L} \partial_{\xi}^{\alpha_L} \e^{+ \ii \sigma(X-Y_L',Y_L)} \bigr ) \, \bigl ( \partial_x^b \partial_{\xi}^{\beta} \e^{+ \ii \sigma(X-Z',Z)} \bigr ) \, \bigl ( \partial_x^{a_R} \partial_{\xi}^{\alpha_R} \e^{+ \ii \sigma(X-Y_R',Y_R)} \bigr )
				\, \cdot 
				\notag \\
				&\quad \quad
				\e^{+ \ii \tau \frac{\eps}{2} \sigma(Y_L+Z,Y_R+Z)} \, \partial_x^c G_{\tau'}(x,y_L,y_R,z) \, F(\Ybf') \, g(Z')
			\end{align*}
			with binomial coefficients where the multi indices have to satisfy $a_L + a_R + b + c = n$ and $\alpha_L + \alpha_R + \beta = \nu$. Taking integration by parts, writing out the phase factors explicitly and collecting them so that the non-prime variables factor out and integrating over these non-prime variables to get Dirac deltas, we can simplify these terms to 
			\begin{align}
				\frac{1}{(2\pi)^{6d}} &\int_{\Pspace} \dd \Ybf \int_{\pspace} \dd Z \int_{\Pspace} \dd \Ybf' \int_{\pspace} \dd Z'
				\, \cdot 
				\notag \\
				& \quad \quad
				\e^{-\ii \left( x-\frac{\tau\eps}{2}(y_R+z)-y_L' \right ) \cdot \eta_L} \, \e^{-\ii \left( x+\frac{\tau\eps}{2}(y_L+z)-y_R' \right ) \cdot \eta_R} \, \e^{-\ii \left( x+\frac{\tau\eps}{2}(y_L-y_R)-z' \right )\cdot\zeta}
				\, \cdot 
				\notag \\
				&\quad \quad
				\e^{+ \ii (\xi-\eta_L') \cdot y_L} \, \e^{+ \ii (\xi-\eta_R') \cdot y_R} \, \e^{+ \ii (\xi-\zeta') \cdot z} \, \partial_x^c G_{\tau'}(x,y_L,y_R,z)
				\, \cdot 
				\notag \\
				&\quad \quad
				\bigl ( \partial_{y_L'}^{a_L} \partial_{\eta_L'}^{\alpha_L} \partial_{y_R'}^{a_R} \partial_{\eta_R'}^{\alpha_R} F \bigr )(\Ybf') \, \bigl ( \partial_{z'}^b \partial_{\zeta'}^{\beta} g \bigr )(Z') 
				\notag \\
				&= \frac{1}{(2\pi)^{3d}} \int_{\Pspace} \dd \Ybf \int_{\pspace} \dd Z
				\, 
				\e^{+ \ii (y_L \cdot \eta_L + y_R \cdot \eta_R + z \cdot \zeta)} \, \partial_x^c G_{\tau'}(x,y_L,y_R,z)
				\, \cdot 
				\notag \\
				& \quad \quad
				\bigl ( \partial_{x_L}^{a_L} \partial_{\xi_L}^{\alpha_L} \partial_{x_R}^{a_R} \partial_{\xi_R}^{\alpha_R} F \bigr ) \bigl ( x - \tfrac{\tau\eps}{2} (y_R+z) , \xi-\eta_L , x + \tfrac{\tau\eps}{2} (y_L+z) , \xi - \eta_R \bigr )
				\, \cdot 
				\notag \\
				&\quad \quad
				\bigl ( \partial_x^b \partial_{\xi}^{\beta} g \bigr ) \bigl ( x + \tfrac{\tau\eps}{2} (y_L-y_R) , \xi - \zeta \bigr ) 
				. 
				\label{appendix:oscillatory_integrals:eqn:semi_super_product_partial_derivatives_type_term}
			\end{align}
			Now we take six non-negative integers $N_L$, $N_R$, $N$, $K_L$, $K_R$ and $K$ that we will choose large enough later on to ensure integrability. Exploiting the operators $L_{\eta_L}, L_{\eta_R}, L_\zeta, L_{y_L}, L_{y_R}$ and $L_z$ defined as in~\eqref{appendix:oscillatory_integrals:eqn:L_operator} and taking integration by parts in the same way as in the proof of Lemma~\ref{super_calculus_extension_by_duality:lem:Hoermander_symbol_classes_contained_in_semi_super_Moyal_space}, we can express each of the terms~\eqref{appendix:oscillatory_integrals:eqn:semi_super_product_partial_derivatives_type_term} as yet another finite sum of $(\tau\eps)^{\sabs{a_L''} + \sabs{a_L'''} + \sabs{a_R''} + \sabs{a_R'''} + \sabs{b''} + \sabs{b'''}}$ times terms of the form
			\begin{align}
				\int_{\Pspace} \dd \Ybf &\int_{\pspace} \dd Z \, \e^{+ \ii (y_L \cdot \eta_L + y_R \cdot \eta_R + z \cdot \zeta)}
				\, \sexpval{y_L}^{-2N_L} \, \sexpval{y_R}^{-2N_R} \, \sexpval{z}^{-2N} \, \sexpval{\eta_L}^{-2K_L} \, \sexpval{\eta_R}^{-2K_R} \, \sexpval{\zeta}^{-2K}
				\, \cdot \notag \\
				&
				\varphi_{N_L a_L'}(y_L) \, \varphi_{N_R a_R'}(y_R) \, \varphi_{N b'}(z) \; \partial_x^c \partial_{y_L}^{c_L} \partial_{y_R}^{c_R} \partial_z^{c'} G_{\tau'}(x,y_L,y_R,z)
				\, \cdot \notag \\
				&
				\bigl ( \partial_{x_L}^{a_L + a''_L + a'''_L} \partial_{\xi_L}^{\alpha_L + \alpha'_L} \partial_{x_R}^{a_R + a_R'' + a_R'''} \partial_{\xi_R}^{\alpha_R + \alpha'_R} F \bigr ) \bigl ( x - \tfrac{\tau\eps}{2} (y_R+z) , \xi - \eta_L , x + \tfrac{\tau\eps}{2} (y_L+z) , \xi-\eta_R \bigr )
				\cdot \notag \\
				&
				\bigl ( \partial_x^{b + b'' + b'''} \partial_{\xi}^{\beta + \beta'} g \bigr ) \bigl ( x + \tfrac{\tau\eps}{2} (y_L-y_R) , \xi - \zeta \bigr ) 
				.
				\label{appendix:oscillatory_integrals:eqn:semi_super_product_partial_derivatives_type_term-each-summand}
			\end{align}
			Here $\varphi_{N\alpha}$, $N\in\N_0$, $\alpha\in\N_0^d$, are bounded and smooth functions described in~\eqref{appendix:oscillatory_integrals:eqn:derivatives_Japanese_bracket}. Importantly for our book keeping later on, the sum is over all multi indices satisfying 
			\begin{align*}
				N_L &\geq \tfrac{1}{2} \sabs{\alpha'_L}
				, 
				&
				N_R &\geq \tfrac{1}{2} \sabs{\alpha'_R}
				, 
				&
				N &\geq \tfrac{1}{2} \sabs{\beta}
				, 
				\\
				K_L &\geq \tfrac{1}{2} \babs{a'_L + a''_R + b'' + c_L} 
				, 
				&
				K_R &\geq \tfrac{1}{2} \babs{a'''_L + a'_R + b''' + c_R} 
				, 
				&
				K &\geq \tfrac{1}{2} \babs{a''_L + a'''_R + b' + c'} 
				. 
			\end{align*}
			To show that once we choose these six integers large enough, the integrand is indeed absolutely integrable, we plug in the \emph{a priori} estimate~\eqref{appendix:oscillatory_integrals:eqn:semi_super_product_estimate_oscillatory_integral} for $G_{\tau'}$ and estimate the Hörmander symbols appropriately: 
			\begin{align*}
				&\sexpval{y_L}^{-2N_L} \, \sexpval{y_R}^{-2N_R} \, \sexpval{z}^{-2N} \, \sexpval{\eta_L}^{-2K_L} \, \sexpval{\eta_R}^{-2K_R} \, \sexpval{\zeta}^{-2K} 
				\, \cdot \\
				&\qquad 
				\Babs{\varphi_{N_L a_L'}(y_L) \, \varphi_{N_R a_R'}(y_R) \, \varphi_{N b'}(z) \; \partial_x^c \partial_{y_L}^{c_L} \partial_{y_R}^{c_R} \partial_z^{c'} G_{\tau'}(x,y_L,y_R,z)}
				\, \cdot \\
				&\qquad
				\Bigl \rvert \bigl ( \partial_{x_L}^{a_L + a''_L + a'''_L} \partial_{\xi_L}^{\alpha_L + \alpha'_L} \partial_{x_R}^{a_R + a_R'' + a_R'''} \partial_{\xi_R}^{\alpha_R + \alpha'_R} F \bigr ) \Bigr . 
				\\
				&\qquad \qquad \qquad  
				\Bigl . \bigl ( x - \tfrac{\tau\eps}{2} (y_R+z) , \xi - \eta_L , x + \tfrac{\tau\eps}{2} (y_L+z) , \xi-\eta_R \bigr ) \Bigr \rvert 
				\, \cdot \\
				&\qquad
				\Babs{\bigl ( \partial_x^{b + b'' + b'''} \partial_{\xi}^{\beta + \beta'} g \bigr ) \bigl ( x + \tfrac{\tau\eps}{2} (y_L-y_R) , \xi - \zeta \bigr )} \\
				&\leq C \, p_{L(n,\nu)}^m(F) \, p_{L(n,\nu)}^{m'}(g) 
				\, \cdot \\
				&\qquad
				\sexpval{y_L}^{-2N_L + \sabs{c} + \sabs{c_L} + \sabs{c_R} + \sabs{c'}} \, \sexpval{y_R}^{-2N_R + \sabs{c} + \sabs{c_L} + \sabs{c_R} + \sabs{c'}} \, \sexpval{z}^{-2N + \sabs{c} + \sabs{c_L} + \sabs{c_R} + \sabs{c'}} 
				\, \cdot \\
				&\qquad
				\sexpval{\eta_L}^{-2K_L} \, \sexpval{\eta_R}^{-2K_R} \, \sexpval{\zeta}^{-2K} 
				\, \cdot \\
				&\qquad
				\bexpval{(\xi - \eta_L , \xi - \eta_R)}^{m - \rho (\sabs{\alpha_L} + \sabs{\alpha_R} + \sabs{\alpha_L'} + \sabs{\alpha_R'})} \; \sexpval{\xi - \zeta}^{m' - \rho(\sabs{\beta} + \sabs{\beta'})} .
			\end{align*}
			The constant $C$ can be chosen so that it holds for all of the relevant multi indices and is uniform in $\tau \in [0,1]$. A little bit of book keeping reveals that if we want this estimate to hold for all relevant multi indices we may pick 
			\begin{align*}
				L(n,\nu) = \sabs{n} + \sabs{\nu} + 2 (N_L + N_R + N + K_L + K_R + K)
			\end{align*}
			for the Fréchet seminorms~\eqref{magnetic_Weyl_calculus:eqn:max_seminorm_Hoermander_symbols} of $F$ and $g$. 
			
			For the next step we will need two inequalities: the first is \eqref{symbol_super_calculus:eqn:inequality_nesting_Hoermander_super_Hoermander_classes_m_geq_0}, and the second one is~\eqref{appendix:oscillatory_integrals:eqn:Japanese_bracket_xi_xi_squeezing_estimate_Japanese_bracket_xi}. Combining these with the estimates
			\begin{gather*}
			    \bexpval{(\xi - \eta_L , \xi - \eta_R)}^{m - \rho (\sabs{\alpha_L} + \sabs{\alpha_R} + \sabs{\alpha_L'} + \sabs{\alpha_R'})} \leq \bexpval{(\xi - \eta_L , \xi - \eta_R)}^{m - \rho (\sabs{\alpha_L} + \sabs{\alpha_R})} , \\
			    \sexpval{\xi - \zeta}^{m' - \rho(\sabs{\beta} + \sabs{\beta'})} \leq \sexpval{\xi - \zeta}^{m' - \rho\sabs{\beta}}
			\end{gather*}
			and Peetre's inequality $\sexpval{\xi - \zeta}^m \leq 2^{\nicefrac{\sabs{m}}{2}} \, \sexpval{\xi}^m \, \sexpval{\zeta}^{\sabs{m}}$ for $m \in \R$ allows us to collect $\sexpval{\xi}$ (at the expense of having to increase $K_L$, $K_R$ and $K$ later on), leading to the estimate 
			\begin{align*}
				\ldots &\leq C' \, p_{L(n,\nu)}^m(F) \, p_{L(n,\nu)}^{m'}(g) \, \sexpval{\xi}^{m + m' - \rho\sabs{\nu}}
				\, \cdot \\
				&\qquad
				\sexpval{y_L}^{-2N_L + \sabs{c} + \sabs{c_L} + \sabs{c_R} + \sabs{c'}} \, \sexpval{y_R}^{-2N_R + \sabs{c} + \sabs{c_L} + \sabs{c_R} + \sabs{c'}} \, \sexpval{z}^{-2N + \sabs{c} + \sabs{c_L} + \sabs{c_R} + \sabs{c'}} 
				\, \cdot \\
				&\qquad
				\sexpval{\eta_L}^{-2K_L + \sabs{m - \rho (\sabs{\alpha_L} + \sabs{\alpha_R})}} \, \sexpval{\eta_R}^{-2K_R + \sabs{m - \rho (\sabs{\alpha_L} + \sabs{\alpha_R})}} \, \sexpval{\zeta}^{-2K + \sabs{m' - \rho\sabs{\beta}}} 
			\end{align*}
			with a modified constant $C'$ that absorbed $C$ and $\sqrt{2}$ raised to an appropriate power. Here we have also used $\nu = \alpha_L + \alpha_R + \beta$, which was mentioned at the early stage of the proof.
			
			Now if we choose $N_L$, $N_R$, $N$, $K_L$, $K_R$ and $K$ so large that the last six exponents are smaller than $-d$ for all relevant choices of the multi indices, the latter being the dimension of configuration space $\R^d$, then we can make this term absolutely integrable for all relevant choices of our multi indices, 
			\begin{align*}
				K_L , K_R &> \tfrac{1}{2} (d + \sabs{m} + \rho \sabs{\nu}) 
				, \; 
				\\
				K &> \tfrac{1}{2} (d + \sabs{m'} + \rho \sabs{\nu}) 
				, 
				\\
				N_L , N_R , N &> K_L + K_R + K + \tfrac{1}{2} (d + \sabs{n})
				. 
			\end{align*}
			Once integrated and summed up, this leads us to the estimate 
			\begin{align}
				\babs{ \partial_x^n \partial_{\xi}^{\nu} I_{\tau\tau'}(X) } \leq C'' \, p_{L(n,\nu)}^m(F) \, p_{L(n,\nu)}^{m'}(g) \, \sexpval{\xi}^{m + m' - \rho\sabs{\nu}} 
				\label{appendix:oscillatory_integrals:eqn:semi_super_product_Hoermander_symbol_super_symbol_semi_norm_estimate}
			\end{align}
			where $C'' > 0$ is a suitable constant independent of $\tau , \tau' \in [0,1]$. Consequently, the oscillatory integral  $I_{\tau\tau'}(X)$ exists as a Hörmander symbol of order $m + m'$. More precisely, equation~\eqref{appendix:oscillatory_integrals:eqn:semi_super_product_Hoermander_symbol_super_symbol_semi_norm_estimate} states that $(F , g) \mapsto I_{\tau\tau'}$ is a continuous bilinear map with respect to the Fréchet topologies.
			
			Lastly, to confirm the continuity of $I_{\tau\tau'}$ in $(\tau,\tau')$, we remind ourselves that $\tau$ raised to some non-negative power appears as a prefactor that garnishes \eqref{appendix:oscillatory_integrals:eqn:semi_super_product_partial_derivatives_type_term-each-summand}. Because we may exchange limits and oscillatory integration, the evident continuity of $\tau \mapsto \e^{+ \ii \tau \frac{\eps}{2} \sigma(Y_L + Z , Y_R + Z)}$ and the assumed continuity of $\tau' \mapsto G_{\tau'}$ lead to the continuity of $(\tau,\tau') \mapsto I_{\tau\tau'}$. 
			\item The proof is largely identical, we just need to use the Fréchet seminorms~\eqref{symbol_super_calculus:eqn:max_seminorms_super_Hoermander_symbols} for $F$ in our estimates. More precisely, we may bound the integrand of~\eqref{appendix:oscillatory_integrals:eqn:semi_super_product_partial_derivatives_type_term-each-summand} from above by 
			\begin{align*}
				\ldots \leq C \, &q_{L(n,\nu)}^{m_L,m_R}(F) \, p_{L(n,\nu)}^m(g) \, \sexpval{\xi}^{m_L+m_R+m-\rho\abs{\nu}} 
				\, \cdot \\
				&
				\sexpval{y_L}^{-2N_L + \abs{n} + 2(K_L+K_R+K)} \, \sexpval{y_R}^{-2N_R + \abs{n} + 2(K_L+K_R+K)} \, \sexpval{z}^{-2N + \abs{n} + 2(K_L+K_R+K)}
				\, \cdot \\
				&
				\sexpval{\eta_L}^{-2K_L + \abs{m_L} + \rho \abs{\nu}} \, \sexpval{\eta_R}^{-2K_R + \abs{m_R} + \rho \abs{\nu}} \,  \sexpval{\zeta}^{-2K + \abs{m}+\rho\abs{\nu}} 
				.
			\end{align*}
			Here $C > 0$ is a constant that depends on the multi indices and is independent of $\tau$, and $L(n,\nu)$ is the same constant as above. This means we arrive at slightly different values for $N_L$, $N_R$, $N$, $K_L$, $K_R$ and $K$ than in (1), we can make the right-hand side absolutely integrable: we first need to fix $K_L$, $K_R$ and $K$ so that they satisfy 
			\begin{align*}
				K_{L,R} &> \tfrac{1}{2} (d + \sabs{m_{L,R}} + \rho \sabs{\nu}) 
				,
				\\
				K &> \tfrac{1}{2} (d + \abs{m} + \rho \sabs{\nu}) 
				, 
			\end{align*}
			and then choose $N_L$, $N_R$ and $N$ for which 
			\begin{align*}
				N_L , N_R , N &> \tfrac{1}{2} (d + \sabs{n}) + K_L + K_R + K 
			\end{align*}
			holds. All this shows that there is a constant $C' > 0$ such that
			\begin{align*}
				\sexpval{\xi}^{-(m_L+m_R+m)+\rho\abs{\nu}} \, \babs{\partial_x^n \partial_{\xi}^{\nu} I_{\tau\tau'}(X)} \leq C' \, q_{L(n,\nu)}^{m_L,m_R}(F) \, p_{L(n,\nu)}^m(g) 
			\end{align*}
			holds for all $X \in \pspace$. As before, this shows the existence of the oscillatory integral $I_{\tau\tau'} \in S^{m_L + m_R + m}_{\rho,0}(\pspace)$ as a Hörmander symbol and the continuity of the map $(F,g) \mapsto I_{\tau\tau'}$. The continuity of $(\tau,\tau') \mapsto I_{\tau\tau'}$ can be deduced as in (1) from the continuity of the phase factor in $\tau$ and the continuity of $\tau' \mapsto G_{\tau'}$. This proves~(2) and completes the proof.
		\end{enumerate}
	\end{proof}
	%

	\subsection{Magnetic super Weyl product $\super^B$} 
	\label{appendix:oscillatory_integrals:super_weyl_product}
	First, we need to furnish the proof that Hörmander symbol classes lie in the magnetic super Moyal algebra. 
	\begin{proof}[Lemma~\ref{symbol_super_calculus:lem:Hoermander_symbols_in_super_Moyal_algebra}]
		\begin{enumerate}[(1)]
			\item Suppose $F \in S^m_{\rho,0}(\Pspace)$ is a Hörmander symbol and $G \in \Schwartz(\Pspace)$ a Schwartz function. We need to show that both $F \super^B G$ and $G \super^B F$ are Schwartz functions. Since the strategy of the proof is identical to that of Lemma~\ref{super_calculus_extension_by_duality:lem:Hoermander_symbol_classes_contained_in_semi_super_Moyal_space}, we will be more brief.
		
			To show that $F \super^B G \in \Schwartz(\Pspace)$, we need to bound all of the the Schwartz seminorms
			\begin{align*}
				\bnorm{F \super^B G}_{\Schwartz,J_L Q_L J_R Q_R n_L \nu_L n_R \nu_R} := \sup_{\Xbf \in \Pspace} \sexpval{x_L}^{J_L} \, \sexpval{\xi_L}^{Q_L} \, \sexpval{x_R}^{J_R} \, \sexpval{\xi_R}^{Q_R} \Babs{ \partial_{x_L}^{n_L} \partial_{\xi_L}^{\nu_L} \partial_{x_R}^{n_R} \partial_{\xi_R}^{\nu_R} (F \super^B G)(\Xbf)}
				, 
			\end{align*}
			which are indexed by $J_L, Q_L, J_R, Q_R\in\R$ and the multi indices $n_L , \nu_L , n_R , \nu_R\in \N_0^d$.
		
			When distributing the derivative $\partial_{x_L}^{n_L} \partial_{\xi_L}^{\nu_L} \partial_{x_R}^{n_R} \partial_{\xi_R}^{\nu_R}$ of $F \super^B G$, by using~\eqref{magnetic_super_PsiDOs:eqn:magnetic_super_Weyl_product} we get a finite sum of terms of the type
			\begin{align*}
				\frac{1}{(2\pi)^{8d}} &\int_{\Pspace} \dd \Ybf \int_{\Pspace} \dd \Zbf \int_{\Pspace} \dd \Ybf' \int_{\Pspace} \dd \Zbf'
				\, \cdot \\
				&\quad \quad
				\bigl ( \partial_{x_L}^{a_L} \partial_{\xi_L}^{\alpha_L} \partial_{x_R}^{a_R} \partial_{\xi_R}^{\alpha_R} \e^{+ \ii \Sigma (\Xbf-\Ybf',\Ybf)} \bigr ) \, \bigl ( \partial_{x_L}^{b_L} \partial_{\xi_L}^{\beta_L} \partial_{x_R}^{b_R} \partial_{\xi_R}^{\beta_R} \e^{+ \ii \Sigma (\Xbf-\Zbf',\Zbf)} \bigr )
				\, \cdot \\
				&\quad \quad
				\e^{+ \ii \frac{\eps}{2} \Sigma (r(\Ybf),\Zbf)} \, \partial_{x_L}^{c_L} \e^{- \ii \lambda \trifluxp(x_L,y_L,z_L)} \, \partial_{x_R}^{c_R} \e^{- \ii \lambda \trifluxp(x_R,z_R,y_R)} \, F(\Ybf') \, G(\Zbf')
				\\
				&= \frac{1}{(2\pi)^{8d}} \int_{\Pspace} \dd \Ybf \int_{\Pspace} \dd \Zbf \int_{\Pspace} \dd \Ybf' \int_{\Pspace} \dd \Zbf' \, \e^{+ \ii \Sigma (\Xbf-\Ybf',\Ybf)} \, \e^{+ \ii \Sigma (\Xbf-\Zbf',\Zbf)} \, \e^{+ \ii \frac{\eps}{2} \Sigma (r(\Ybf),\Zbf)} 
				\, \cdot \\
				&\quad \quad
				\partial_{x_L}^{c_L} \e^{- \ii \lambda \trifluxp(x_L,y_L,z_L)} \, \partial_{x_R}^{c_R} \e^{- \ii \lambda \trifluxp(x_R,z_R,y_R)}
				\, \cdot \\
				&\quad \quad 
				\partial_{y_L'}^{a_L} \partial_{\eta_L'}^{\alpha_L} \partial_{y_R'}^{a_R} \partial_{\eta_R'}^{\alpha_R} F(\Ybf') \, \partial_{z_L'}^{b_L} \partial_{\zeta_L'}^{\beta_L} \partial_{z_R'}^{b_R} \partial_{\zeta_R'}^{\beta_R} G(\Zbf') 
				\\
				&= \frac{1}{(2\pi)^{8d}} \int_{\Pspace} \dd \Ybf \int_{\Pspace} \dd \Zbf \int_{\Pspace} \dd \Ybf' \int_{\Pspace} \dd \Zbf'
				\, \cdot \\
				&\quad \quad
				\e^{- \ii (x_L-y_L'-\frac{\eps}{2}z_L) \cdot \eta_L} \, \e^{- \ii (x_L-z_L'+\frac{\eps}{2}y_L) \cdot \zeta_L} \, \e^{- \ii (x_R-y_R'+\frac{\eps}{2}z_R) \cdot \eta_R} \, \e^{- \ii (x_R-z_R'-\frac{\eps}{2}y_R) \cdot \zeta_R}
				\, \cdot \\
				&\quad \quad
				\e^{+ \ii (\xi_L-\eta_L') \cdot y_L} \, \e^{+ \ii (\xi_L-\zeta_L') \cdot z_L} \, \e^{+ \ii (\xi_R-\eta_R') \cdot y_R} \, \e^{+ \ii (\xi_R-\zeta_R') \cdot z_R}
				\, \cdot \\
				&\quad \quad
				\partial_{x_L}^{c_L} \e^{- \ii \lambda \trifluxp(x_L,y_L,z_L)} \, \partial_{x_R}^{c_R} \e^{- \ii \lambda \trifluxp(x_R,z_R,y_R)} \, \partial_{y_L'}^{a_L} \partial_{\eta_L'}^{\alpha_L} \partial_{y_R'}^{a_R} \partial_{\eta_R'}^{\alpha_R} F(\Ybf') \, \partial_{z_L'}^{b_L} \partial_{\zeta_L'}^{\beta_L} \partial_{z_R'}^{b_R} \partial_{\zeta_R'}^{\beta_R} G(\Zbf') 
				\\
				&= \frac{1}{(2\pi)^{4d}} \int_{\Pspace} \dd \Ybf \int_{\Pspace} \dd \Zbf \; \e^{+ \ii (y_L \cdot \eta_L + z_L \cdot \zeta_L + y_R \cdot \eta_R + z_R \cdot \zeta_R)}
				\, \cdot \\
				&\quad \quad
				\partial_{x_L}^{c_L} \e^{- \ii \lambda \trifluxp(x_L,y_L,z_L)} \, \partial_{x_R}^{c_R} \e^{- \ii \lambda \trifluxp(x_R,z_R,y_R)}
				\, \cdot \\
				&\quad \quad
				\partial_{x_L}^{a_L} \partial_{\xi_L}^{\alpha_L} \partial_{x_R}^{a_R} \partial_{\xi_R}^{\alpha_R} F \bigl ( x_L-\tfrac{\eps}{2}z_L , \xi_L-\eta_L , x_R + \tfrac{\eps}{2}z_R , \xi_R - \eta_R \bigr )
				\, \cdot \\
				&\quad \quad
				\partial_{x_L}^{b_L} \partial_{\xi_L}^{\beta_L} \partial_{x_R}^{b_R} \partial_{\xi_R}^{\beta_R} G \bigl ( x_L+\tfrac{\eps}{2}y_L , \xi_L-\zeta_L , x_R-\tfrac{\eps}{2}y_R , \xi_R-\zeta_R \bigr ) 
			\end{align*}
			so that the multi indices satisfy 
			\begin{align*}
				n_{L,R} &= a_{L,R} + b_{L,R} + c_{L,R} 
				, 
				\\
				\nu_{L,R} &= \alpha_{L,R} + \beta_{L,R} 
				. 
			\end{align*}
			With the help of $L$ operators~\eqref{appendix:oscillatory_integrals:eqn:L_operator} we insert $\sexpval{y_L}^{-2N_L}$ and similar factors in the other position variables, and then do the same for the momentum variables. The derivatives distribute and yield a finite sum of terms of the form 
			\begin{align}
				&\int_{\Pspace} \dd \Ybf \int_{\Pspace} \dd \Zbf \, \e^{+ \ii (y_L \cdot \eta_L + z_L \cdot \zeta_L + y_R \cdot \eta_R + z_R \cdot \zeta_R)}
				\, \cdot \notag \\
				&\quad \quad
				\sexpval{y_L}^{-2N_L} \, \sexpval{z_L}^{-2M_L} \, \sexpval{y_R}^{-2N_R} \, \sexpval{z_R}^{-2M_R} \, \sexpval{\eta_L}^{-2D_L} \, \sexpval{\zeta_L}^{-2K_L} \, \sexpval{\eta_R}^{-2D_R} \, \sexpval{\zeta_R}^{-2K_R} 
				\, \cdot \notag \\
				&\quad \quad
				\varphi_{N_L b'_L}(y_L) \, \varphi_{M_L a'_L}(z_L) \, \varphi_{N_R b'_R}(y_R) \, \varphi_{M_R a'_R}(z_R) 
				\, \cdot \notag \\
				&\quad \quad
				\partial_{x_L}^{c_L} \partial_{y_L}^{a''_L} \partial_{z_L}^{b''_L} \e^{- \ii \lambda \trifluxp(x_L,y_L,z_L)} \, \partial_{x_R}^{c_R} \partial_{y_R}^{a''_R} \partial_{z_R}^{b''_R} \e^{- \ii \lambda \trifluxp(x_R,z_R,y_R)}
				\, \cdot \notag \\
				&\quad \quad
				\partial_{x_L}^{a_L + a'''_L} \partial_{\xi_L}^{\alpha_L + \alpha'_L} \partial_{x_R}^{a_R + a'''_R} \partial_{\xi_R}^{\alpha_R + \alpha'_R} F \bigl ( x_L - \tfrac{\eps}{2} z_L , \xi_L - \eta_L , x_R + \tfrac{\eps}{2} z_R , \xi_R - \eta_R \bigr )
				\, \cdot \notag \\
				&\quad \quad
				\partial_{x_L}^{b_L + b'''_L} \partial_{\xi_L}^{\beta_L + \beta'_L} \partial_{x_R}^{b_R + b'''_R} \partial_{\xi_R}^{\beta_R + \beta'_R} G \bigl ( x_L + \tfrac{\eps}{2} y_L , \xi_L - \zeta_L , x_R - \tfrac{\eps}{2}y_R , \xi_R-\zeta_R \bigr ) 
				,
				\label{appendix:oscillatory_integrals:eqn:oscillatory_integral_super_Weyl_product_Schwartz_function_prototypical_term}
			\end{align}
			where $\varphi_{N_L b'_L}(y_L)$ stems from deriving $\sexpval{y_L}^{-2N_L}$ and similarly for the other terms. As before, we will pick large enough integers $N_L$, $N_R$, $M_L$, $M_R$, $D_L$, $D_R$, $K_L$ and $K_R$ later on. The degrees of the newly introduced multi indices are bounded by 
			\begin{align}
			    \label{appendix:oscillatory_integrals:eqn:bound_for_multi_indices_N_L_and_M_L}
				\sabs{\alpha'_L} &\leq 2N_L 
				, 
				&
				\sabs{\beta'_L} &\leq 2M_L 
				, 
				\\ 
				\sabs{\alpha'_R} &\leq 2N_R 
				, 
				&
				\sabs{\beta'_R} &\leq 2M_R
				, 
				\\
				\abs{b'_L + a''_L + b'''_L} &\leq 2D_L 
				, 
				&
				\abs{a'_L + b''_L + a'''_L} &\leq 2K_L 
				, 
				\\ 
				\label{appendix:oscillatory_integrals:eqn:bound_for_multi_indices_D_R_and_K_R}
				\abs{b'_R + a''_R + b'''_R} &\leq 2D_R 
				, 
				&
				\abs{a'_R + b''_R + a'''_R} &\leq 2K_R 
				.
			\end{align}
			We proceed to estimate the integrand as usual. Up to some power of $\eps$, once we fix the non-negative constants $J_L$, $J_R$, $Q_L$ and $Q_R$, in terms of the Fréchet seminorms of $S^m_{\rho,0}(\Pspace)$ and $\Schwartz(\Pspace)$, we can deduce the following upper bound
			\begin{align} 
				\Babs{\mbox{$\displaystyle \int$grand}} &\leq \sexpval{y_L}^{-2N_L} \, \sexpval{y_R}^{-2N_R} \, \sexpval{z_L}^{-2M_L} \, \sexpval{z_R}^{-2M_R} \, \sexpval{\eta_L}^{-2D_L} \, \sexpval{\eta_R}^{-2D_R} \, \sexpval{\zeta_L}^{-2K_L} \, \sexpval{\zeta_R}^{-2K_R} 
				\, \cdot \notag \\
				&\quad \quad
				\Babs{\varphi_{N_L b'_L}(y_L) \, \varphi_{N_R b'_R}(y_R) \, \varphi_{M_L a'_L}(z_L) \, \varphi_{M_R a'_R}(z_R)}
				\, \cdot \notag \\
				&\quad \quad
				\Babs{\partial_{x_L}^{c_L} \partial_{y_L}^{a''_L} \partial_{z_L}^{b''_L} \e^{- \ii \lambda \trifluxp(x_L,y_L,z_L)}} \, \Babs{\partial_{x_R}^{c_R} \partial_{y_R}^{a''_R} \partial_{z_R}^{b''_R} \e^{- \ii \lambda \trifluxp(x_R,z_R,y_R)}}
				\, \cdot \notag \\
				&\quad \quad
				\Babs{\partial_{x_L}^{a_L + a'''_L} \partial_{\xi_L}^{\alpha_L + \alpha'_L} \partial_{x_R}^{a_R + a'''_R} \partial_{\xi_R}^{\alpha_R + \alpha'_R} F \bigl ( x_L - \tfrac{\eps}{2} z_L , \xi_L - \eta_L , x_R + \tfrac{\eps}{2} z_R , \xi_R - \eta_R \bigr )}
				\, \cdot \notag \\
				&\quad \quad
				\Babs{\partial_{x_L}^{b_L + b'''_L} \partial_{\xi_L}^{\beta_L + \beta'_L} \partial_{x_R}^{b_R + b'''_R} \partial_{\xi_R}^{\beta_R + \beta'_R} G \bigl ( x_L + \tfrac{\eps}{2} y_L , \xi_L-\zeta_L , x_R - \tfrac{\eps}{2} y_R , \xi_R - \zeta_R \bigr )}
				\displaybreak[0]
				\notag \\
				&\leq C \, p_N^m(F) \, \Bigl ( \max_{\sabs{b_L} + \sabs{b_R} + \sabs{\beta_L} + \sabs{\beta_R} \leq N} \snorm{G}_{\Schwartz,J_L Q'_L J_R Q'_R b_L \beta_L b_R \beta_R} \Bigr )
				\, \cdot \notag \\
				&\quad \quad
				\sexpval{y_L}^{-2N_L + \sabs{c_L} + \sabs{a''_L} + \sabs{b''_L}} \, \sexpval{y_R}^{-2N_R + \sabs{c_R} + \sabs{a''_R} + \sabs{b''_R}} 
				\, \cdot \notag \\
				&\quad \quad
				\sexpval{z_L}^{-2M_L + \sabs{c_L} + \sabs{a''_L} + \sabs{b''_L}} \, \sexpval{z_R}^{-2M_R + \sabs{c_R} + \sabs{a''_R} + \sabs{b''_R}} \, \sexpval{\eta_L}^{-2D_L} \, \sexpval{\zeta_L}^{-2K_L}
				\, \cdot \notag \\
				&\quad \quad
				\sexpval{\eta_R}^{-2D_R} \, \sexpval{\zeta_R}^{-2K_R} \, \sexpval{(\xi_L - \eta_L , \xi_R - \eta_R)}^{m - \rho(\sabs{\alpha_L} + \sabs{\alpha'_L} + \sabs{\alpha_R} + \sabs{\alpha'_R})}
				\, \cdot \notag \\
				&\quad \quad
				\sexpval{x_L + \tfrac{\eps}{2} y_L}^{-J_L} \, \sexpval{\xi_L - \zeta_L}^{-Q_L - \sabs{m}} \, \sexpval{x_R - \tfrac{\eps}{2} y_R}^{-J_R} \, \sexpval{\xi_R - \zeta_R}^{-Q_R - \sabs{m}}
				.
				\label{appendix:oscillatory_integrals:eqn:estimate_oscillatory_integral_super_Weyl_product_Schwartz}
			\end{align}
			Here we have set $Q'_{L,R} = Q_{L,R} + \sabs{m}$ and
			\begin{align}
				N := \sabs{n_L} + \sabs{n_R} + \sabs{\nu_L} + \sabs{\nu_R} + 2(N_L + N_R + M_L + M_R + D_L + D_R + K_L + K_R)
				. 
				\label{appendix:oscillatory_integrals:eqn:estimate_oscillatory_integral_super_Weyl_product_constant_N}
			\end{align}
			Once we plug in the estimates on the Japanese bracket in the proof of Lemma~\ref{super_calculus_extension_by_duality:lem:Hoermander_symbol_classes_contained_in_semi_super_Moyal_space}~(1) and combine them with 
			\begin{align*}
				\sexpval{(\xi_L - \eta_L , \xi_R - \eta_R)}^{m - \rho (\sabs{\alpha_L} + \sabs{\alpha'_L} + \sabs{\alpha_R} + \sabs{\alpha'_R})} \leq 2^{\frac{\sabs{m}}{2}} \, \sexpval{\xi_L}^{\sabs{m}} \, \sexpval{\xi_R}^{\sabs{m}} \, \sexpval{\eta_L}^{\sabs{m}} \, \sexpval{\eta_R}^{\sabs{m}} 
				, 
			\end{align*}
			we arrive at the estimate 
			\begin{align}
				\ldots &\leq C' \, p_N^m(F) \, \Bigl ( \max_{\sabs{b_L} + \sabs{b_R} + \sabs{\beta_L} + \sabs{\beta_R} \leq N} \snorm{G}_{\Schwartz,J_L Q'_L J_R Q'_R b_L \beta_L b_R \beta_R} \Bigr )
				\, \cdot \notag \\
				&\quad \quad
				\sexpval{x_L}^{-J_L} \, \sexpval{x_R}^{-J_R} \, \sexpval{\xi_L}^{-Q_L} \, \sexpval{\xi_R}^{-Q_R} 
				\, \cdot \notag \\
				&\quad \quad
				\sexpval{y_L}^{-2N_L + \abs{n_L} + 2(D_L+K_L) + J_L} \, \sexpval{y_R}^{-2N_R + \abs{n_R} + 2(D_R+K_R) + J_R} 
				\, \cdot \notag \\
				&\quad \quad
				\sexpval{z_L}^{-2M_L + \abs{n_L} + 2(D_L+K_L)} \, \sexpval{z_R}^{-2M_R + \abs{n_R} + 2(D_R+K_R)} 
				\, \cdot \notag \\
				&\quad \quad
				\sexpval{\eta_L}^{-2D_L + \abs{m}} \, \sexpval{\zeta_L}^{-2K_L + Q_L + \abs{m}} \, \sexpval{\eta_R}^{-2D_R + \abs{m}} \, \sexpval{\zeta_R}^{-2K_R + Q_R + \abs{m}}
			\end{align}
			that no longer depends on the precise value of the multi indices, \ie it holds for all terms in the big, but finite sum. Consequently, we can ensure integrability of all terms if we pick the integrability constants for the spatial variables first, 
			%
			\begin{align*}
				D_{L,R} &> \tfrac{1}{2} (\sabs{m} + d) 
				, 
				\\
				K_{L,R} &> \tfrac{1}{2} (\sabs{m} + d + Q_{L,R}) 
				, 
			\end{align*}
			and then fix the integrability constants in the momentum variables, 
			\begin{align*}
				M_{L,R} &> \tfrac{1}{2} (\sabs{n_{L,R}} + d) + D_{L,R} + K_{L,R} 
				, 
				\\
				N_{L,R} &> \tfrac{1}{2} (\sabs{n_{L,R}} + d + J_{L,R}) + D_{L,R} + K_{L,R} 
				. 
			\end{align*}
			Summing over all terms shows that for $F \super^B G$ and any of the derivatives 
			\begin{align}
				&\Babs{\partial_{x_L}^{n_L} \partial_{\xi_L}^{\nu_L} \partial_{x_R}^{n_R} \partial_{\xi_R}^{\nu_R} F \super^B G(X_L,X_R)} 
				\notag \\
				&\quad \leq 
				C'' \, p_N^m(F) \; \Bigl ( \max_{\sabs{b_L} + \sabs{b_R} + \sabs{\beta_L} + \sabs{\beta_R} \leq N} \snorm{G}_{\Schwartz,J_L Q'_L J_R Q'_R b_L \beta_L b_R \beta_R} \Bigr )
				\, \sexpval{x_L}^{-J_L} \, \sexpval{x_R}^{-J_R} \, \sexpval{\xi_L}^{-Q_L} \, \sexpval{\xi_R}^{-Q_R}
				\label{appendix:oscillatory_integrals:eqn:estimate_oscillatory_integral_super_Weyl_product_Schwartz_final_estimate}
			\end{align}
			holds for arbitrary values of $J_L$, $J_R$, $Q_L$ and $Q_R$, where the constant $C'' > 0$ depends on the order of the derivative and the just mentioned four decay constants. This shows that $F\super^B G\in\Schwartz(\Pspace)$ and hence $F\in\MoyalAlgebra_L^B(\Pspace)$. Similarly, we can also prove that if $F\in S_{\rho,0}^m(\Pspace)$, then $F\in\MoyalAlgebra_R^B(\Pspace)$ by showing that $G\super^B F\in\Schwartz(\Pspace)$ for all $G\in\Schwartz(\Pspace)$. All this shows that $S_{\rho,0}^m(\Pspace)\subset\sMoyalAlg$, which proves the first part.
			\item The proof here is analogous to (1), except that we need to involve the seminorms~\eqref{symbol_super_calculus:eqn:max_seminorms_super_Hoermander_symbols} instead in the estimate~\eqref{appendix:oscillatory_integrals:eqn:estimate_oscillatory_integral_super_Weyl_product_Schwartz}. We just need to keep track of left and right degrees separately. Note that unlike in the proof of Lemma~\ref{super_calculus_extension_by_duality:lem:Hoermander_symbol_classes_contained_in_semi_super_Moyal_space}~(2), left and right places no longer trade places as there we had to deal with $F^{\mathrm{t}}(X_L,X_R) = F(X_R,X_L)$ rather than $F$. 
		
			Suppose $F \in S^{m_L,m_R}_{\rho,0}(\Pspace)$ and $G \in \Schwartz(\Pspace)$. As before, $n_L , \nu_L , n_R , \nu_R \in \N_0^d$ are multi indices for the derivatives of $F \super^B G$ and we pick arbitrary $J_L , J_R , Q_L , Q_R > 0$ in advance. We proceed to estimate the integrand of \eqref{appendix:oscillatory_integrals:eqn:oscillatory_integral_super_Weyl_product_Schwartz_function_prototypical_term} by making use of the inequality~\eqref{symbol_super_calculus:eqn:inequality_nesting_Hoermander_super_Hoermander_classes_m_geq_0}: 
			\begin{align*}
				\Babs{\mbox{$\displaystyle \int$grand}} &\leq C \, q^{m_L,m_R}_N(F) \; \Bigl ( \max_{\sabs{b_L} + \sabs{b_R} + \sabs{\beta_L} + \sabs{\beta_R} \leq N} \snorm{G}_{\Schwartz,J_L \tilde{Q}_L J_R \tilde{Q}_R b_L \beta_L b_R \beta_R} \Bigr ) 
				\, \cdot \\
				&\quad \quad
				\sexpval{y_L}^{-2N_L + \sabs{c_L} + \sabs{a''_L} + \sabs{b''_L}} \, \sexpval{y_R}^{-2N_R + \sabs{c_R} + \sabs{a''_R} + \sabs{b''_R}} 
				\, \cdot \notag \\
				&\quad \quad
				\sexpval{z_L}^{-2M_L + \sabs{c_L} + \sabs{a''_L} + \sabs{b''_L}} \, \sexpval{z_R}^{-2M_R + \sabs{c_R} + \sabs{a''_R} + \sabs{b''_R}} \, \sexpval{\eta_L}^{-2D_L} \, \sexpval{\zeta_L}^{-2K_L}
				\, \cdot \notag \\
				&\quad \quad
				\sexpval{\eta_R}^{-2D_R} \, \sexpval{\zeta_R}^{-2K_R} \, \sexpval{\xi_L-\eta_L}^{m_L - \rho(\sabs{\alpha_L} + \sabs{\alpha'_L})} \, \sexpval{\xi_R-\eta_R}^{m_R - \rho(\sabs{\alpha_R} + \sabs{\alpha'_R})}
				\, \cdot \notag \\
				&\quad \quad
				\sexpval{x_L + \tfrac{\eps}{2}y_L}^{-J_L} \, \sexpval{\xi_L-\zeta_L}^{-Q_L - \sabs{m_L}} \, \sexpval{x_R-\tfrac{\eps}{2}y_R}^{-J_R} \, \sexpval{\xi_R-\zeta_R}^{-Q_R - \sabs{m_R}}
				\\
				&\leq C' \, q^{m_L,m_R}_N(F) \; \Bigl ( \max_{\sabs{b_L} + \sabs{b_R} + \sabs{\beta_L} + \sabs{\beta_R} \leq N} \snorm{G}_{\Schwartz,J_L \tilde{Q}_L J_R \tilde{Q}_R b_L \beta_L b_R \beta_R} \Bigr ) 
				\, \cdot \\
				&\quad \quad
				\, \sexpval{x_L}^{-J_L} \, \sexpval{x_R}^{-J_R} \, \sexpval{\xi_L}^{-Q_L} \, \sexpval{\xi_R}^{-Q_R} 
				\, \cdot \notag \\
				&\quad \quad
				\sexpval{y_L}^{-2N_L + \sabs{n_L} + 2(D_L + K_L) + J_L} \, \sexpval{y_R}^{-2N_R + \sabs{n_R} + 2(D_R + K_R) + J_R}
				\, \cdot \notag \\
				&\quad \quad
				\sexpval{z_L}^{-2M_L + \sabs{n_L} + 2(D_L + K_L)} \, \sexpval{z_R}^{-2M_R + \sabs{n_R} + 2(D_R + K_R)}
				\, \cdot \notag \\
				&\quad \quad
				\sexpval{\eta_L}^{-2D_L + \sabs{m_L}} \, \sexpval{\zeta_L}^{-2K_L + Q_L + \sabs{m_L}} \, \sexpval{\eta_R}^{-2D_R + \sabs{m_R}} \, \sexpval{\zeta_R}^{-2K_R + Q_R + \sabs{m_R}}
				,
			\end{align*}
			where we have set $\tilde{Q}_{L,R} = Q_{L,R} + \sabs{m_{L,R}}$. With the exception of 
			\begin{align*}
				D_{L,R} &> \tfrac{1}{2} (\sabs{m_{L,R}} + d) 
				, 
				\\
				K_{L,R} &> \tfrac{1}{2} (\sabs{m_{L,R}} + d + Q_{L,R})
				, 
			\end{align*}
			the constants $N_{L,R}$, $M_{L,R}$ and $N$ need to be chosen as in the proof of (1). That ensures integrability and overall, we arrive at the conclusion that for any $n_{L,R} , \nu_{L,R} \in \N_0^d$ and $J_{L,R} , Q_{L,R} > 0$ there exists $C'' > 0$ for which \eqref{appendix:oscillatory_integrals:eqn:estimate_oscillatory_integral_super_Weyl_product_Schwartz_final_estimate} with $p_N^m(F)$ replaced by $q_N^{m_L,m_R}(F)$ and $Q'_{L,R}$ replaced by $\tilde{Q}_{L,R}$ holds true for all $(X_L,X_R) \in \Pspace$. That shows that $F\super^B G\in\Schwartz(\Pspace)$, and we can similarly prove that $G\super^B F\in\Schwartz(\Pspace)$. All this shows that $S_{\rho,0}^{m_L,m_R}(\Pspace)\subset\sMoyalAlg$ and finishes the proof.
		\end{enumerate}
	\end{proof}
	\begin{remark}
		The above proof actually shows that the magnetic super Weyl product $\super^B$ actually gives rise to continuous bilinear maps
		\begin{gather*}
		    S_{\rho,0}^m(\Pspace) \times \Schwartz(\Pspace) \longrightarrow \Schwartz(\Pspace), \qquad \Schwartz(\Pspace) \times S_{\rho,0}^m(\Pspace) \longrightarrow \Schwartz(\Pspace), \\
		    S_{\rho,0}^{m_L,m_R}(\Pspace) \times \Schwartz(\Pspace) \longrightarrow \Schwartz(\Pspace), \qquad \Schwartz(\Pspace) \times S_{\rho,0}^{m_L,m_R}(\Pspace) \longrightarrow \Schwartz(\Pspace) .
		\end{gather*}
	\end{remark}
	The fact that Hörmander symbols lie in the magnetic super Moyal algebra implies that we can talk about the magnetic super Weyl product of two Hörmander symbols, which results in a tempered distribution. The proof that this yields a Hörmander symbol is quite similar in strategy to the proof above, we just need to employ different seminorms in our estimates. Before we proceed to the statement and proof of the lemma, we make the following remark:
	\begin{remark}
		Again, as in Lemma~\ref{appendix:oscillatory_integrals:lem:semi_super_product_existence_oscillatory_integral}, the following lemma contains a stronger result than necessary. The case $a_L = a_R = b_L = b_R = \alpha_L = \alpha_R = \beta_L = \beta_R = 0$ and $\tau = \tau' = 1$ in Lemma~\ref{appendix:oscillatory_integrals:lem:existence_oscillatory_integral_super_Weyl_product} is enough for our purpose, and we will use the full strength of this lemma to derive the asymptotic expansion of $F\super^B G$ in the forthcoming work. 
	\end{remark}
	\begin{lemma}\label{appendix:oscillatory_integrals:lem:existence_oscillatory_integral_super_Weyl_product}
		Suppose the magnetic field satisfies Assumption~\ref{intro:assumption:bounded_magnetic_field}, and that the constants satisfy $\rho , \tau , \tau' \in [0,1]$, $\eps \in (0,1]$ and $m , m_L , m_R , m' , m'_L , m'_R \in \R$. Moreover, assume that the function 
		\begin{align*}
			[0,1] \ni \tau' \mapsto G_{\tau'} \in \Cont^{\infty}_{\mathrm{b}} \bigl ( \R^d_{x_L} \times \R^d_{x_R} \, , \, \Cont^{\infty}_{\mathrm{pol}}(\R_{y_L}^d \times \R_{y_R}^d \times \R_{z_L}^d \times \R_{z_R}^d) \bigr )
		\end{align*}
		depends on $\tau'$ in a continuous fashion and that for all 
		$h_L , h_R , r_L , r_R , s_L , s_R \in \N_0^d$, there exists $C_{h_L h_R r_L r_R s_L s_R} > 0$ such that
		\begin{align} 
			&\Babs{\partial_{x_L}^{h_L} \partial_{x_R}^{h_R} \partial_{y_L}^{r_L} \partial_{y_R}^{r_R} \partial_{z_L}^{s_L} \partial_{z_R}^{s_R} G_{\tau'}(x_L,x_R,y_L,y_R,z_L,z_R)} \notag \\
			&\qquad \qquad 
			\leq C_{h_L h_R r_L r_R s_L s_R} \, \bigl ( \sexpval{y_L} + \sexpval{y_R} + \sexpval{z_L} + \sexpval{z_R} \bigr )^{\sabs{h_L} + \sabs{h_R} + \sabs{r_L} + \sabs{r_R} + \sabs{s_L} + \sabs{s_R}} 
			.
			\label{appendix:oscillatory_integrals:eqn:super_Weyl_product_oscillatory_integral-flux-factor-estimate}
		\end{align}
		For two functions $F$ and $G$ on $\Pspace$, $a_L , a_R , b_L , b_R , \alpha_L , \alpha_R , \beta_L , \beta_R \in \N_0^d$ and $\tau,\tau'\in[0,1]$, we define the oscillatory integral as
		\begin{align} 
			I_{\tau\tau'}(\Xbf) := \frac{1}{(2\pi)^{4d}} \int_{\Pspace} \dd \Ybf \int_{\Pspace} \dd \Zbf \, &\e^{+ \ii \Sigma(\Xbf,\Ybf+\Zbf)} \, \e^{+ \ii \tau \frac{\eps}{2} \Sigma(r(\Ybf),\Zbf)} \, G_{\tau'}(x_L,x_R,y_L,y_R,z_L,z_R) 
			\label{appendix:oscillatory_integrals:eqn:super_Weyl_product_oscillatory_integral}
			\, \cdot \\
			&
			y_L^{a_L} \, \eta_L^{\alpha_L} \, y_R^{a_R} \, \eta_R^{\alpha_R} \, z_L^{b_L} \zeta_L^{\beta_L} \, z_R^{b_R} \zeta_R^{\beta_R} \, (\Fourier_\Sigma F)(\Ybf) \, (\Fourier_\Sigma G)(\Zbf) 
			.
			\notag 
		\end{align}
		Then in the following circumstances the oscillatory integral exists: 
		\begin{enumerate}[(1)]
			\item The map 
			\begin{align*}
				(F,G) \mapsto I_{\tau\tau'} &: S_{\rho,0}^m(\Pspace)\times S_{\rho,0}^{m'}(\Pspace) \longrightarrow S_{\rho,0}^{m + m' - \rho(\sabs{a_L} + \sabs{a_R} + \sabs{b_L} + \sabs{b_R})}(\Pspace)
			\end{align*}
			is a continuous bilinear map. Furthermore, given symbols $F\in S_{\rho,0}^m(\Pspace)$ and $G\in S_{\rho,0}^{m'}(\Pspace)$, we obtain a continuous map $[0,1] \times [0,1]\ni (\tau,\tau') \mapsto I_{\tau\tau'}\in S_{\rho,0}^{m + m' - \rho(\sabs{a_L} + \sabs{a_R} + \sabs{b_L} + \sabs{b_R})}(\Pspace)$.
			\item The map 
			\begin{align*}
				(F,G) \mapsto I_{\tau\tau'} &: S_{\rho,0}^{m_L,m_R}(\Pspace) \times S_{\rho,0}^{m_L',m_R'}(\Pspace) \longrightarrow S_{\rho,0}^{m_L + m_L' - \rho(\abs{a_L}+\abs{b_L}) , m_R + m_R' - \rho(\abs{a_R}+\abs{b_R})}(\Pspace)
			\end{align*}
			is a continuous bilinear map. Furthermore, given symbols $F\in S_{\rho,0}^{m_L,m_R}(\Pspace)$ and $G\in S_{\rho,0}^{m_L',m_R'}(\Pspace)$, the map $(\tau,\tau') \mapsto I_{\tau\tau'}$ gives rise to a continuous map from $[0,1] \times [0,1]$ to $S_{\rho,0}^{m_L + m_L' - \rho(\abs{a_L}+\abs{b_L}) , m_R + m_R' - \rho(\abs{a_R}+\abs{b_R})}(\Pspace)$.
		\end{enumerate}
	\end{lemma}
	\begin{proof}
		\begin{enumerate}[(1)]
			\item Since the proof will be largely identical to that of Lemma~\ref{appendix:oscillatory_integrals:lem:semi_super_product_existence_oscillatory_integral} we will skip some of the details; we will re-use all of the notation. 
			
			First of all, it suffices to consider the case $a_L = a_R = b_L = b_R = \alpha_L = \alpha_R = \beta_L = \beta_R = 0$: by repeated partial integration, we can convert the monomial in $\Ybf$ and $\Zbf$ into derivatives acting on $F$ and $G$. But those are just Hömander symbols of the same ($\rho = 0$) or lower ($\rho > 0$) order. 
			
			To show the existence of the oscillatory integrals and control the relevant Hörmander seminorms at the same time, we compute its derivatives $\partial_{x_L}^{n_L} \partial_{\xi_L}^{\nu_L} \partial_{x_R}^{n_R} \partial_{\xi_R}^{\nu_R} I_{\tau\tau'}$ for arbitrary values of $n_L , n_R , \nu_L , \nu_R \in \N_0^d$. This give rise to a finite sum of terms of the type
			\begin{align*}
				\frac{1}{(2\pi)^{4d}} \int_{\Pspace} \dd \Ybf \int_{\Pspace} \dd \Zbf
				\, & \, \e^{+ \ii (y_L \cdot \eta_L + z_L \cdot \zeta_L + y_R \cdot \eta_R + z_R \cdot \zeta_R)} \, \partial_{x_L}^{c_L} \partial_{x_R}^{c_R} G_{\tau'}(x_L,x_R,y_L,y_R,z_L,z_R)
				\, \cdot \\
				&
				\partial_{x_L}^{a_L} \partial_{\xi_L}^{\alpha_L} \partial_{x_R}^{a_R} \partial_{\xi_R}^{\alpha_R} F \bigl ( x_L - \tfrac{\tau\eps}{2} z_L , \xi_L - \eta_L , x_R + \tfrac{\tau\eps}{2} z_R , \xi_R - \eta_R \bigr )
				\, \cdot \\
				&
				\partial_{x_L}^{b_L} \partial_{\xi_L}^{\beta_L} \partial_{x_R}^{b_R} \partial_{\xi_R}^{\beta_R} G \bigl ( x_L + \tfrac{\tau\eps}{2} y_L , \xi_L - \zeta_L , x_R - \tfrac{\tau\eps}{2} y_R , \xi_R - \zeta_R \bigr )
				.
			\end{align*}
			To prove the existence of this oscillatory integral, we will use the same family of $L$ operators~\eqref{appendix:oscillatory_integrals:eqn:L_operator} in all the relevant variables, which, after repeated partial integration, results in a finite sum of terms of the type 
			\begin{align} 
				\int_{\Pspace} \dd \Ybf & \int_{\Pspace} \dd \Zbf \, \e^{+ \ii (y_L \cdot \eta_L + z_L \cdot \zeta_L + y_R \cdot \eta_R + z_R \cdot \zeta_R)}
				\, \cdot 
				\notag \\
				&
				\sexpval{y_L}^{-2N_L} \, \sexpval{y_R}^{-2N_R} \, \sexpval{z_L}^{-2M_L} \, \sexpval{z_R}^{-2M_R} \, \sexpval{\eta_L}^{-2D_L} \, \sexpval{\eta_R}^{-2D_R} \, \sexpval{\zeta_L}^{-2K_L} \, \sexpval{\zeta_R}^{-2K_R} 
				\, \cdot 
				\notag \\
				&
				\varphi_{N_L b'_L}(y_L) \, \varphi_{N_R b'_R}(y_R) \, \varphi_{M_L a'_L}(z_L) \, \varphi_{M_R a'_R}(z_R) 
				\, \cdot 
				\notag \\
				&
				\partial_{x_L}^{c_L} \partial_{x_R}^{c_R} \partial_{y_L}^{a''_L} \partial_{y_R}^{a''_R} \partial_{z_L}^{b''_L} \partial_{z_R}^{b''_R} G_{\tau'}(x_L,x_R,y_L,y_R,z_L,z_R)
				\, \cdot 
				\notag \\
				&
				\partial_{x_L}^{a_L + a'''_L} \partial_{\xi_L}^{\alpha_L + \alpha'_L} \partial_{x_R}^{a_R + a'''_R} \partial_{\xi_R}^{\alpha_R + \alpha'_R} F \bigl ( x_L - \tfrac{\tau\eps}{2} z_L , \xi_L - \eta_L , x_R + \tfrac{\tau\eps}{2} z_R , \xi_R - \eta_R \bigr )
				\, \cdot 
				\notag \\
				&
				\partial_{x_L}^{b_L + b'''_L} \partial_{\xi_L}^{\beta_L + \beta'_L} \partial_{x_R}^{b_R + b'''_R} \partial_{\xi_R}^{\beta_R + \beta'_R} G \bigl ( x_L + \tfrac{\tau\eps}{2} y_L , \xi_L - \zeta_L , x_R - \tfrac{\tau\eps}{2} y_R , \xi_R - \zeta_R \bigr ) 
				.
				\label{appendix:oscillatory_integrals:eqn:super_Weyl_product_oscillatory_integral_terms_of_type}
			\end{align}
			As before, the bounded, smooth functions $\varphi_{N a}(y)$ arise from taking the $a \in \N_0^d$ derivative of $\sexpval{y}^{-N}$. 
			
			Lower bounds on the values of the non-negative integers $N_L$, $N_R$, $M_L$, $M_R$, $D_L$, $D_R$, $K_L$ and $K_R$ will be given below. Importantly, the sum is over all primed multi indices satisfying~\eqref{appendix:oscillatory_integrals:eqn:bound_for_multi_indices_N_L_and_M_L}--\eqref{appendix:oscillatory_integrals:eqn:bound_for_multi_indices_D_R_and_K_R}. Our estimates for the seminorms further involve the $(n_L,\nu_L,n_R,\nu_R)$-dependent constant $N$ given by~\eqref{appendix:oscillatory_integrals:eqn:estimate_oscillatory_integral_super_Weyl_product_constant_N}. Up to a common constant that is uniform in $\tau$ and $\tau'$, the integrand of \eqref{appendix:oscillatory_integrals:eqn:super_Weyl_product_oscillatory_integral_terms_of_type} can be bounded in absolute value in terms of the standard Hörmander seminorms $p_N^m$ of $F$ and $G$, as well as constants $C , C' > 0$ that can be chosen uniformly in the multiindices and $\tau$ and $\tau'$: 
			\begin{align*}
				&\sexpval{y_L}^{-2N_L} \, \sexpval{y_R}^{-2N_R} \, \sexpval{z_L}^{-2M_L} \, \sexpval{z_R}^{-2M_R} \, \sexpval{\eta_L}^{-2D_L} \, \sexpval{\eta_R}^{-2D_R} \, \sexpval{\zeta_L}^{-2K_L} \, \sexpval{\zeta_R}^{-2K_R} 
				\, \cdot \\
				&\quad \quad
				\Babs{\varphi_{N_L b'_L}(y_L) \, \varphi_{N_R b'_R}(y_R) \, \varphi_{M_L a'_L}(z_L) \, \varphi_{M_R a'_R}(z_R)}
				\, \cdot \\
				&\quad \quad
				\Babs{\partial_{x_L}^{c_L} \partial_{x_R}^{c_R} \partial_{y_L}^{a''_L} \partial_{y_R}^{a''_R} \partial_{z_L}^{b''_L} \partial_{z_R}^{b''_R} G_{\tau'}(x_L,x_R,y_L,y_R,z_L,z_R)}
				\, \cdot \\
				&\quad \quad
				\Babs{\partial_{x_L}^{a_L + a'''_L} \partial_{\xi_L}^{\alpha_L + \alpha'_L} \partial_{x_R}^{a_R + a'''_R} \partial_{\xi_R}^{\alpha_R + \alpha'_R} F \bigl ( x_L - \tfrac{\tau\eps}{2} z_L , \xi_L-\eta_L , x_R + \tfrac{\tau\eps}{2}z_R , \xi_R-\eta_R \bigr )}
				\, \cdot \\
				&\quad \quad
				\Babs{\partial_{x_L}^{b_L + b'''_L} \partial_{\xi_L}^{\beta_L + \beta'_L} \partial_{x_R}^{b_R + b'''_R} \partial_{\xi_R}^{\beta_R + \beta'_R} G \bigl ( x_L + \tfrac{\tau\eps}{2} y_L , \xi_L - \zeta_L , x_R - \tfrac{\tau\eps}{2} y_R , \xi_R - \zeta_R \bigr )} 
				\\
				&\quad \leq 
				C \, p_N^m(F) \, p_N^{m'}(G) \, \sexpval{y_L}^{-2N_L + \sabs{c_L} + \sabs{c_R} + \sabs{a''_L} + \sabs{a''_R} + \sabs{b''_L} + \sabs{b''_R}} \, \sexpval{z_L}^{-2M_L + \sabs{c_L} + \sabs{c_R} + \sabs{a''_L} + \sabs{a''_R} + \sabs{b''_L} + \sabs{b''_R}}
				\, \cdot \\
				&\quad \quad
				\sexpval{y_R}^{-2N_R + \sabs{c_L} + \sabs{c_R} + \sabs{a''_L} + \sabs{a''_R} + \sabs{b''_L} + \sabs{b''_R}} \, \sexpval{z_R}^{-2M_R + \sabs{c_L} + \sabs{c_R} + \sabs{a''_L} + \sabs{a''_R} + \sabs{b''_L} + \sabs{b''_R}}
				\, \sexpval{\eta_L}^{-2D_L} \, \sexpval{\zeta_L}^{-2K_L} \, \cdot \\
				&\quad \quad
				\sexpval{\eta_R}^{-2D_R} \, \sexpval{\zeta_R}^{-2K_R} 
				\, \sexpval{(\xi_L-\eta_L , \xi_R-\eta_R)}^{m - \rho(\sabs{\alpha_L} + \sabs{\alpha'_L} + \sabs{\alpha_R} + \sabs{\alpha'_R})} \, \cdot \\
				&\quad \quad
				\sexpval{(\xi_L-\zeta_L , \xi_R-\zeta_R)}^{m' - \rho(\sabs{\beta_L} + \sabs{\beta'_L} + \sabs{\beta_R} + \sabs{\beta'_R})}
				\\
				&\quad \leq C' \, p_N^m(F) \, p_N^{m'}(G) \, \sexpval{(\xi_L , \xi_R)}^{m+m'-\rho(\sabs{\nu_L}+\sabs{\nu_R})} \, 
				\, \cdot \\
				&\quad \quad
				\sexpval{y_L}^{-2N_L + \sabs{n_L} + \sabs{n_R} + 2(D_L+D_R+K_L+K_R)}
				\sexpval{y_R}^{-2N_R + \sabs{n_L} + \sabs{n_R} + 2(D_L+D_R+K_L+K_R)} \, 				\, \cdot \\
				&\quad \quad
				\sexpval{z_L}^{-2M_L + \sabs{n_L} + \sabs{n_R} + 2(D_L+D_R+K_L+K_R)} \,
				\sexpval{z_R}^{-2M_R + \sabs{n_L} + \sabs{n_R} + 2(D_L+D_R+K_L+K_R)}
				\, \cdot \\
				&\quad \quad
				\sexpval{\eta_L}^{-2D_L + \sabs{m} + \rho(\sabs{\nu_L} + \sabs{\nu_R})} \, \sexpval{\eta_R}^{-2D_R + \sabs{m} + \rho(\sabs{\nu_L} + \sabs{\nu_R})} 
				\, \cdot \\
				&\quad \quad
				\sexpval{\zeta_L}^{-2K_L + \sabs{m'} + \rho(\sabs{\nu_L} + \sabs{\nu_R})} \, \sexpval{\zeta_R}^{-2K_R + \sabs{m'} + \rho(\sabs{\nu_L} + \sabs{\nu_R})} 
				.
			\end{align*}
			If we first fix $D_L,D_R,K_L$ and $K_R$ such that
			\begin{align*}
				D_L , D_R &> \tfrac{1}{2} \bigl ( d + \sabs{m} + \rho (\sabs{\nu_L} + \sabs{\nu_R}) \bigr )
				,
				\\
				K_L , K_R &> \tfrac{1}{2} \bigl ( d + \sabs{m'} + \rho (\sabs{\nu_L} + \sabs{\nu_R}) \bigr )
				,
			\end{align*}
			then we can choose the other four constants so that they all satisfy 
			\begin{align*}
				N_L , N_R , M_L , M_R > \tfrac{1}{2} (d + \abs{n_L} + \sabs{n_R}) + (D_L + D_R + K_L + K_R) 
				. 
			\end{align*}
			Given that all the sums are finite, this shows that there is a constant $\tilde{C} > 0$ such that
			\begin{align*}
				\sexpval{(\xi_L , \xi_R)}^{-(m+m')+\rho(\abs{\nu_L}+\abs{\nu_R})} \, \babs{\partial_{x_L}^{n_L} \partial_{\xi_L}^{\nu_L} \partial_{x_R}^{n_R} \partial_{\xi_R}^{\nu_R} I_{\tau\tau'}(\Xbf)} \leq \tilde{C} \, p_N^m(F) \, p_N^{m'}(G) 
			\end{align*}
			holds for all $\Xbf \in \Pspace$. 
			
			When we restore the monomials in the oscillatory integral~\eqref{appendix:oscillatory_integrals:eqn:super_Weyl_product_oscillatory_integral} that we eliminated in the first step of the proof, we deduce that the oscillatory integral $I_{\tau\tau'} \in S_{\rho,0}^{m+m'-\rho(\abs{a_L}+\abs{a_R}+\abs{b_L}+\abs{b_R})}(\Pspace)$ exists as a Hörmander symbol, where we have restored the indices $a_L$, $a_R$, $b_L$ and $b_R$ from the polynomial prefactor in equation~\eqref{appendix:oscillatory_integrals:eqn:super_Weyl_product_oscillatory_integral}. Furthermore, the map $(F,G) \mapsto I_{\tau\tau'}$ gives rise to a continuous bilinear map from $S_{\rho,0}^m(\Pspace) \times S_{\rho,0}^{m'}(\Pspace)$ to $S_{\rho,0}^{m+m'-\rho(\abs{a_L}+\abs{a_R}+\abs{b_L}+\abs{b_R})}(\Pspace)$. 
			
			Analogous to the proof of Lemma~\ref{appendix:oscillatory_integrals:lem:semi_super_product_existence_oscillatory_integral}, by using the uniformness in $\tau$ and $\tau'$ of the bounds in the last inequality, we deduce that the map 
			\begin{align*}
				[0,1] \times [0,1] \ni (\tau,\tau') \mapsto I_{\tau\tau'} \in S_{\rho,0}^{m+m'-\rho(\abs{a_L}+\abs{a_R}+\abs{b_L}+\abs{b_R})}(\Pspace) 
			\end{align*}
			is continuous. This proves~(1).
			\item The proof for super Hörmander symbols is virtually identical, we just need to make separate estimates for left and right variables and use the super Hörmander class seminorms~\eqref{symbol_super_calculus:eqn:max_seminorms_super_Hoermander_symbols}.
			
			All steps are identical to the proof of (1), save, of course, for how we estimate the integrand of~\eqref{appendix:oscillatory_integrals:eqn:super_Weyl_product_oscillatory_integral_terms_of_type}. So we will skip right to the crucial estimate. In this case, the integrand of~\eqref{appendix:oscillatory_integrals:eqn:super_Weyl_product_oscillatory_integral_terms_of_type} can be estimated by
			\begin{align*}
			    &
				C \; q_N^{m_L,m_R}(F) \; q_N^{m_L',m_R'}(G)  
				\, \cdot \\
				&\qquad \qquad
				\sexpval{y_L}^{-2N_L + \sabs{c_L} + \sabs{c_R} + \sabs{a''_L} + \sabs{a''_R} + \sabs{b''_L} + \sabs{b''_R}} \, \sexpval{y_R}^{-2N_R + \sabs{c_L} + \sabs{c_R} + \sabs{a''_L} + \sabs{a''_R} + \sabs{b''_L} + \sabs{b''_R}}
				\, \cdot \\
				&\qquad \qquad
				\sexpval{z_L}^{-2M_L + \sabs{c_L} + \sabs{c_R} + \sabs{a''_L} + \sabs{a''_R} + \sabs{b''_L} + \sabs{b''_R}} \, \sexpval{z_R}^{-2M_R + \sabs{c_L} + \sabs{c_R} + \sabs{a''_L} + \sabs{a''_R} + \sabs{b''_L} + \sabs{b''_R}}
				\, \cdot \\
				&\qquad \qquad
				\sexpval{\eta_L}^{-2D_L} \, \sexpval{\eta_R}^{-2D_R} \, \sexpval{\zeta_L}^{-2K_L} \, \sexpval{\zeta_R}^{-2K_R} \, 
				\, \cdot \\
				&\qquad \qquad
				\sexpval{\xi_L-\eta_L}^{m_L-\rho(\sabs{\alpha_L}+\sabs{\alpha'_L})} \, \sexpval{\xi_L-\zeta_L}^{m_L'-\rho(\sabs{\beta_L}+\sabs{\beta'_L})}
				\, \cdot \\
				&\qquad \qquad
				\sexpval{\xi_R-\eta_R}^{m_R-\rho(\sabs{\alpha_R}+\sabs{\alpha'_R})} \, \sexpval{\xi_R-\zeta_R}^{m_R'-\rho(\sabs{\beta_R}+\sabs{\beta'_R})}
				\\
				&\qquad \leq 
				C' \; q_N^{m_L,m_R}(F) \; q_N^{m_L',m_R'}(G) \; \sexpval{\xi_L}^{m_L+m_L'-\rho\sabs{\nu_L}} \; \sexpval{\xi_R}^{m_R+m_R'-\rho\sabs{\nu_R}} 
				\, \cdot \\
				&\qquad \qquad
				\sexpval{y_L}^{-2N_L+\sabs{n_L}+\sabs{n_R}+2(D_L+D_R+K_L+K_R)}
				\sexpval{y_R}^{-2N_R+\sabs{n_L}+\sabs{n_R}+2(D_L+D_R+K_L+K_R)} \, 
				\, \cdot \\
				&\qquad \qquad
				\sexpval{z_L}^{-2M_L+\sabs{n_L}+\sabs{n_R}+2(D_L+D_R+K_L+K_R)} \,
				\sexpval{z_R}^{-2M_R+\sabs{n_L}+\sabs{n_R}+2(D_L+D_R+K_L+K_R)}
				\, \cdot \\
				&\qquad \qquad
				\sexpval{\eta_L}^{-2D_L+\sabs{m_L}+\rho\sabs{\nu_L}} \,  \sexpval{\eta_R}^{-2D_R+\sabs{m_R}+\rho\sabs{\nu_R}} \, \sexpval{\zeta_L}^{-2K_L+\sabs{m_L'}+\rho\sabs{\nu_L}} \,\sexpval{\zeta_R}^{-2K_R+\sabs{m_R'}+\rho\sabs{\nu_R}} 
				.
			\end{align*}
			Here we have re-used the constant $N$ from the proof of~(1). To ensure integrability, we choose $D_L,D_R,K_L$ and $K_R$ such that
			\begin{align*}
				D_{L,R} > \tfrac{1}{2} (d + \sabs{m_{L,R}} + \rho \sabs{\nu_{L,R}})
				, \qquad
				K_{L,R} > \tfrac{1}{2} (d + \sabs{m'_{L,R}} + \rho \sabs{\nu_{L,R}})
				.
			\end{align*}
			For these fixed $D_L$, $D_R$, $K_L$ and $K_R$, we impose the same lower bounds on $N_L$, $N_R$, $M_L$ and $M_R$ as in the proof of~(1). These conditions imposed on $D_L$, $D_R$, $K_L$, $K_R$, $N_L$, $N_R$, $M_L$ and $M_R$ ensures the integrability of the right-hand side of the last inequality in $y_L$, $y_R$, $z_L$, $z_R$, $\eta_L$, $\eta_R$, $\zeta_L$ and $\zeta_R$.

			All this shows that there is a constant $\tilde{C}>0$ such that
			\begin{align*}
				\sexpval{\xi_L}^{-(m_L+m_L')+\rho\abs{\nu_L}} \, \sexpval{\xi_R}^{-(m_R+m_R')+\rho\abs{\nu_R}} \, &\babs{\partial_{x_L}^{n_L} \partial_{\xi_L}^{\nu_L} \partial_{x_R}^{n_R} \partial_{\xi_R}^{\nu_R} I_{\tau\tau'}(\Xbf)} \leq \tilde{C} \, q_N^{m_L,m_R}(F) \; q_N^{m_L',m_R'}(G)
			\end{align*}
			holds for all $\Xbf \in \Pspace$. Consequently, we see that the map $(F,G)\mapsto I_{\tau\tau'}$ gives rise to a continuous bilinear map from $S_{\rho,0}^{m_L,m_R}(\Pspace) \times S_{\rho,0}^{m_L',m_R'}(\Pspace)$ to $S_{\rho,0}^{m_L+m_L',m_R+m_R'}(\Pspace)$.
			
			When we restore the monomials in the oscillatory integral~\eqref{appendix:oscillatory_integrals:eqn:super_Weyl_product_oscillatory_integral} that we eliminated in the first step of the proof, we deduce as before that 
			\begin{align*}
				(F,G) \mapsto I_{\tau\tau'} : S_{\rho,0}^{m_L,m_R}(\Pspace) \times S_{\rho,0}^{m_L,m_R}(\Pspace) \longrightarrow S_{\rho,0}^{m_L + m_L' - \rho(\sabs{a_L} + \sabs{b_L}),m_R + m_R' - \rho(\sabs{a_R} + \sabs{b_R})}(\Pspace)
			\end{align*}
			is a continuous bilinear map. Similarly, we conclude that $(\tau,\tau') \mapsto I_{\tau\tau'}$ is continuous. This proves~(2) and the proof is complete.
		\end{enumerate}
	\end{proof}
	%
\end{appendix}
%

\bibliographystyle{alpha}
\bibliography{bibliography}

\begin{thebibliography}{BvESB94}

\bibitem[AMP10]{Mantoiu_Purice:continuity_spectra:2009}
Nassim Athmouni, Marius Măntoiu, and Radu Purice.
\newblock {On the continuity of spectra for families of magnetic
  pseudodifferential operators}.
\newblock {\em {J. Math. Phys.}}, 51:083517, 2010.

\bibitem[Baa88a]{Baaj:twisted_X_products_1:1988}
S.~Baaj.
\newblock {Calcul pseudo-différentiel et produits croisés de C*-algèbres.
  I}.
\newblock {\em Comptes rendus de l'Académie des sciences. Série 1,
  Mathématique}, 307(11):581–586, 1988.

\bibitem[Baa88b]{Baaj:twisted_X_products_2:1988}
S.~Baaj.
\newblock {Calcul pseudo-différentiel et produits croisés de C*-algèbres.
  II}.
\newblock {\em Comptes rendus de l'Académie des sciences. Série 1,
  Mathématique}, 307(12):663–666, 1988.

\bibitem[Bea77]{Beals:characterization_psido:1977}
R.~Beals.
\newblock {Characterization of Pseudodifferential Operators}.
\newblock {\em Duke Mathematics Journal}, 44:45--57, 1977.

\bibitem[Bel84]{Bellissard:quantum_Hall_effect_noncommutative:1984}
Jean Bellissard.
\newblock {\em {Ordinary Quantum Hall Effect and Noncommutative Cohomology}},
  volume~3 of {\em Teubner Physics Texts}, page 61–74.
\newblock Vch Pub, 1984.

\bibitem[Bel86]{Bellissard:K_theory_Cast_algebras_solid_state_physics:1986}
Jean Bellissard.
\newblock {\em {$K$-Theory of $C^*$-algebras in solid state physics}}, volume
  257 of {\em Lecture Notes in Physics}, page 99–156.
\newblock Springer-Verlag, 1986.

\bibitem[Bel88]{Bellissard:Cstar_algebras_solid_state_physics:1988}
J.~Bellissard.
\newblock {\em {$C^*$-Algebras in Solid State Physics: 2D electrons in a
  uniform magnetic field}}, volume 136 of {\em London Mathematical Society
  Lecture Note Series}, page 49–76.
\newblock Cambridge University Press, Cambridge, 1988.

\bibitem[BGKS05]{Bouclet_Germinet_Klein_Schenker:linear_response_theory_magnetic_Schroedinger_operators_disorder:2005}
Jean-Marc Bouclet, François Germinet, Abel Klein, and Jeffrey~H. Schenker.
\newblock {Linear response theory for magnetic Schrödinger operators in
  disordered media}.
\newblock {\em J. Func. Anal.}, 226:301--372, 2005.

\bibitem[BLM13]{Belmonte_Lein_Mantoiu:mag_twisted_actions:2010}
Fabian Belmonte, Max Lein, and Marius Măntoiu.
\newblock {Magnetic twisted actions on general abelian $C^*$-algebras}.
\newblock {\em Journal of Operator Theory}, 69:33–58, 2013.

\bibitem[Bon97]{Bony:characterization_psido:1996}
J.~M. Bony.
\newblock {Caractérisacion des opérateurs pseudo-différentiels}.
\newblock {\em École Polytechnique, Séminaire EDP}, XXIII, 1996-1997.

\bibitem[BR02]{Bratteli_Robinson:operator_algebras_1:2002}
Ola Bratteli and Derek~W. Robinson.
\newblock {\em {Operator Algebras and Quantum Statistical Mechanics 1}}.
\newblock {Theoretical and Mathematical Physics}. Springer-Verlag, 2002.

\bibitem[BR03]{Bratteli_Robinson:operator_algebras_2:2003}
Ola Bratteli and Derek~W. Robinson.
\newblock {\em {Operator Algebras and Quantum Statistical Mechanics 2}}.
\newblock {Theoretical and Mathematical Physics}. Springer-Verlag, 2003.

\bibitem[BS98]{Bellissard_Schulz_Baldes:kinetic_theory_quantum_transport_aperiodic_media:1998}
Jean Bellissard and Hermann {Schulz-Baldes}.
\newblock {A Kinetic Theory for Quantum Transport in Aperiodic Media}.
\newblock {\em J. Stat. Phys.}, 91:991–1027, 1998.

\bibitem[BvESB94]{Bellissard_van_Elst_Schulz_Baldes:noncommutative_geometry_quantum_hall_effect:1994}
Jean Bellissard, A.~van Elst, and Hermann Schulz-Baldes.
\newblock {The noncommutative geometry of the quantum Hall effect}.
\newblock {\em J. Math. Phys.}, 35(10):5373–5451, 10 1994.

\bibitem[CHP18]{Cornean_Helffer_Purice:simple_proof_Beals_criterion_magnetic_PsiDOs:2018}
Horia Cornean, Bernard Helffer, and Radu Purice.
\newblock {A Beals Criterion for magnetic pseudo-differential operators proved
  with magnetic Gabor frames}.
\newblock {\em Comm. PDE}, 43(8):1196–1204, 2018.

\bibitem[Con80]{Connes:Cstar_algebras_and_differential_geometry:1980}
A.~Connes.
\newblock {$C^*$-algèbres et geomètrie differentielle}.
\newblock {\em C. R. Acad. Sci. Paris Sèrie A}, 290(599–604), 1980.

\bibitem[Con94]{Connes:noncommutative_geometry:1994}
Alain Connes.
\newblock {\em {Noncommutative Geometry}}.
\newblock Academic Press, 1994.

\bibitem[DG08]{Dombrowski_Germinet:linear_response_theory:2008}
Nicolas Dombrowski and François Germinet.
\newblock {Linear response theory for random Schrödinger operators and
  noncommutative geometry}.
\newblock {\em Mark. Proc. Rel. Fields}, 14(3):403--426, 2008.

\bibitem[DH73]{Hoermander:Fourier_integral_operators_2:1973}
J.~J. Duistermaat and Lars Hörmander.
\newblock {Fourier Integral Operators 2}.
\newblock {\em Acta Mathematica}, 128(1):183–269, 1973.

\bibitem[Dir47]{Dirac:foundations_qm_en:1947}
P.~A.~M. Dirac.
\newblock {\em {The Principles of Quantum Mechanics}}.
\newblock Oxford University Press, third edition, 1947.

\bibitem[Dix77]{Dixmier:C_star_algebras:1977}
Jacques Dixmier.
\newblock {\em {$C^*$-Algebras}}, volume~15.
\newblock North-Holland, Amsterdam, 1977.

\bibitem[Dix81]{Dixmier:von_Neumann_algebras:1981}
Jacques Dixmier.
\newblock {\em {Von Neumann Algebras}}, volume~27.
\newblock North-Holland, Amsterdam, 1981.

\bibitem[DL11]{DeNittis_Lein:Bloch_electron:2009}
Giuseppe {De Nittis} and Max Lein.
\newblock {Applications of Magnetic $\Psi$DO Techniques to SAPT – Beyond a
  simple review}.
\newblock {\em Rev. Math. Phys.}, 23:233--260, 2011.

\bibitem[DL17]{DeNittis_Lein:linear_response_theory:2017}
Giuseppe {De Nittis} and Max Lein.
\newblock {\em {Linear Response Theory}}, volume~21 of {\em Springer Briefs in
  Mathematical Physics}.
\newblock Springer, 2017.

\bibitem[FL13]{Fuerst_Lein:scaling_limits_Dirac:2008}
Martin Fürst and Max Lein.
\newblock {Semi- and Non-relativistic Limit of the Dirac Dynamics with External
  Fields}.
\newblock {\em Annales Henri Poincaré}, 14:1305–1347, 2013.

\bibitem[Fol89]{Folland:harmonic_analysis_hase_space:1989}
Gerald~B. Folland.
\newblock {\em {Harmonic Analysis on Phase Space}}, volume 122 of {\em Annals
  of Mathematics Studies}.
\newblock Princeton University Press, 1989.

\bibitem[GK69]{Gohberg_Krein:linear_nonselfadjoint_operators:1969}
Israel~C. Gohberg and Mark~Grigorievich Krein.
\newblock {\em {Introduction to the Theory of Linear Nonselfadjoint Operators
  in Hilbert Space}}, volume~18 of {\em Translations of Mathematical
  Monographs}.
\newblock American Mathematical Society, 1969.

\bibitem[GLT13]{Gat_Lein_Teufel:resonance_phenomena_wavepacket_qubit:2013}
Omri Gat, Max Lein, and Stefan Teufel.
\newblock {Resonance phenomena in the interaction of a many-photon wave packet
  and a qubit}.
\newblock {\em J. Phys. A}, 46(31):315301, 2013.

\bibitem[GLT14]{Gat_Lein_Teufel:semiclassical_dynamics_particle_spin:2013}
Omri Gat, Max Lein, and Stefan Teufel.
\newblock {Semiclassics for particles with spin via a Wigner-Weyl-type
  calculus}.
\newblock {\em Annales Henri Poincaré}, 15(10):1967–1991, 2014.

\bibitem[GS94]{Grigis_Sjoestrand:microlocal_analysis:1994}
A.~Grigis and Johannes Sjöstrand.
\newblock {\em {Microlocal analysis for differential operators. An
  introduction}}, volume 196 of {\em London Mathematical Society Lecture Note
  Series}.
\newblock Cambridge University Press, 1994.

\bibitem[HLP18a]{Ha_Lee_Ponge:pseudodifferential_theory_on_noncommutative_torus_1:2018}
Hyunsu Ha, Gihyun Lee, and Raphael Ponge.
\newblock {Pseudodifferential calculus on noncommutative tori, I. Oscillating
  integrals}.
\newblock {\em arXiv}, 1803.03575:1–46, 2018.

\bibitem[HLP18b]{Ha_Lee_Ponge:pseudodifferential_theory_on_noncommutative_torus_2:2018}
Hyunsu Ha, Gihyun Lee, and Raphael Ponge.
\newblock {Pseudodifferential calculus on noncommutative tori, II. Main
  properties}.
\newblock {\em arXiv}, 1803.03580:1–46, 2018.

\bibitem[HS89]{Helffer_Sjoestrand:mag_Schroedinger_equation:1989}
Bernard Helffer and Johannes Sjöstrand.
\newblock {\em {Équation de Schrödinger avec champ magnétique et équation
  de Harper}}, volume 345 of {\em Lecture Notes in Physics}, page 118–197.
\newblock Springer-Verlag, 1989.

\bibitem[Hö71]{Hoermander:Fourier_integral_operators_1:1971}
Lars Hörmander.
\newblock {Fourier Integral Operators I}.
\newblock {\em Acta Mathematica}, 127:79–183, 1971.

\bibitem[Hö79]{Hoermander:Weyl_calculus:1979}
Lars Hörmander.
\newblock {The Weyl Calculus of Pseudo-Differential Operators}.
\newblock {\em Communications on Pure and Applied Mathematics},
  XXXII:359–443, 1979.

\bibitem[IMP07]{Iftimie_Mantoiu_Purice:magnetic_psido:2006}
Viorel Iftimie, Marius Măntoiu, and Radu Purice.
\newblock {Magnetic Pseudodifferential Operators}.
\newblock {\em Publications of the Research Institute for Mathematical
  Sciences}, 43:585–623, 2007.

\bibitem[IMP10]{Iftimie_Mantoiu_Purice:commutator_criteria:2008}
Viorel Iftimie, Marius Măntoiu, and Radu Purice.
\newblock {Commutator Criteria for Magnetic Pseudodifferential Operators}.
\newblock {\em Communications in Partial Differential Equations},
  35:1058–1094, 2010.

\bibitem[IP11]{Iftimie_Purice:magnetic_Fourier_integral_operators:2011}
Viorel Iftimie and Radu Purice.
\newblock {Magnetic Fourier integral operators}.
\newblock {\em J. Pseudo-Differ. Oper. Appl.}, 2:141–218, 2011.

\bibitem[Kg75]{Kumanogo:PsiDOs_multiple_symbols_L2_boundedness:1975}
Hitoshi Kumano-go.
\newblock {Pseudo-differential operators of multiple symbol and the
  Calderón-Vaillancourt theorem}.
\newblock {\em J. Math. Soc. Japan}, 27(1):113–120, 1975.

\bibitem[Kg81]{Kumanogo:pseudodiff:1981}
Hitoshi Kumano-go.
\newblock {\em {Pseudodifferential Operators}}.
\newblock The MIT Press, 1981.

\bibitem[Lei10a]{Lein:two_parameter_asymptotics:2008}
Max Lein.
\newblock {Two-parameter Asymptotics in Magnetic Weyl Calculus}.
\newblock {\em J. Math. Phys.}, 51:123519, 2010.

\bibitem[Lei10b]{Lein:quantization_semiclassics:2010}
Max Lein.
\newblock {Weyl Quantization and Semiclassics}, 2010.
\newblock Technische Universität München.

\bibitem[Lei11]{Lein:progress_magWQ:2010}
Max Lein.
\newblock {\em {Semiclassical Dynamics and Magnetic Weyl Calculus}}.
\newblock Phd thesis, Technische Universität München, Munich, Germany, 2011.

\bibitem[Len99]{Lenz:random_operators_crossed_products:1999}
Daniel~H. Lenz.
\newblock {Random Operators and Crossed Products}.
\newblock {\em Mathematical Physics, Analysis and Geometry}, 2(2):197–220,
  1999.

\bibitem[Lin76]{Lindblad:generator_semigroups:1976}
G.~Lindblad.
\newblock On the generators of quantum dynamical semigroups.
\newblock {\em Communications in Mathematical Physics}, 48(2):119--130, 1976.

\bibitem[LMR10]{Lein_Mantoiu_Richard:anisotropic_mag_pseudo:2009}
Max Lein, Marius Măntoiu, and Serge Richard.
\newblock {Magnetic pseudodifferential operators with coefficients in
  $C^*$-algebras}.
\newblock {\em Publ. RIMS Kyoto Univ.}, 46:755–788, 2010.

\bibitem[Mar02]{Martinez:intro_microlocal_analysis:2002}
Andre Martinez.
\newblock {\em {An Introduction to Semiclassical and Microlocal Analysis}}.
\newblock Springer-Verlag, 2002.

\bibitem[Moy49]{Moyal:Weyl_calculus:1949}
José~Enrique Moyal.
\newblock {Quantum mechanics as a statistical theory}.
\newblock {\em Proc. Cambridge Phil. Soc}, 45(99):211, 1949.

\bibitem[MP04]{Mantoiu_Purice:magnetic_Weyl_calculus:2004}
Marius Măntoiu and Radu Purice.
\newblock {The Magnetic Weyl Calculus}.
\newblock {\em J. Math. Phys.}, 45(4):1394–1417, 2004.

\bibitem[MPR05]{Mantoiu_Purice_Richard:twisted_X_products:2004}
Marius Măntoiu, Radu Purice, and Serge Richard.
\newblock {\em {Twisted Crossed Products and Magnetic Pseudodifferential
  Operators}}, page 137–172.
\newblock Theta, 2005.

\bibitem[MPR07]{Mantoiu_Purice_Richard:Cstar_algebraic_framework:2007}
Marius Măntoiu, Radu Purice, and Serge Richard.
\newblock {Spectral and propagation results for magnetic Schrödinger
  operators; A C*-algebraic framework}.
\newblock {\em Journal of Functional Analysis}, 250(1):42 -- 67, 2007.

\bibitem[Mü99]{Mueller:product_rule_gauge_invariant_Weyl_symbols:1999}
M.~Müller.
\newblock {Product rule for gauge invariant Weyl symbols and its application to
  the semiclassical description of guiding centre motion }.
\newblock {\em J. Phys. A}, 32:1035–1052, 1999.

\bibitem[Pro14]{Prodan:non_commutative_geometry_topological_insulators:2014}
Emil Prodan.
\newblock {The Non-Commutative Geometry of the Complex Classes of Topological
  Insulators}.
\newblock {\em Topol. Quantum Matter}, 1:1–16, 2014.

\bibitem[PST03a]{PST:effective_dynamics_Bloch:2003}
Gianluca Panati, Herbert Spohn, and Stefan Teufel.
\newblock {Effective dynamics for Bloch electrons: Peierls substitution and
  Beyond}.
\newblock {\em Commun. Math. Phys.}, 242:547–578, 10 2003.

\bibitem[PST03b]{PST:sapt:2002}
Gianluca Panati, Herbert Spohn, and Stefan Teufel.
\newblock {Space-Adiabatic Perturbation Theory}.
\newblock {\em Adv. Theor. Math. Phys.}, 7(1):145–204, 2003.

\bibitem[PST07]{PST:Born-Oppenheimer:2007}
Gianluca Panati, Herbert Spohn, and Stefan Teufel.
\newblock {The time-dependent Born-Oppenheimer approximation}.
\newblock {\em M2AN}, 41(2):297--314, March 2007.

\bibitem[Rob87]{Robert:tour_semiclassique:1987}
Didier Robert.
\newblock {\em {Autour de l'Approximation Semi-Classique}}.
\newblock Birkhäuser, 1987.

\bibitem[Rod75]{Rodino:PsiDOs_manifolds:1975}
Luigi Rodino.
\newblock {A class of pseudo-differential operators on the product of two
  manifolds and applications}.
\newblock {\em Ann. Scuola Norm. Sup. Pisa}, IV(2):287–302, 1975.

\bibitem[RS72]{Reed_Simon:M_cap_Phi_1:1972}
Michael Reed and Barry Simon.
\newblock {\em {Methods of Mathematical Physics I: Functional Analysis}},
  volume~1 of {\em Methods of Modern Mathematical Physics}.
\newblock Academic Press, 1972.

\bibitem[Str57]{Stratonovich:distributions_rep_space:1957}
R.~L. Stratonovich.
\newblock {On distributions in representation space}.
\newblock {\em Sov. Phys. JETP}, 4:891, 1957.

\bibitem[Tak03]{Takesaki:operator_algebras_2:2003}
Masamichi Takesaki.
\newblock {\em {Theory of Operator Algebras II}}, volume 125 of {\em
  Encyclopaedia of Mathematical Sciences}.
\newblock Springer-Verlag, 2003.

\bibitem[Tar08]{Tarasov:quantum_mechanics_non-hamiltonian_dissipative_systems:2008}
Vasily~E. Tarasov.
\newblock {\em {Quantum Mechanics of Non-Hamiltonian and Dissipative Systems}},
  volume~7 of {\em Monograph Series on Nonlinear Science and Complexity}.
\newblock Elsevier Science, 2008.

\bibitem[Tay81]{Taylor:PsiDO:1981}
Michael~E. Taylor.
\newblock {\em Pseudodifferential Operators}.
\newblock Princeton University Press, 1981.

\bibitem[Tre67]{Treves:topological_vector_spaces:1967}
François Treves.
\newblock {\em {Topological vector spaces, distributions and kernels}}.
\newblock Academic Press, 1967.

\bibitem[Wey27]{Weyl:qm_gruppentheorie:1927}
Hermann Weyl.
\newblock {Quantenmechanik und Gruppentheorie}.
\newblock {\em {Zeitschrift für Physik}}, 46(1-2):1–46, 1927.

\bibitem[Wig32]{Wigner:Wigner_transform:1932}
Eugene Wigner.
\newblock {On the quantum correction for thermodynamic equilibrium}.
\newblock {\em Phys. Rev.}, 40(5):749–759, 1932.

\bibitem[Wit97]{Witt:weak_topology_symbol_spaces:1997}
Ingo Witt.
\newblock {\em {The Weak Symbol Topology and Continuity of Pseudodifferential
  Operators}}, page 412–422.
\newblock Wiley-VCH Verlag GmbH, 1997.

\end{thebibliography}

\end{document}